\documentclass[a4paper, 12pt, final]{article}
\pdfoutput=1 
\usepackage{amsfonts, geometry, microtype, hyperref, amsthm, amsmath, amssymb, tikz, leftidx, cancel, mathrsfs}
\usetikzlibrary{matrix, arrows.meta}

\usepackage[section,nottoc]{tocbibind}
\usepackage[main=english, polish]{babel}

\usepackage[utf8]{inputenc}
\hypersetup{unicode, linktoc=all, colorlinks, allcolors=black}

\newtheorem*{theorem*}{Theorem}
\newtheorem*{proposition*}{Proposition}
\newtheorem*{definition*}{Definition}
\newtheorem{theorem}{Theorem}[section]
\newtheorem{axioms}[theorem]{Axioms}
\newtheorem{proposition}[theorem]{Proposition}
\newtheorem{lemma}[theorem]{Lemma}
\newtheorem{corollary}[theorem]{Corollary}
\newtheorem{definition}[theorem]{Definition}
\newtheorem{remark}[theorem]{Remark}
\newtheorem{principle}[theorem]{Principle}
\newtheorem{conjecture}[theorem]{Conjecture}
\newtheorem{problem}[theorem]{Problem}

\theoremstyle{remark}
\newtheorem*{digression}{Digression}

\title{Categories of Physical Processes}
\author{\foreignlanguage{polish}{Stanis{\l}aw Szawiel} \\ \small{Institute of Mathematics, University of Warsaw} \\
\small{ul.~Banacha 2, 00-913 Warsaw, Poland}}
\date{\today}

\begin{document}
\maketitle

\begin{abstract}
We study the mathematical foundations of physics. We reconstruct textbook quantum theory from a single symmetric monoidal functor
\begin{displaymath}
GNS : \mathbf{Phys} \longrightarrow \ast\mathbf{Mod},
\end{displaymath}
based on the Gelfand-Naimark-Segal construction and the notion of representability.

We derive the probabilistic interpretation of quantum mechanics, including the Born rule, the Schr{\"o}dinger and Heisenberg pictures, the relation between symmetries and group representations, and a theory of quantum Markov processes, including wave function collapse. Inclusion of the classical limit and deformation quantization is briefly sketched.

Gauge symmetry and extended locality cannot currently be accommodated, due to conceptual difficulties discussed in an appendix.
\end{abstract}
\tableofcontents

\section*{Introduction}
\addcontentsline{toc}{section}{Introduction}
\subsection{A Simple Idea}
Let us attempt to do physics synthetically, and postulate a category of physical processes, $\mathbf{Phys}$. It's objects are to be \emph{states}, denoted $\varphi, \psi, \ldots$, and it's morphisms are to be \emph{physical processes}, such as time evolution $U(t) : \varphi \rightarrow \psi$, or maybe even asymptotic scattering, such as pair production $\gamma + \gamma \rightarrow e^{-} + e^{+}$. Which category are we in?

Every state $\varphi$ should determine its observable quantities, $\mathcal{O}(\varphi)$, which, owing to the nature of numbers, are to form an algebra. We will use real and complex numbers, but the enterprising reader can attempt to follow along our discussion replacing $\mathbb{C}/\mathbb{R}$ with $\mathbb{Z}\lbrack i\rbrack/\mathbb{Z}$, leading to a sort of ``arithmetical physics''. This is technically challenging, and I have been unable to do so.

For the purposes of this introduction we will assume that $\mathcal{O}(\varphi)$ is a unital $C^{\ast}$-algebra, but this is not necessary, and is in fact deeply insufficient. State spaces of gauge theories determine Hopf algebroids (with exotic morphisms) instead of algebras, at least classically. In such theories the characteristic algebraic structure and the notion of observable part ways (we discuss this in depth in appendix \ref{appendixa}). For these, and other reasons I will not commit to deep technical considerations in functional analysis.

Every physical process should be accompanied by a description of what happens to any observable quantity. Thus we postulate that $\mathcal{O}$ is a functor
\begin{displaymath}
\mathcal{O} : \mathbf{Phys} \longrightarrow C^{\ast}\mathbf{Alg}^{op}.
\end{displaymath}
Note the variance. We are effectively treating algebras as noncommutative spaces.

It's well known that the world is not deterministic\footnote{At least not in any sense that is operationalizable in contemporary experiments. Ideas such as superdeterminism cannot really be tested, only pushed back.}, so we do not expect observables to have specific values in a given state, but we do require an average, or expectation value
\begin{displaymath}
\langle - \rangle_{\varphi} : \mathcal{O}(\varphi) \longrightarrow \mathbb{C},
\end{displaymath}
a linear functional, or a measure on the noncommutative phase space. Again, this is not strictly true. Even in ordinary probability theory there are random variables without an expectation value, and the massless 2d quantum scalar field appears to be a noncommutative object of this type (cf.~\cite[$\mathsection$1.5]{witten}), defining a state analogous to the Lebesgue measure -- with no probabilistic interpretation. I do not yet know how to capture such phenomena.
\begin{digression}[The ``$L^{1}$ digression'']
This is also related to the problem multiplying local operators in quantum field theory \cite[Lecture 3]{wittenpert}. If we tentatively write $\mathcal{O}(\varphi) = L^{1}(\varphi)$ -- a noncommutative $L^{1}$ space -- then the problem becomes clear: $L^{1}$ functions do not form an algebra under pointwise multiplication. If $\mathcal{O}(\varphi)$ is to be an algebra, then it must be either a lot bigger or a lot smaller than ``all the observables with an expectation value''. Our choice above means the latter, but one can lead a happy mathematical life with the former choice \cite[$\mathsection$2.5]{tao}.

This suggests that demanding expectation values from all entities in QFT is unfounded. In particular I expect all products of all operators to be definable without trickery. They will be singular objects (beyond distributions) which merely happen to have no expectation value.

The absence of free-form constructions in noncommutative geometry prevents progress in this direction. In set theory $L^{1}(\mu)$ is a discovery, a structure unknown to $\mu$ itself. In the noncommutative setting it's an intrinsic feature of $\mu$, to be given before $\mu$ has a chance to exist.
\end{digression}
When no confusion can arise we will write $\varphi$ for $\langle - \rangle_{\varphi}$.

For a process  $f : \varphi \rightarrow \psi$ we can now construct a diagram
\begin{center}
\begin{tikzpicture}
\node (a) at (0,0) {$\mathcal{O}(\varphi)$};
\node (b) at (-3,0) {$\mathcal{O}(\psi)$};
\node (c) at (-1.5,-1.5) {$\mathbb{C}$};

\path[->] (b) edge node[auto] {$\mathcal{O}(f)$} (a)
		  (a) edge node[auto] {$\varphi$} (c)
		  (b) edge node[auto, swap] {$\psi$} (c);
\end{tikzpicture}
\end{center}
Since we imagine that $\mathcal{O}(f)$ completely explains what happens to the observable quantities -- including their expectation values -- we require this diagram to commute:
\begin{displaymath}
\mathcal{O}(f)^{\ast}\varphi = \psi.
\end{displaymath}

States are not meant to be complete descriptions of the world, even if all concrete constructions (mechanics, field theory, etc.) treat them as such. Physics studies subsystems of the world, and so we need a method to build bigger systems out of smaller ones. We must investigate the idea of physically composing systems and their states. Any cosmological considerations will require further conceptual refinements, in particular making sense of the notion of \emph{self-measurement} (cf.~section \ref{interpretation}).

Unlike abstract logical objects, like terms and propositions, physical entities cannot be duplicated or deleted without effort. They appear to form a \emph{linear type system} \cite{rosetta}. Consequently we postulate that $\mathbf{Phys}$ is a symmetric monoidal category, and that the $\mathcal{O}$ functor is such as well\footnote{Which monoidal structure on $C^{\ast}$-algebras? It suffices for it to functorially extend to completely positive maps. In particular the maximal and minimal structures are both fine.}. We interpret $\varphi \otimes \psi$ as the \emph{noninteracting composite} of $\varphi$ and $\psi$. The states are put in two parallel worlds, which are identical except for the distinctiveness provided by $\varphi$ and $\psi$. This structure is an idealization of carefully putting things side by side, while screening all interactions.
\begin{digression}
Systematically treating $\mathbf{Phys}$ as a type system leads to extremely interesting philosophical considerations, allowing definitions such as ``causality is necessary linear implication'' and ``matter sources are infinitesimally close possible worlds''. Counterfactual conditionals can be given natural meanings, based in physical laws. A mathematical analysis of (in)commensurability is possible. Such ideas will be pursued elsewhere.
\end{digression}
Since $\varphi$ and $\psi$ are independent in $\varphi \otimes \psi$, we postulate a weak noncommutative independence condition
\begin{displaymath}
\langle - \rangle_{\varphi \otimes \psi} = \langle - \rangle_{\varphi} \otimes \langle - \rangle_{\psi}.
\end{displaymath}

We have already postulated plenty, but for reasons mysterious to me physicists require more. The following definition is fundamental.
\begin{definition*}
We say that a linear functional $\varphi : A \rightarrow \mathbb{C}$ is \emph{represented} by a vector $v$ in a Hermitian $A$-module $H$ when
\begin{displaymath}
\varphi(a) = \langle av, v \rangle_{H},
\end{displaymath}
for all $a \in A$, where $\langle -, -\rangle_{H}$ is the Hermitian form on $H$.
\end{definition*}
Physicists like to work with representations. Which representation should we choose for our expectation value $\langle - \rangle_{\varphi}$? Does such a representation even exist? Mathematics comes to the rescue.
\begin{theorem*}[Gelfand-Naimark-Segal, cf.~\ref{positiveuniversality}]
Let $\varphi$ be a state on a $C^{\ast}$-algebra $A$. Then:
\begin{enumerate}
\item The category of representations of $\varphi$ on Hilbert spaces has an initial object.
\item A representation is initial iff it is topologically cyclic over $A$.
\end{enumerate}
\end{theorem*}
The category of representations consists of all representations and representing-vector-preserving maps between them. The initial object is given by the well known Gelfand-Naimark-Segal construction. Being minimalists, we choose this smallest representation. In the body of this paper we will not restrict ourselves to representations on Hilbert spaces, keeping in mind anomalous gauge theories, and ``no ghost'' theorems. We aim for the statement ``the ghosts don't decouple'' to have a natural mathematical meaning. Saying ``the construction doesn't work'' is not it.

So far we have an object function $\varphi \mapsto GNS(\varphi)$ mapping a state to the representation of its expectation value over $\mathcal{O}(\varphi)$. Our categorical senses tingle. We consider a process $f : \varphi \rightarrow \psi$, and construct the following diagram:
\begin{center}
\begin{tikzpicture}
\node (m) at (6, 0) {$GNS(\varphi)$};
\node (g) at (0, 2) {$GNS(\mathcal{O}(f)^{\ast}\varphi)$};
\node (fm) at (0,0) {$\mathcal{O}(f)^{\ast}GNS(\varphi)$};

\node (phi) at (6, -2) {$\mathcal{O}(\varphi)$};
\node (psi) at (0, -2) {$\mathcal{O}(\psi)$};

\path[->] (fm) edge (m)
		  (g) edge node[auto,swap] {$\exists!$} (fm)
		  (g) edge[dashed] node[auto] {$GNS(f)$} (m)
		  (psi) edge node[auto] {$\mathcal{O}(f)$} (phi);
\end{tikzpicture}
\end{center}
Recall that $\mathcal{O}(f)^{\ast}\varphi = \psi$, so the top module is a representation of $\psi = \langle - \rangle_{\psi}$ over $\mathcal{O}(\psi)$. Modules can be pulled back by homomorphisms, and $\mathcal{O}(f)^{\ast}GNS(\varphi)$ is simply the pullback of the module $GNS(\varphi)$ along $\mathcal{O}(f)$. It's patently obvious that it represents the pullback state $\psi$. But since $GNS(\psi)$ is initial there is a unique vertical map displayed above. The canonical homomorphism $\mathcal{O}(f)^{\ast}GNS(\varphi) \rightarrow GNS(\varphi)$ is a homomorphism of modules over $\mathcal{O}(f)$. We declare $GNS(f)$ to be the composite, so that the diagram commutes. The following theorem follows exclusively from the further application of universal properties.
\begin{theorem*}[cf.~\ref{gnsconstruction}]
The construction above gives a symmetric monoidal functor
\begin{displaymath}
GNS : \mathbf{Phys}^{op} \longrightarrow \ast \mathbf{Mod},
\end{displaymath}
fibered over $C^{\ast}\mathbf{Alg}$.
\end{theorem*}
The codomain is the category of $\ast$-modules. These are representations of $C^{\ast}$-algebras with isometric module homomorphisms along algebra homomorphisms. The theorem above includes the well known fact that
\begin{displaymath}
GNS(\varphi \otimes \psi) = GNS(\varphi) \otimes GNS(\psi),
\end{displaymath}
but it should be emphasized that the \emph{entire value of the this construction is that it defines a functor}. Every single statement and application below is completely dependent on it, just to make sense. Without functoriality this whole enterprise would be worthless.

The contravariance of $GNS$ may be concerning -- don't we want a covariant representation? Not really -- it is this functor that has all the crucial properties that we need in our formalization of physics. But the physically natural direction is easy to recover. We just compose with taking an adjoint (leaving objects untouched):
\begin{center}
\begin{tikzpicture}
\node (a) at (-5,0) {$\mathbf{Phys}$};
\node[text height = 1.5ex, text depth = .25ex] (b) at (0,0) {$\ast\mathbf{Mod}^{op}$};
\node (c) at (5,0) {$\ast\mathbf{Mod}_{adj}$};

\path[->] (a) edge node[auto] {$GNS^{op}$} (b)
		  (b) edge node[auto] {adjoint} (c)
		  (a) edge[bend right = 15] node[auto,swap] {$GNS_{c}$} (c);
\end{tikzpicture}
\end{center}
The codomain category is the category of $\ast$-modules and adjoint homomorphisms -- whose definition is an exercise for the reader (cheaters can skip to definition \ref{adjointhomomorphism}). $GNS_{c}$ is called the \emph{covariant representation}, and is symmetric monoidal just like $GNS$.

At this point we abandon our synthetic pretense. For now, we have all the information we need, and $\mathbf{Phys}$ can be defined as the comma category
\begin{displaymath}
\mathbf{Phys} = 1 \downarrow \mathcal{S},
\end{displaymath}
where $\mathcal{S}$ is the state functor on $C^{\ast}$-algebras
\begin{displaymath}
\mathcal{S} : C^{\ast}\mathbf{Alg}^{op} \longrightarrow \mathbf{Set}.
\end{displaymath}
This means that the objects of $\mathbf{Phys}$ are pairs $(A, \varphi)$, with $\varphi$ a state on $A$, and the morphisms $(A, \varphi) \rightarrow (B, \psi)$ are $C^{\ast}$-algebra homomorphisms $f: B \rightarrow A$ such that $f^{\ast}\varphi = \psi$. The functor $\mathcal{O}$ forgets the state, and
\begin{displaymath}
(A, \varphi) \otimes (B, \psi) = (A \otimes B, \varphi \otimes \psi).
\end{displaymath}

It is important to not forget the synthetic pretense -- the main contribution of this paper is the construction \emph{scheme} for $\mathbf{Phys}$, and \emph{not any specific construction}. While this version of $\mathbf{Phys}$ covers quite a lot, it's not close to being the final thing. The gauge theory problem and ``$L^{1}$ digression'' lose none of their confounding power.

Despite these shortcomings, $\mathbf{Phys}$ captures physics in a stunningly beautiful way. We now turn to demonstrate this.
\subsection{Functorial Physics}
\subsubsection*{Symmetries}
Why would a $G$-symmetric state define a unitary representation of $G$? Textbooks present a rather torturous derivation of this fact. I propose using composition:
\begin{center}
\begin{tikzpicture}
\node (a) at (-3.5,0) {$G$};
\node (b) at (0,0) {$\mathbf{Phys}$};
\node (c) at (4.5,0) {$\ast\mathbf{Mod}_{adj}$};

\path[->] (a) edge (b)
		  (b) edge node[auto]{$GNS_{c}$} (c);
\end{tikzpicture}
\end{center}
Pretty easy! Here we treat $G$ as a one object groupoid, and a $G$-equivariant object is just a functor out of $G$. In fact $G$ can be an arbitrary groupoid, such as inhomogeneous time (various other categories of time are discussed in section \ref{categoriesoftime}).

The picture above describes the following situation. The single object of $G$ maps to a state $\varphi$ in $\mathbf{Phys}$. Every morphism $g \in G$ maps to a process
\begin{displaymath}
g : \varphi \longrightarrow \varphi,
\end{displaymath}
compatibly with identities and composition. This in turn gives homomorphisms of observables and unitary maps of representations:
\begin{align*}
\mathcal{O}(g) : \mathcal{O}(\varphi) &\longrightarrow \mathcal{O}(\varphi) \\
GNS_{c}(g) : GNS(\varphi) &\longrightarrow GNS(\varphi).
\end{align*}
The former preserve the expectation value $\langle - \rangle_{\varphi}$, and the latter preserve the vector representing this expectation $\Omega \in GNS(\varphi)$. These two maps are compatible in the sense that we have the following identity of inner products in $GNS(\varphi)$:
\begin{displaymath}
\langle (g\cdot a)v, w \rangle = \langle agv, gw \rangle,
\end{displaymath}
where on the left $g$ acts only on the observable $a$, and on the right $g$ acts only on the vectors $v$ and $w$. This means that the mapping $a \mapsto g \cdot a$ on observables is unitarily implemented by conjugation $g^{\ast}(-)g$ in $GNS(\varphi)$, as it should be.

The more fundamental compatibility, from which the former follows is
\begin{displaymath}
GNS_{c}(g)(av) = \mathcal{O}(g^{-1})(a) GNS_{c}(g)(v),
\end{displaymath}
for all observables $a \in \mathcal{O}(\varphi)$ and vectors $v \in GNS(\varphi)$. This is just what it means to be a morphism in $\ast\mathbf{Mod}_{adj}$ over an isomorphism of algebras.

All this is fully compatible with composite systems. If $\varphi$ is $G$-equivariant and $\psi$ is $G^{\prime}$-equivariant, then $\varphi \otimes \psi$ is \emph{naturally} $(G \times G^{\prime})$-equivariant, again just because of composition.

By the wonders of category theory ($\mathbf{Cat}$ being cartesian closed) passing to the equivariant $GNS$ construction is as trivial as adorning all formulas with a $G$ in the exponent. It's all just composition. We obtain the following symmetric monoidal functors,
\begin{center}
\begin{tikzpicture}
\node(a) at (0,0) {$\mathbf{Phys}^{G}$};
\node(b) at (4,0) {$\ast\mathbf{Mod}_{adj}^{G}$};
\path[->] (a) edge node[auto] {$GNS_{c}^{G}$} (b);
\node(c) at (7,0) {$\mathbf{Rep}(G)$};
\path[->] (b) edge node[auto] {$U$} (c);
\end{tikzpicture}
\end{center}
where $U$ is the forgetful functor from equivariant modules to unitary representations. A major step in the construction of physical theories is investigating the fibers of $U$.
\begin{digression}
In gauge theories the distinction between ``internal'' and ``external'' symmetries -- actual symmetries and gauge equivalences, appears to be unsustainable. Since gauge equivalences do not alter physical states, none of the preceding discussion seems to apply. We refer again to appendix \ref{appendixa}, where tentative ideas on how to proceed are presented.
\end{digression}
\subsubsection*{Probability Theory}
Many a tome has been written on the supposed mysteries of quantum mechanics. Here we will merely present certain mathematical devices, in the hope that they subtract from, rather than add to the mystery.

Let $\mathbf{Prob}$ be the category of compact probability spaces (with Radon measures) and probability preserving continuous maps (measurable maps require $W^{\ast}$-algebras). For such a space $X$ we may perform two constructions. First we can construct the algebra $C(X)$, of continuous complex-valued functions on $X$. There is a natural state on $C(X)$, given by the expectation value
\begin{align*}
&\mathbb{E} : C(X) \longrightarrow \mathbb{C}\\
&\mathbb{E}(f) = \int_{X} f \, d\mathbb{P}.
\end{align*}
Since $\mathbf{Phys}$ is just algebras with states, this defines a symmetric monoidal functor $C: \mathbf{Prob} \rightarrow \mathbf{Phys}$. It's fully faithful by Gelfand duality, and so we will speak of probability spaces in $\mathbf{Phys}$.

The other construction is $L^{2}(X)$. Gathering all the extra structures on $L^{2}$, we see a symmetric monoidal functor $L^{2} : \mathbf{Prob}^{op} \rightarrow \ast\mathbf{Mod}$. These constructions provide the link between quantum theory and probability.
\begin{theorem*}[cf.~\ref{probabilisticinterpretation}]
The following diagram of symmetric monoidal functors commutes:
\begin{center}
\begin{tikzpicture}
\node (a) at (0,0) {$\mathbf{Prob}^{op}$};
\node (b) at (-3,-2) {$\mathbf{Phys}^{op}$};
\node (c) at (3, -2) {$\ast\mathbf{Mod}$};

\path[->] (a) edge node[auto, swap] {$C^{op}$} (b)
		  (a) edge node[auto] {$L^{2}$} (c)
		  (b) edge node[auto] {$GNS$} (c);
\end{tikzpicture}
\end{center}
\end{theorem*}
\begin{proof}
Totally trivial: $L^{2}(X)$ is cyclic over $C(X)$, and $1 \in L^{2}(X)$ represents the expectation value. By the GNS theorem we are done.
\end{proof}

$L^{2}$ acts by pullback of functions on maps of probability spaces, and taking an adjoint we get a diagram for $L^{2}$ and $GNS_{c}$, where $L^{2}$ acts as ``fiber integration'' or ``density pushforward''.

This theorem provides us with a spectacular application. Let $a \in \mathcal{O}(\varphi)$ be a normal observable. That means that the $C^{\ast}$-algebra generated by $a$, $\langle a \rangle \subseteq \mathcal{O}(\varphi)$ is commutative. Pulling back $\varphi : \mathcal{O}(\varphi) \rightarrow \mathbb{C}$ along this inclusion, we obtain a probability space
\begin{displaymath}
P_{\varphi}(a) = (Spec_{m}(\langle a \rangle), \varphi \lvert_{\langle a \rangle}),
\end{displaymath}
where $Spec_{m}(\langle a \rangle)$ is the Gelfand spectrum of $\langle a \rangle$. The inclusion $\langle a \rangle \subseteq \mathcal{O}(\varphi)$ defines an an \emph{ontological restriction} map
\begin{displaymath}
R: \varphi \longrightarrow P_{\varphi}(a)
\end{displaymath}
in $\mathbf{Phys}$. Is restricting observables really a ``physical process''? Something like this certainly does seem to happen before any measurement. In any case, don't be quick to kick this morphism out of $\mathbf{Phys}$, because the following theorem, and it's proof, are worth keeping around.
\begin{theorem*}[Eigenvalue-Eigenvector Link, cf.~\ref{eelink} and \ref{generalizedeigenvalueeigenvectorlink}]
Let $\lambda \in \mathbb{C}$. The following are equivalent:
\begin{enumerate}
\item $a\Omega = \lambda \Omega$, where $\Omega$ is any vector representing $\varphi$.
\item $a = \lambda$ almost everywhere in $P_{\varphi}(a)$.
\end{enumerate}
\end{theorem*}
\begin{proof}
Just compute $GNS(R)$ using the previous theorem:
\begin{displaymath}
GNS(R) : L^{2}(\varphi \lvert_{\langle a \rangle}) \longrightarrow GNS(\varphi).
\end{displaymath}
This is a morphism of representations of $\varphi$ over the inclusion map $\langle a \rangle \subseteq \mathcal{O}(\varphi)$. So
\begin{displaymath}
a\Omega = \lambda\Omega \textnormal{ iff } a\cdot 1 = \lambda \cdot 1 \textnormal{ in } L^{2} \textnormal{ iff } a = \lambda \textnormal{ a.e.}
\end{displaymath}
\end{proof}
Beyond this argument, $P_{\varphi}(a)$ simply \emph{is} a probability space, with $\langle - \rangle_{\varphi}$ identified as the expectation value on that space. By Gelfand duality $a$ defines a random variable $P_{\varphi}(a) \rightarrow \mathbb{C}$. As a mathematical structure, the Born rule emerges automatically from our formalism. One mystery is reduced to another -- the other being the phenomenological connection between probability theory and reality. This connection is a much more fundamental, and unduly neglected mystery. Still, philosophers have taken note and spilled plenty of ink over it \cite{probabilityphilosophy}.

\subsubsection*{Quantum Markov Processes}
Is pair production really a process in $\mathbf{Phys}$? Not exactly, but it can easily be accommodated\footnote{As long as you believe that QED has an actual scattering matrix.}. First we recall classical Markov processes.

Let $X$ be a compact Hausdorff space. Then the Radon probability measures on $X$, $M(X)$ also form a compact Hausdorff space. A Markov process from $X$ to $Y$ is just a continuous map
\begin{displaymath}
X \longrightarrow M(Y).
\end{displaymath}
The points of $X$ don't map to specific points in $Y$, but rather to probability measures on $Y$ giving distributions of ``where they could have gone''.

Probability measures can be pushed forward, multiplied, and their families integrated against other measures. All this structure amounts to saying that $M$ is a lax monoidal monad
\begin{displaymath}
M : \mathbf{CptHaus} \longrightarrow \mathbf{CptHaus}.
\end{displaymath}
The category of Markov processes is the Kleisli category of this monad $\mathbf{CptHaus}_{M}$, which is monoidal for obvious, and formal category-theoretic reasons \cite{zawadowski}.

Recently, computer scientists (!) have discovered the following amazing theorem.
\begin{theorem*}[Generalized Gelfand Duality, theorem 5.1 in \cite{stochasticgelfand}]
Gelfand duality extends to a contravariant monoidal equivalence between $\mathbf{CptHaus}_{M}$ and the category of completely positive unital maps between commutative $C^{\ast}$-algebras.
\end{theorem*}
This extension is easy to explain using ordinary Gelfand duality. To a completely positive unital map $\Phi : C(Y) \rightarrow C(X)$ we assign the Markov process
\begin{displaymath}
x \mapsto \Phi^{\ast}\delta_{x},
\end{displaymath}
where $\delta_{x}$ is the Dirac delta at $x \in X$, and $\Phi^{\ast}\delta_{x} \in M(Y)$ is its pullback, with $\delta_{x}$ considered as a linear functional on $C(X)$.

This allows us to generalize the entire construction to Markov processes -- simply construct $\mathbf{Phys}$ using completely positive maps instead of algebra homomorphisms. Call the result $\mathbf{Phys}_{M}$. The previously introduced category $\mathbf{Prob}$ can be defined as $1 \downarrow M$ -- the elements of $M$, and probability spaces with Markov maps between them can be defined as
\begin{displaymath}
\mathbf{Prob}_{M} = 1/\mathbf{CptHaus}_{M}.
\end{displaymath}

The entire construction extends and complete probabilistic compatibility is maintained.
\begin{theorem*}[Non-Unitary GNS Representation, cf.~\ref{nonunitarygns}]
There is a commuting prism of symmetric monoidal functors:
\begin{center}
\begin{tikzpicture}
\node (a) at (-1,0) {$\mathbf{Prob}^{op}$};
\node (b) at (-3,-1.5) {$\mathbf{Phys}^{op}$};
\node (c) at (3, -1.5) {$\ast \mathbf{Mod}$};

\node (d) at (-1,-3) {$\mathbf{Prob}_{M}^{op}$};
\node (e) at (-3,-4.5) {$\mathbf{Phys}_{M}^{op}$};
\node (f) at (3, -4.5) {$\mathbf{Hilb}$};

\path[->] (a) edge node[auto, swap] {$C^{op}$} (b)
		  (a) edge node[auto] {$L^{2}$} (c)
		  (a) edge[loosely dashed] (d)
		  (b) edge node[auto] {$GNS$} (c);

\path[->] (d) edge node[auto, swap] {$C^{op}$} (e)
		  (d) edge node[auto] {$L^{2}$} (f)
		  (c) edge node[auto] {$U$} (f)
		  (b) edge (e)
		  (e) edge node[auto] {$GNS_{M}$} (f);
\end{tikzpicture}
\end{center}
\end{theorem*}
On top we see the usual $GNS$ representation, and its relation to probability spaces. The unlabeled vertical arrows are inclusions, and $U$ is the forgetful functor to Hilbert spaces. On the bottom we see the stochastic extension of $GNS$, $GNS_{M}$. Its values are no longer homomorphisms of $\ast$-modules, but merely bounded linear maps. The $L^{2}$ functor also extends in a natural manner.

As before, we define the covariant representation, $GNS_{M,c}$ as the adjoint of $GNS_{M}$:
\begin{displaymath}
GNS_{M,c} = GNS_{M}^{\ast}.
\end{displaymath}

The construction of $GNS_{M}$ is no longer completely trivial. A version for measurable maps between probability spaces would require extending generalized Gelfand duality to von Neumann algebras, and more importantly their morphisms. Rather than focus on the details, let's look at two examples, covered in detail in section \ref{wavefunctioncollapsesection}.

\paragraph{State Vector Collapse.}
Let $\varphi: A \rightarrow \mathbb{C}$ be a state, $P \in A$ a self-adjoint projection, and let
\begin{align*}
\Phi : A &\longrightarrow A\\
a &\mapsto PaP.
\end{align*}
This completely positive map is a noncommutative version of probabilistic conditioning (imagine that $P$ is the indicator function of some event in a probability space). Its $GNS_{M}$ representation can be computed as follows.
\begin{proposition*}[State Vector Collapse]\hspace{1em}
\begin{enumerate}
\item If $\varphi$ is represented by $\Omega$ then $\Phi^{\ast}\varphi$ is represented by $P\Omega$.
\item $GNS_{M}(\Phi)$ is the composite
\begin{center}
\begin{tikzpicture}
\node (a) at (-2,0) {$GNS(\Phi^{\ast}\varphi)$};
\node (b) at (3,0) {$GNS(\varphi)$};
\node (c) at (6,0) {$GNS(\varphi)$};
\path[{Hooks[right]}->] (a) edge node[auto] {inclusion} (b);
\path[->] (b) edge node[auto] {$P$} (c); 
\end{tikzpicture}
\end{center}
\item Consequently, $GNS_{M,c}(\Phi)$ is cyclic (maps $\Omega$ to $P\Omega$), and is the composite
\begin{center}
\begin{tikzpicture}
\node (a) at (9,0) {$GNS(\Phi^{\ast}\varphi)$};
\node (b) at (0,0) {$GNS(\varphi)$};
\node (c) at (3,0) {$GNS(\varphi)$};
\path[->] (c) edge node[auto] {orthogonal projection} (a);
\path[->] (b) edge node[auto] {$P$} (c); 
\end{tikzpicture}
\end{center}
\end{enumerate}
\end{proposition*}
If $A = End(H)$ and $\Omega \in H$, then the inclusions and projections are identities (unless $P\Omega = 0$), and we are left with just the action of $P$ on $H$.

\paragraph{Scattering Theory.}
Let $S : \mathcal{F}(H) \rightarrow \mathcal{F}(H)$ be a unitary scattering operator on the Fock space of some Hilbert space $H$. Let $H_{\alpha}, H_{\beta} \subseteq \mathcal{F}(H)$ be subspaces of states of particles of type $\alpha$ and $\beta$, respectively. We can decompose this scattering matrix in to its ``matrix elements'' $S_{\alpha\beta} :\alpha \rightarrow \beta$, which are quantum Markov processes.
\begin{proposition*}[Matrix Element Decomposition]
There is a process $S_{\alpha\beta} : \alpha \longrightarrow \beta$ in $\mathbf{Phys}_{M}$ such that $GNS_{M,c}(S_{\alpha\beta})$ is the composite
\begin{center}
\begin{tikzpicture}
\node(a) at (0,0) {$H_{\alpha}$};
\node(b) at (3,0) {$\mathcal{F}(H)$};
\node(c) at (6,0) {$\mathcal{F}(H)$};
\node(d) at (9,0) {$H_{\beta}$};
\path[->] (b) edge node[auto]{$S$} (c)
		  (c) edge node[auto] {projection} (d);
\path[{Hooks[right]}->] (a) edge node[auto] {inclusion} (b);
\end{tikzpicture}
\end{center}
\end{proposition*}
This proposition allows giving the informal expression $\gamma + \gamma \rightarrow e^{+} + e^{-}$ its \emph{intended mathematical meaning}. I view the accurate reproduction of physical discourse as a critical indicator of success.

\subsubsection*{Classical Physics and Differential Geometry}
We must unfortunately shift gears and redo everything in a topos. This is briefly outlined in section \ref{sins}, and will be fully fleshed out in a forthcoming paper. The reader is issued a \emph{stack warning} at this point -- proficiency with stacks is assumed past this point.

Let $E$ be a ringed topos, with ring $\mathbf{R}$. Examples to keep in mind are models of synthetic differential geometry \cite{sdg}, especially the Cahiers topos, which contains the convenient vector spaces \cite{kockreyes}. It is unfortunately not obvious whether they occur in valid examples.

Let $\mathbb{E}$ be the stack of objects over $E$ (i.e.~the codomain fibration), and let $\mathbb{E}_{lc}$ be the substack generated by the global sections. It's the stack of ``locally constant'' objects of $E$, which are obtainable by gluing a cocycle of trivial families. In contrast, $\mathbb{E}$ contains all families, with fibers glued ``completely arbitrarily''. The difference between $\mathbb{E}$ and $\mathbb{E}_{lc}$ is like the difference between all bundles and the locally trivial ones. The inclusion $\mathbb{E}_{lc} \subseteq \mathbb{E}$ is fully faithful. The construction of $\mathbf{Phys}$, $\ast\mathbf{Mod}$ as a stacks, and $GNS$ as a stack morphism uses $\mathbb{E}_{lc}$ as a ``universe of sets'' to ensure expected behavior (physics can go wild without the ``$lc$'' in $\mathbb{E}_{lc}$ cf.~remark \ref{wildphysics}).

The easiest way to perform the construction is to invoke stack semantics \cite{shulmansemantics} on an appropriate formula defining $GNS$, substituting $\mathbb{E}_{lc}$ whenever the category of sets is mentioned. The result is that $GNS$ becomes a morphism of monoidal stacks over $E$.

\paragraph{Infinitesimal Symmetries.}
Assume the Kock-Lawvere axiom, and let $D$ be the first order infinitesimals (defined internally as $\{x \in \mathbf{R} : x^2 = 0\}$). Next let $G$ be a group object in $E$, considered as a one object groupoid (more generally, we allow a prestack of groupoids over $E$). The $G$-equivariant states are prestack morphisms
\begin{displaymath}
G \longrightarrow \mathbf{Phys},
\end{displaymath}
analogously to before\footnote{The stackification of $G$ is the stack of $G$-torsors \cite{bunge}, so it's convenient to keep prestacks around.}. Differentiating this amounts to the evaluation of this prestack morphism at $D$. By the Kock-Lawvere axiom this results in an antihomomorphism of Lie algebras
\begin{displaymath}
Lie(G) \longrightarrow \ast Der(\mathcal{O}(\varphi)),
\end{displaymath}
from the Lie algebra of $G$ to the $\ast$-derivations on the observables of $\varphi$. If these derivations have generators, then we can ask about their compatibility with the $GNS$ representation.
\begin{theorem*}
Let $X \in Lie(G)$ act as the inner derivation $[Q, -]$ on $\mathcal{O}(\varphi)$, for some $Q \in \mathcal{O}(\varphi)$. Then $GNS(X)$ acts on $GNS(\varphi)$, and
\begin{displaymath}
GNS(X) = Q \textnormal{ iff } Q\Omega = 0
\end{displaymath}
\end{theorem*}
Thus infinitesimal generators coincide in the Heisenberg and Schr{\"o}dinger pictures only if the representing vector is invariant under the \emph{chosen} generator. This invariance can always be sabotaged, since the center, $Z(A)$, always includes the scalars. Recall that the center is just Hochschild cohomology $HH^{0}(A)$. The theorem above suggests that keeping around choices for generators is a good idea, which in turn suggests lifting the entire formalism to higher (e.g.~derived) categories.

\paragraph{The Classical Limit.}
How do Poisson brackets appear in this setting? $R$ defines the affine line $\mathbb{A}^{1}$ in $E$, and $\hbar$-dependent families of states are simply maps $\mathbb{A}^{1} \rightarrow \mathbf{Phys}$, with $\mathbb{A}^{1}$ seen as categorically discrete. The classical limit is just the restriction to infinitesimal $\hbar$.

$D$ is is an amazingly tiny object in the sense of Lawvere \cite[Appendix 4]{sdg}, and so, restricting such a $\hbar$-family to the first order infinitesimals one studies the construction of $\mathbf{Phys}$ on the trivial families in $E/D$. The Kock-Lawvere axiom shows that these are just the first order deformations of linear functionals on an algebra equipped with a $\ast$-Hochschild cocycle. The antisymmetric part of this cocyle determines a Poisson structure, and the symmetric part controls the deformations of any singularities (principal connections with isotropy groups and spacetimes with nontrivial isometries are examples of singular points in their respective stacks). The $\ast$-part of the cocycle is traditionally taken to be trivial. Working with $\mathbb{E}_{lc}$ protects us from considering any nontrivial families in $E/D$, which are plentiful.

The monoidal structure on $\mathbf{Phys}$ restricts to a product operation on $\ast$-Hochschild cocycles, which generalizes the usual product of Poisson structures. In this sense, classical and quantum composition are fully compatible.

In particular we obtain ``classical Hilbert spaces of states'', which for pure states $x$ on a Poisson manifold $X$ amount to the ``walking $L^{2}$ spaces''
\begin{displaymath}
GNS(x) = L^{2}(\delta_{x}),
\end{displaymath}
considered as modules over $C^{\infty}(X)$. Any nontrivial dynamical flow on $X$ completely changes the entire spaces (outside of fixed points), making them relatively useless. Being one-dimensional is also a drawback. However, there is a ``classical Schr{\"o}dinger equation'' -- it's just a deformation of $L^{2}(\delta_{x})$ as a $C^{\infty}(X)$-module, in some tangent direction in $T_{x}X$. Flow-invariant probability measures $\mu$ support a ``Schr{\"o}dinger equation'' on $L^{2}(\mu)$, with the Poisson bracket interpreted as a differential operator.

This discussion shows, contrary to certain claims in the literature, that the degrees of linearity or non-linearity of quantum and classical theories are exactly the same. The only difference is that classical states don't like sharing their sectors.

\subsection{Current Limitations and Perplexities}
\subsubsection*{Classical Thermodynamics}
We have traded classical thermodynamics, in which entropy is a postulated observable, for statistical mechanics, where there is a formula for entropy. The latter is included in our formalism, under the guise of Markov processes, and the former excluded, as a matter of form. The derivation of thermodynamics from the more modern, probability-based statistical mechanics requires making sense of the ``coarse-graining'' operation, even in a classical setting. This, in turn, requires measure theory in potentially infinite dimensions. This problem in mathematical analysis will have to patiently wait for a proper solution. Physicists should also consider solving problem of actually specifying the measures on the ignored degrees of freedom. This is a serious issue -- no decisive discussion of the thermodynamics of computation can take place before this (for a rare point of clarity on this see \cite{ladyman}). Classical thermodynamics is not essential to the program outlined below. The other omissions are more serious, and will be the focus of future work.

\subsubsection*{Gauge Symmetry, Gravity, and Extended Locality}
Gauge theories are, by my own standards, not included in this formalism. My current understanding of this problem is presented in appendix \ref{appendixa}. The moral of that story is that higher categories are essential for the proper treatment of theories with gauge equivalences, and that the conceptual structures underpinning gauge theories are not clear at all. The notion of symmetry may have to be revised.

Next in line is the general notion of locality. I have specifically taken care to avoid saying ``spacetime'' in any part of this work. String theory looms large, and the door to emergent spacetime must be kept open, even if nothing passes through. But, independently of ideology, locality -- especially extended locality -- is conceptually confusing. General Relativity is a theory \emph{of} spacetime, not \emph{in} spacetime. The idea of locality in gravitationally coupled theories is extremely unclear, and will be investigated in forthcoming work \cite{stackygr}.

The common ground between extended locality and this work is inaccessible due to the following perplexing questions:
\begin{enumerate}
\item Does $\lambda\varphi^{4}$ define an extended field theory?
\item Does Yang-Mills theory define an extended field theory?
\end{enumerate}
How far do these theories extend? In which dimensions? Why would Dp-brane excitations define a p-category, and not the usual Hilbert space (0-category!) found in textbooks? Are defects with prescribed support inherently perturbative, non-dynamical objects? After all, D-brane modes can induce physical motion. Does this imply that defect cobordisms describe off-shell processes? None of these issues are clear to me.

There is one hint available: the $\lambda \varphi^{4}$ lagrangian does not appear to define an extended lagrangian (cf.~\cite[Appendix]{freedquant}), and the scalar field does not have any interesting boundary conditions in higher codimensions. This suggests that scalar field theory does not extend, and that there is a hierarchy of $n$-extendible theories, with $n \in \mathbb{N}$. Structures like $\mathbf{Phys}$ would then describe its bottom floor.

The last and greatest omission is string theory. The standard perturbative formalism does define an object in $\mathbf{SymMonCat}/\mathbf{Phys}$, the 2-category of ``generalized physical theories'', but this construction does not properly capture any dualities. The central idea of string theory still seems to be missing. At a more technical level, string theory contains higher gauge fields, leading us back to the problem of integrating gauge symmetry with the construction of $\mathbf{Phys}$.

\subsubsection*{Conceptual Limitations of $C^{\ast}$-algebras}
Despite the disavowal of $C^{\ast}$-algebras in the introduction, some concept of completeness providing a supply of modules isomorphic to their duals seems necessary to give the internal constructions of section \ref{sins} realistic examples.

Nevertheless, there is a long list of reasons, beyond the ``$L^{1}$ digression'' in the introduction, for abandoning $C^{\ast}$-algebras, particularly the ``$C$'' part of $C^{\ast}$, and their topological kin, as the nexus of formalization of quantum theory:
\begin{enumerate}
\item Let $\mathcal{F}$ be the space of classical fields of some field theory. As is evident in \cite{freedfield}, any serious development of classical field theory requires the consideration of the de Rham bicomplex $\Omega^{\ast}(\mathcal{F} \times M) = \Omega^{\ast}(\mathcal{F}) \otimes \Omega^{\ast}(M)$, where $M$ is spacetime. The algebra $C^{\infty}(\mathcal{F})$ is simply not enough, as it does not determine the required bicomplex.
\item The incorporation of fermions requires working with superalgebras, even classically. Otherwise deformation quantization can never yield anticommutation relations. This is no problem on its own, but:
\item Fermionic fields are odd points of superfunction spaces. To preserve them one must work with ringed sites over these function spaces.

For example, the space of sections of a superbundle $E \rightarrow X$, $\Gamma(E)$, defined naturally as a subobject of the sheaf exponential $E^{X}$, gives rise to the site $Y \downarrow \Gamma(E)$, where $Y : \mathbf{Sm} \rightarrow Sh(\mathbf{Sm})$ is the Yoneda embedding of supermanifolds into the category of sheaves over itself. The natural algebra of observables to consider in this case is the sheaf of superalgebras
\begin{displaymath}
(U \longrightarrow \Gamma(E)) \longmapsto C^{\infty}(U).
\end{displaymath}
The global sections $1 \rightarrow \Gamma(E)$ consist of purely even fields, and so considering only them is insufficient. Doing so would result in a complete absence of fermionic observables, and consequently no possibility of anticommutation relations in quantum field theory.
\item The incorporation of gauge invariance complicates the picture even more. We refer again to appendix \ref{appendixa}. Even a naive incorporation of the BV-BRST formalism would require complexes of objects.
\end{enumerate}
Naively adding these points together, we are faced with sheaves of differential graded super-$C^{\ast}$-algebras, as the \emph{bare minimum} for expressing the standard model. Always true to form, gravity demands much more:
\begin{enumerate}
\item[5.] General Relativity is properly thought of as a dynamical \emph{theory of spacetime}, rather than a theory of the gravitational \emph{field in spacetime}. This means that gravity is prior to other fields, and requires the consideration of the ``space of all spacetimes'', i.e.~the stack of Lorentzian manifolds. This stack will be analyzed in detail in forthcoming work \cite{stackygr}. The unfortunate result of this analysis is that we must internalize everything into the category of sheaves on that stack.
\end{enumerate}
So a \emph{minimal} incorporation of fermions, gauge fields, and gravity necessitates a consideration of internal sheaves of differential graded super-$C^{\ast}$-algebras.

We cannot simply ignore these foundational structural issues. The rift between formal mathematics and physics cannot be allowed to grow any larger than it is right now. And despite the advent of ``physical mathematics'' \cite{moore}, of perhaps because of it, the rift has been growing.

\subsubsection*{The Problem of Wilsonian Ice Cubes}
If the project of section \ref{sins} can be successfully populated with interesting examples, then the Wilsonian picture of renormalization, and in particular of critical phenomena, will become available. The very essence of considering families of theories is turning the $GNS$ functor into a morphism of stacks.

However \emph{localized phase transitions} will still be a mystery. Consider the process of making ice cubes. Since the thermodynamical temperature is an external parameter, and not a localizable dynamical quantity, the act of making our cubes destroys the stars and makes the intergalactic medium boil. I would like to think that the production of ice cubes does not require traversing a family of parallel realities, each with its own distinct physics.

Despite the tongue-in-cheek narration, the problem is serious. It's not just that mixed phases must be far from equilibrium. It's what mixed phases actually are, as a mathematical structure. What does it mean to have ice \emph{here} and not \emph{there}? The crucial point is that the Wilsonian picture is metatheoretical -- we deal with the space of all theories. These theories describe only parts of the world, but they think otherwise. The ``logical signature'' of an effective field theory looks just like any other QFT. As a matter of formal structure they no different from fundamental theories.

There must be a dynamical theory of localized phase transitions. How are distinct effective descriptions patched together in spacetime? The statistical ensembles cannot form a sheaf on spacetime (or any similar structure), since the ``rest of the world'' is almost never a reservoir of the appropriate type. Despite this, thermometers work even when there is no well-defined temperature. What is the meaning of the numbers they produce?

The problem of reconciling effective theories with their spatiotemporal domains of validity is a critical conceptual component of mathematical physics. Doubly so when we realize that our experiments are localized in spacetime.
\subsection{Motivation}
My aim is to take the language of the physicists at face value -- path integrals and all, to the greatest possible extent allowed by the law. Give it mathematical semantics, and ultimately express (much less prove) conjectures like ``Witten's theorem'' -- that spaces of vacua in certain Yang-Mills theories have trivial dependence on $\hbar$ (cf.~conjecture \ref{wittenstheorem}). Without being castrated by premature mathematical formalization, this language has proven to give its practitioners powerful vision, and insight into the mathematical world, not to mention a basket of Nobel prizes and a Fields medal. Edward Witten, in particular, has sight where mathematicians are blind. But we cannot allow mediators or middlemen to guide us to the truth. Nature is a good approximation to mathematics, but it's not the real thing.

We must abandon the fear of not making it back to the mathematical mainland -- that we can never get the stories right the way they're told, take an intellectual swim, and listen to what physicists actually have to say. Doing this, one sees that the arguments used by Witten \cite{witten} are compelling, in the sense that they can be expressed in a fully typed formal language, whose expected semantics take values in the complete mathematical theory of quantum fields\footnote{A similar statement about string theory would be false, at least today.}. Language and its meaning -- these two objects are separable, and the former \emph{dictates the form} of the latter. This is a severe restriction, and invaluable tool that we have the bad habit of discarding, mangling the types of objects physicists discuss beyond recognition. The content of this work is \emph{uniquely determined just by trying not to do it}.

Most mathematicians and physicists confuse an understanding of this language -- and its source, ``physical intuition'' -- with the construction of mathematical approximations to the expected full semantics. This makes listening difficult, since it requires disentangling the intended statements from their faulty mathematical cloak. Fortunately there are physicists who speak clearly. Weinberg, \emph{after explicitly distancing himself from ``rigor''}, managed to convey QFT with conceptual clarity that is unmatched by other texts \cite{weinberg}. Among these texts I include the entire literature on constructive quantum field theory.

Another effect of this confusion is that an eminently reasonable question, such as
\begin{quote}
Is a D-brane \emph{actually} a tachyonic condensate \cite{stringfieldtheory}, or \emph{actually} a boundary condition \cite[8.7]{polchinski}?
\end{quote}
can be ineffectually answered by
\begin{quote}
\emph{Actually}, a D-brane is, \emph{by definition}, a certain $KK$-class \cite{brodzki}.
\end{quote}
By definition! None of these D-brane notions can coincide -- that would be a type error\footnote{as in programming and computer science}. The best we can hope for is that a single object of a different kind determines, in appropriate circumstances, the members of these three diverse categories. Giving a premature definition makes this not only formally impossible, but also steers thinking away from these crucial foundational issues. Type errors cannot be corrected by cleverness or computation, since \emph{types reflect intent}. The only way to deal with them is to change one's mind.

As stated, my interest in physics is the construction of this language, and its semantics. The purpose is to allow the import of physical intuition, developed over the past century, into mathematics. Since this intuition greatly exceeds our mathematical understanding (e.g.~\cite{wildramification}), this should allow great progress, not just in stating theorems (as has been happening in the past decades), but in the technique of proof. Rather than receive toys from physicists, I scheme to steal the toy factory. A grand heist.

The present work is the first step in this program. Here I begin outlining the form of a mathematical structure in which the entirety of physics has a common meeting ground. The language developed here has \emph{a chance} to faithfully express the stories that physicists tell. It is incomplete in its current form, but more complete, by far, than anything I have found in the literature.

\subsection{Detailed Organization}
In section \ref{preliminaries} we establish definitions, conventions, and recall elementary algebraic results in their most useful forms (for our purposes, at least). This section was written with topoi in mind, so we work in considerable generality, in excess of what is actually needed outside of section \ref{sins}. We work with arbitrary $\ast$-algebras and nondegenerate Hermitian $\ast$-modules over them.

As seen in the introduction, the lack of topology is not a technical limitation. The reader can effortlessly redo the entire paper for $C^{\ast}$-algebras, and likely (with some effort) even for von Neumann algebras\footnote{All the work is in the morphisms, since $W^{\ast}$-algebras ``are'' $C^{\ast}$-algebras.}. We will not do any of this, for reasons stated previously.

Section \ref{basicconstruction} introduces the basic construction scheme. We begin by studying the notion of representability of a state, without any normalization or positivity conditions. We characterize representability in theorem \ref{representabilityconditions}, and provide the proper generalization of the GNS construction for such objects. We show that representable states are convex in all linear functionals, and establish that the state functor is symmetric lax monoidal.

Next we study positivity. No topology is required. In theorem \ref{positiveuniversality} we show that the GNS representation of a positive state is initial among all pre-Hilbert representations, and proceed to link our variant of the GNS construction to the traditional one. We define complete positivity for $\ast$-algebra maps and derive a variant of the Stinespring factorization theorem -- theorem \ref{stinepring}. It is used in section \ref{nonunitarygnssection} to show that quantum Markov processes have GNS representations.

Next we work to define all the variants of the category of physical processes. They include taking everything, just the positive states, or just the admissible morphisms. Admissibility is required for turning the GNS construction into a functor, which is then automatically strong symmetric monoidal. All processes between positive states are admissible. Finally we give two versions of the covariant representation, depending on how much topology is allowed.

Section \ref{examplesection} is devoted to sample computations and examples. We compute the action of the $GNS$ functor in relation to the functor of pulling back states. We show how to incorporate antilinear processes into $\mathbf{Phys}$, with theorem \ref{conjugationtheorem} protecting us from boundless confusion. We tackle the problem of non-normalized states, providing a functorial normalization procedure. Finally we discuss the classic examples of the GNS construction, the $L^{2}$ space and pre-Hilbert spaces over their endomorphism algebras.

In section \ref{physics1} we begin the formal reconstruction of textbook physics from our formalism. Theorem \ref{schrodingerpicture} and corollary \ref{schrodingerpicturecorollary} serve as an example factory, showing how to lift Schr{\"o}dinger picture unitary operators to maps in $\mathbf{Phys}$ while preserving their intended representations.

Next we tackle the probabilistic interpretation of quantum mechanics, starting with the fundamental relation between the $GNS$ functor and the $L^{2}$ functor, given by theorem \ref{probabilisticinterpretation}. This allows us to derive a canonical random variable from any normal observable, giving the eigenvalue-eigenvector link and the Born rule. The link persists in much greater generality, and we rederive it in such in theorem \ref{generalizedeigenvalueeigenvectorlink}. Since our algebras include the $C^{\ast}$- and $W^{\ast}$-categories, and we deny ourselves the use of spectral theory, the presentation is not as elegant as in the introduction.

Next we discuss symmetries and group representations. The functorial nature of GNS makes this essentially trivial. We show how to deal with time reversal and inhomogeneous, irreversible time evolution.

Finally we characterize the monoidal structure on $\mathbf{Phys}$ for normalized states in terms of axioms describing system composition.

In section \ref{statisticalphysics} we generalize our constructions to noncommutative Markov processes. We extend the notion of admissibility to $\ast$-linear maps which are not necessarily homomorphisms, and show that all completely positive maps are admissible for positive states. We extend the $GNS$ functor to admissible $\ast$-linear maps and show that this extension is maintains complete probabilistic compatibility, as given by theorem \ref{probabilisticinterpretation}, in theorem \ref{completeprobabilisticcompatibility}. To state that theorem we extend Gelfand duality, following \cite{stochasticgelfand}, to Markov processes between compact Hausdorff spaces. To illustrate this extension we show how arbitrary orthogonal projections can be seen as representations of noncommutative conditioning maps.

In the final subsections we propose investigating the relation of the non-unitary GNS representation to bordism representations, information theory.

Section \ref{sins} is provided for interested readers, and sketches the internalization of the GNS representation into models of synthetic differential geometry. The formalisms of infinitesimal symmetries, the classical limit and deformation quantization can all be seen to have a place there. The intended application of this construction is discussed in appendix \ref{appendixb}.

\section{Algebraic Preliminaries}\label{preliminaries}
\subsubsection*{Conventions}
We assume that algebras have units, and that homomorphisms preserve them. We do not assume commutativity. By ``module'' we mean left module, likewise for ideals. Unlabeled tensor products are taken over $\mathbb{C}$, except in section \ref{cyclicmodulessection}, where they are over $\mathbb{Z}$.
\subsection{$\ast$-Algebras}
\begin{definition}
A $\ast$-algebra is a $\mathbb{C}$-algebra $A$, together with a conjugate-linear, involutive anti-homomorphism $\ast : A \rightarrow A$.
\end{definition}
A map of $\ast$-algebras is a $\mathbb{C}$-algebra homomorphism which preserves the $\ast$ operation. In this way $\ast$-algebras organize into a category, which we will denote by $\ast \mathbf{Alg}$.

In the commutative case, the role of the $\ast$ operation can be understood completely through Galois descent.
\begin{lemma}
Let $A$ be a commutative $\mathbb{C}$-algebra. Then $\ast$ operations on $A$ correspond to semilinear $Gal(\mathbb{C}/\mathbb{R})$-actions on $A$.
\end{lemma}
\begin{proof}
This is immediate from the definition of a semilinear action, and the fact that $Gal(\mathbb{C}/\mathbb{R})$ is generated by conjugation.
\end{proof}
By Galois descent we obtain the following corollary.
\begin{corollary}\label{commutativestaralgebras}
The category of commutative $\ast$-algebras is equivalent to the category of $\mathbb{C}$-algebras with chosen real form.
\end{corollary}
The equivalence maps $A$ to its $\mathbb{R}$-subalgebra of self-adjoint elements, traditionally denoted by $A_{sa}$.

In the noncommutative case, the $\ast$-operation is not a semilinear Galois action, and its physical significance remains mysterious to me.

\subsection{Bilinear Forms}
We must recall some facts about bilinear forms and their radicals. Let $R$ be a commutative ring.
\begin{definition}
Let $M$ be an $R$ module. If $M$ is equipped with and $R$-bilinear form
\begin{displaymath}
\langle -, - \rangle_{M} : M \otimes_{R} M \longrightarrow R,
\end{displaymath}
we will call it a \emph{bilinear module} over $R$. The bilinear form determines its left and right \emph{radicals}:
\begin{align*}
M^{\perp} & = \{m \in M : \langle m , - \rangle_{M} = 0 \} \\
\leftidx{^{\perp}}M & = \{m \in M : \langle -, m \rangle_{M} = 0 \}.
\end{align*}
Elements of these radicals are called left (right) \emph{degenerate}, respectively, and modules with vanishing left (right) radicals are called left (right) \emph{nondegenerate}.
\end{definition}
For symmetric and Hermitian forms both radicals obviously coincide, and there is a unique notion of nondegeneracy.
\begin{remark}\label{isometryremark}
We will call all maps preserving given bilinear forms \emph{isometries}, even if the forms have no geometric significance.
\end{remark}
\begin{definition}
Let $M$ and $N$ be bilinear modules. An morphism $f : M \rightarrow N$ is called \emph{right adjointable}, if there exists a map $f^{\ast} : N \rightarrow M$ such that 
\begin{displaymath}
\langle f(m), n \rangle_{N} = \langle m, f^{\ast}(n) \rangle_{M},
\end{displaymath}
for all $m \in M$ and $n \in N$. We will call this map a \emph{right adjoint} to $f$. Left adjointable maps are defined analogously.
\end{definition}
If $M$ is right nondegenerate, then $f^{\ast}$ is unique, and adjointness implies the linearity of $f^{\ast}$ (which we require anyway, but is not always necessary). Without additional assumptions adjoints may fail to exist. For symmetric and Hermitian forms there is a unique notion of adjoint.

The following lemma is extremely useful in the various constructions we will undertake. It controls the behavior of degenerate vectors under adjointable maps.
\begin{lemma}\label{isotropiclemma}
Let $M$ and $N$ be bilinear modules. If $f : M \rightarrow N$ has a left adjoint $f^{\ast}$, then
\begin{align*}
f(M^{\perp}) &\subseteq N^{\perp} \\
f^{\ast}(\leftidx{^{\perp}}N) &\subseteq \leftidx{^{\perp}}M.
\end{align*}
\end{lemma}
\begin{proof}
$\langle f(m) , n \rangle_{N} = 0$ iff $\langle m, f^{\ast}(n) \rangle_{M} = 0$. So $f(m)$ is left-degenerate if $m$ is, and $f^{\ast}(n)$ is right-degenerate if $n$ is. 
\end{proof}
In other words, left adjoint maps preserve left radicals and right adjoint maps preserve right radicals. The non-uniqueness of the adjoint is irrelevant, and the linearity of $f$ and $f^{\ast}$ is not required above.

Bilinear modules can be added and multiplied. The orthogonal direct sum $M \oplus N$ has carries the bilinear form
\begin{displaymath}
\langle (m,n) , (m^{\prime}, n^{\prime}) \rangle_{M \oplus N} = \langle m, m^{\prime}\rangle_{M} + \langle n, n^{\prime} \rangle_{N}.
\end{displaymath}
The radicals of a direct sum are easily computed.
\begin{proposition}\label{nondegeneratesums}
Let $M, N$ be bilinear modules, with $M \oplus N$ their orthogonal direct sum. Then their left radicals satisfy
\begin{displaymath}
(M \oplus N)^{\perp} = M^{\perp} \oplus N^{\perp},
\end{displaymath}
with an analogous formula for right radicals.
\end{proposition}
\begin{proof}
Clearly we have $M^{\perp} \oplus N^{\perp} \subseteq (M \oplus N)^{\perp}$. To show the other inclusion suppose that $\langle (m,n), - \rangle_{M \oplus N} = 0$. Evaluating this on $(m^{\prime},0)$ we see that $m \in M^{\perp}$. Evaluating on $(0, n^{\prime})$ we see that $n \in N^{\perp}$.
\end{proof}

The tensor product of bilinear modules $M \otimes_{R} N$ is also bilinear, through the formula
\begin{displaymath}
\langle m \otimes n, m^{\prime} \otimes n^{\prime} \rangle_{M \otimes N} = \langle m, m^{\prime} \rangle_{M} \langle n , n^{\prime} \rangle_{N}.
\end{displaymath}

Without additional assumptions the radicals can misbehave under tensor products. Any bilinear module $M$ determines an exact sequence
\begin{equation}\label{exactsequence}
0 \longrightarrow M^{\perp} \longrightarrow M \longrightarrow Hom_{R}(M, R),
\end{equation}
with the last arrow being $m \mapsto \langle m, - \rangle_{M}$. Tensoring such sequences results in any number of homological situations. Here we will simply assume that nothing can go wrong.
\begin{lemma}\label{flatnesslemma}
Let $M$ and $N$ be bilinear vector spaces over a field $k$. Then their left radicals satisfy
\begin{displaymath}
(M \otimes_{k} N)^{\perp} = M^{\perp} \otimes_{k} N + M \otimes_{k} N^{\perp},
\end{displaymath}
and analogously for the right radicals. In particular, if $M$ and $N$ are nondegenerate, then so is $M \otimes_{k} N$.
\end{lemma}
\begin{proof}
The radical $(M \otimes_{k} N)^{\perp}$ is clearly the kernel of the map
\begin{displaymath}
M \otimes_{k} N \longrightarrow (M \otimes_{k} N)^{\ast},
\end{displaymath}
where $m\otimes n$ maps to $\langle m, - \rangle_{M} \langle n, - \rangle_{N}$. This map is arises as the composite
\begin{displaymath}
M \otimes_{k} N \longrightarrow M^{\ast} \otimes_{k} N^{\ast} \longrightarrow (M \otimes_{k} N)^{\ast},
\end{displaymath}
where the last arrow is the natural one (arising from $\otimes_{k}$ begin a functor), and the first is
\begin{displaymath}
m\otimes n \mapsto \langle m, - \rangle_{M} \otimes \langle n, - \rangle_{N}.
\end{displaymath}
Since the natural map $M^{\ast} \otimes_{k} N^{\ast} \rightarrow (M \otimes_{k} N)^{\ast}$ is injective, the result follows by taking the tensor product of the sequences (\ref{exactsequence}) for $M$ and $N$.
\end{proof}
\begin{remark}\label{derivedremark}\hspace{1em}
\begin{itemize}
\item This lemma is useless in topoi, effectively limiting the supply of examples in section \ref{sins}.
\item The discussion here is a first indicator that derived categories are warranted in a more complete development of $\mathbf{Phys}$. A bilinear module $M$ should be replaced by the complex $\mathbf{M}$ given by
\begin{displaymath}
0 \longrightarrow M^{\perp} \longrightarrow M,
\end{displaymath}
with the nondegenerate form recovered as the cohomology $H^{0}(\mathbf{M})$.
\item The real property required of bilinear modules $M$ in our constructions is that the functor $M \otimes_{R}(-)$ preserves nondegenerate bilinear forms.
\end{itemize}
\end{remark}
The final lemma will serve to define the tensor product of $\ast$-modules, once they have been defined.
\begin{lemma}\label{adjointlemma}
Let $f : M \rightarrow N$ and $g : S \rightarrow T$ be right adjointable maps of bilinear $R$-modules. Then $f \otimes_{R} g : M \otimes_{R} S \rightarrow N \otimes_{R} T$ is also right adjointable.
\end{lemma}
\begin{proof}
The adjoint is obviously $f^{\ast} \otimes g^{\ast}$, for any two right adjoints $f^{\ast}$, $g^{\ast}$ of $f$ and $g$, respectively, since
\begin{multline*}
\langle f\otimes g (m \otimes s) , n \otimes t \rangle = \langle f(m), n \rangle \langle g(s), t \rangle = \\ \langle m, f^{\ast}(n) \rangle \langle s, g^{\ast}(t) \rangle = \langle m \otimes s, f^{\ast} \otimes g^{\ast} (n \otimes t) \rangle.
\end{multline*}
One can then extend by multilinearity to all tensors, or interpret the above as a diagrammatic computation. Either way, the possible degeneracy poses no problems.
\end{proof}
\subsection{$\ast$-Modules}
Let $M$ be a nondegenerate Hermitian complex vector space. By $\underline{End}(M)$ we denote the space of adjointable endomorphisms of $V$. We record the obvious fact that it is a $\ast$-algebra.
\begin{proposition}
$\underline{End}(M)$ is a $\ast$-algebra, with $\ast$ mapping each endomorphism $f$ to its associated $f^{\ast}$.
\end{proposition}
\begin{remark}\label{boundednessremark}
Adjointable maps between Hilbert spaces are exactly the bounded ones. This follows from the uniform boundedness principle.
\end{remark}
\begin{definition}
Let $A$ be a $\ast$-algebra. A $\ast$-module over $A$ is a nondegenerate Hermitian vector space $M$, together with a map of $\ast$-algebras $A \rightarrow \underline{End}(M)$.
\end{definition}
\begin{remark}\label{submoduletrouble}
The intersection of nondegenerate subspaces of a quadratic space may be degenerate, and hence the ``$\ast$-submodule generated by $X$'' need not exist without additional assumptions, such as positivity of the Hermitian form. One must be extremely careful to prove that any expected submodules actually exist.
\end{remark}
Because of this, in the sequel $\ast$-modules will always be named such, and will be strictly distinguished from ordinary modules, which will appear in the course of our constructions.

We will need a notion of homomorphism between $\ast$-modules over different algebras. Let $f : A \rightarrow B$ be a map of $\ast$-algebras, and let $M$ be a $\ast$-module over $A$, and $N$ be a $\ast$-module over $B$.
\begin{definition}
A map of $\ast$-modules $h : M \rightarrow N$ over $f$ is an isometric $\mathbb{C}$-linear map (cf.~remark \ref{isometryremark}), such that $h(am) = f(a) h(m)$ for all $a \in A$ and $m \in M$.
\end{definition}
An ordinary map is simply a map over the identity of the underlying algebra. Our work will also require a slightly more exotic notion of homomorphism.
\begin{definition}\label{adjointhomomorphism}
A linear map $h : N \rightarrow M$ of $\ast$-modules is an adjoint homomorphism over $f$ if it is a coisometry (adjoint of an isometry) of the underlying Hermitian forms and $a h(n) = h(f(a) n)$ for all $a \in A$ and $n \in N$.
\end{definition}
Note the direction. The name comes from the following obvious proposition.
\begin{proposition}\hspace{1em}
\begin{enumerate}
\item Let $h: M \rightarrow N$ have an adjoint $h^{\ast} : N \rightarrow M$. Then $h$ is a homomorphism over $f$ iff $h^{\ast}$ is an adjoint homomorphism over $f$.
\item Adjoint homomorphisms over an invertible map $f$ are exactly the homomorphisms over $f^{-1}$.
\end{enumerate}
\end{proposition}

\subsection{The Fibration of $\ast$-Modules}
\begin{definition}\hspace{1em}
\begin{itemize}
\item The category $\ast\mathbf{Mod}$, of $\ast$-modules and their homomorphisms, has as objects pairs $(A,M)$, where $A$ is an $\ast$-algebra, and $M$ is a $\ast$-module over $A$.

The morphisms are pairs $(f,h): (A, M) \rightarrow (B, N)$, where $f : A \rightarrow B$ is a morphism of $\ast$-algebras, and $h : M \rightarrow N$ is a morphism of $\ast$-modules over $f$.
\item The category $\ast\mathbf{Mod}_{adj}$ is defined analogously, but with maps $(f,h): (A,M) \rightarrow (B,N)$, where $f: B \rightarrow A$ is a map of $\ast$-algebras, and $h$ is an adjoint homomorphism over $f$.
\end{itemize}
\end{definition}
There is an obvious projection functor $\pi : \ast\mathbf{Mod} \rightarrow \ast\mathbf{Alg}$, which forgets the modules. This map is a fibration (in the sense of Grothendieck, cf.~\cite{streicher} or \cite{vistoli}).
\begin{theorem}
$\pi$ is a Grothendieck fibration.
\end{theorem}
\begin{proof}
Let $f : A \rightarrow B$ be a morphism of $\ast$-algebras, and let $N$ be a $\ast$-module over $B$. The cartesian (sometimes called prone) lifting of $f$ can be constructed as follows.

The domain $f^{\ast}N$ is just $N$ as a Hermitian vector space, with module structure given by the composite
\begin{displaymath}
A \xrightarrow{f} B \longrightarrow \underline{End}(N),
\end{displaymath}
where the last arrow is the $\ast$-module structure of $N$.

The homomorphism $f^{\ast}N \rightarrow N$ is just the identity, as a function of sets.

Clearly, such maps are closed under composition, and the morphisms $M \rightarrow N$ over $f$ factor uniquely through the lift $f^{\ast}N \rightarrow N$ to module morphisms over $A$ (i.e.~over the identity on $A$).
\end{proof}

\subsection{Tensor Products}
\subsubsection{Tensor Products of $\ast$-Algebras}
Recall the universal property of the tensor product of rings.
\begin{theorem}[Universal Property of the Tensor Product]
Let $R, S$ be unital rings. Their tensor product $R \otimes_{\mathbb{Z}} S$ is initial among the rings $T$ with ring homomorphisms
\begin{align*}
f : R &\longrightarrow T \\
g : S &\longrightarrow T,
\end{align*}
such that the images of $f$ and $g$ commute in $T$.
\end{theorem}
\begin{proof}
The tensor product certainly is such a ring, with $f$ and $g$ given by
\begin{align*}
r &\longmapsto r \otimes 1\\
s &\longmapsto 1 \otimes s.
\end{align*}
Now consider $T$ and arbitrary maps $f$ and $g$, as in the statement of the theorem. The map
\begin{align*}
R \times S & \longrightarrow T \\
(r,s) & \longmapsto f(r) g(s)
\end{align*}
is clearly bilinear, and hence factors through $R \otimes_{\mathbb{Z}} S$. Since the images of $f$ and $g$ commute, it's a homomorphism of rings. Finally, the composites
\begin{align*}
R &\rightarrow R \times S \longrightarrow T \\
r &\mapsto (r,1)  \mapsto f(r)g(1) \\ \\
S &\rightarrow R \times S \longrightarrow T \\
s &\mapsto (1,s) \mapsto f(1)g(s),
\end{align*}
are $f$ and $g$, respectively, showing that the factorization through $R \otimes_{\mathbb{Z}} S$ recovers $f$ and $g$, and that the factorization is unique.
\end{proof}
\begin{remark}
Let $R \ast S$ be the coproduct of $R$ and $S$ in the category of rings. Then there is an obvious map
\begin{displaymath}
R \ast S \longrightarrow R \otimes_{\mathbb{Z}} S,
\end{displaymath}
which is easily seen to be surjective, by the fact that the images of $R$ and $S$ generate both rings. This gives a different construction of $R \otimes_{\mathbb{Z}} S$, and shows that the identity is a symmetric monoidal functor
\begin{displaymath}
(\mathbf{Rng}, \otimes_{\mathbb{Z}}) \longrightarrow (\mathbf{Rng}, \ast).
\end{displaymath}
Taking opposite categories, we see that this relates the ``naive'' product of noncommutative spaces to their traditional ``product''.
\end{remark}
Now let $A$ and $B$ be $\ast$-algebras. Then $A\otimes B$ is an $\ast$-algebra, with $\ast$ given by
\begin{displaymath}
(a \otimes b)^{\ast} = a^{\ast} \otimes b^{\ast}.
\end{displaymath} 
This is well-defined, since $A$ and $B$ commute in $A \otimes B$. The universal property of the preceding theorem persists.
\begin{theorem}\label{tensorproductuniversalproperty}
$A \otimes B$ is initial among the $\ast$-algebras $C$ with $\ast$-homomorphisms from $A$ and $B$ whose images commute.
\end{theorem}
\begin{proof}
The same proof as before applies to the homomorphism part. It's obvious that the $\ast$-structure is respected.
\end{proof}
\subsubsection{Tensor Products of $\ast$-Modules}
Recall that if $M$ is an $R$-module and $N$ is an $S$-module, then $M \otimes_{\mathbb{Z}} N$ is and $R \otimes_{\mathbb{Z}} S$-module, with $r \otimes s$ acting as
\begin{displaymath}
m \otimes n \longmapsto rm \otimes sn.
\end{displaymath}
The same thing happens with $\ast$-modules, but we must be careful about nondegeneracy and adjointability.
\begin{lemma}
If $M$ is a $\ast$-module over $A$ and $N$ is a $\ast$-module over $B$, then $M \otimes N$ is naturally a $\ast$-module over $A \otimes B$.
\end{lemma}
\begin{proof}
The module structure on $M \otimes N$ is not in question. The bilinear form on $M\otimes N$ is nondegenerate by lemma \ref{flatnesslemma}. To see that the $\ast$-structure is preserved, note that by lemma \ref{adjointlemma} the structure map
\begin{displaymath}
A \otimes B \longrightarrow \underline{End}(M) \otimes \underline{End}(N) \longrightarrow End(M \otimes N)
\end{displaymath}
actually lands in the adjointable maps $\underline{End}(M \otimes N) \subseteq End(M \otimes N)$.
\end{proof}
This allows us to prove the following important theorem.
\begin{theorem}
The fibration of $\ast$-modules $\pi : \ast\mathbf{Mod} \rightarrow \ast \mathbf{Alg}$ is a strong symmetric monoidal functor.
\end{theorem}
\begin{proof}
Properly speaking, this is obvious once we know the domain is symmetric monoidal. But this is obvious: the usual structure on modules extends to $\ast$-modules, since the structure maps
\begin{align*}
M \otimes(N \otimes O) & \xrightarrow{\alpha} (M \otimes N) \otimes O \\
I \otimes M & \xrightarrow{\lambda} M \\
M \otimes I & \xrightarrow{\rho} M \\
M \otimes N & \xrightarrow{\sigma} N \otimes N
\end{align*}
are clearly isometric.
\end{proof}

\subsection{Cyclic Modules}\label{cyclicmodulessection}
Let $R$ be a ring.
\begin{definition}
A cyclic $R$-module is an $R$-module $M$ together with an element $m \in M$ such that $Rm = M$. The distinguished element $m$ is called the cyclic vector.
\end{definition}
We will introduce cyclic modules as pairs $(M, m)$. If no confusion can arise, the cyclic vector will subsequently be omitted.
\begin{definition}
Let $(M, m)$ and $(N, n)$ be cyclic modules. A cyclic morphism $M \rightarrow N$ is a module morphism $M \rightarrow N$ which maps $m$ to $n$.
\end{definition}
The resulting category of cyclic modules over $R$ will be denoted by $Cyc(R)$.

Let $Ideals(R)$ be the partial order of ideals (submodules) of $R$, considered as a category. The following theorem follows immediately from the lattice isomorphism theorem for modules.
\begin{theorem}\label{cyclicmodules}
The functors $Ideals(R) \rightleftarrows Cyc(R)$ given by
\begin{align*}
I \subset R &\mapsto (R/I, [1]) \\
(M, m) &\mapsto Ann_{R}(m)
\end{align*}
constitute an equivalence of categories.
\end{theorem}
Here $Ann_{R}$ stands for the annihilator ideal over the ring $R$, and $[1] \in R/I$ is the class of the unit.
\begin{corollary}\label{coherencedestroyer}
Let $f : R \rightarrow S$ be a homomorphism of rings, and let $(M, m) \in Cyc(R)$ and $(N, n)$ be an $S$-module with chosen element $n \in N$. Then there is at most one homomorphism $M \rightarrow N$ over $f$ which maps $m$ to $n$.
\end{corollary}
\begin{proof}
The maps $M \rightarrow N$ over $f$ correspond to $R$-module maps $M \rightarrow f^{\ast}N$. The element $n \in f^{\ast}N$ is part of a cyclic submodule $Rn$. Since the canonical map over $f$, $f^{\ast}N \rightarrow N$ is (as a function of sets) the identity, the claim follows from theorem \ref{cyclicmodules}.
\end{proof}
The (external) tensor product of modules restricts to the category of cyclic modules.
\begin{proposition}\label{cyclictensorproducts}
Let $(M, m)$ be a cyclic $R$-module, and $(N, n)$ be cyclic $S$-module. Then $M \otimes N$ is cyclic over $R \otimes S$ with cyclic vector $m \otimes n$.
\end{proposition}
\begin{proof}
$R \otimes S (m \otimes n) \subseteq M \otimes N$ is a submodule containing all the simple tensors. Hence it is equal to $M \otimes N$.
\end{proof}
Here is a plentiful source of cyclic modules.
\begin{proposition}\label{hilbertproposition}
Let $V$ be a pre-Hilbert space. Then $V$ is a cyclic module for $\underline{End}(V)$, and any nonzero vector is a cyclic vector.
\end{proposition}
\begin{proof}
Let $v,w \in V$ be nonzero. We will show an adjointable map $V \rightarrow V$ mapping $v$ to $w$. Let $W = Span(v,w) \subseteq V$ be the subspace spanned by $v$ and $w$, and let $W^{\perp}$ be its orthogonal complement (which exists, since $W$ is finite dimensional). Then $W$ is a finite dimensional Hilbert space, and hence cyclic for $\underline{End}(W)$. Let $f \in \underline{End}(W)$ map $v$ to $w$. The map we are looking for is $f \oplus 1_{W^{\perp}}$. Its adjoint is $f^{\ast} \oplus 1_{W^{\perp}}$.
\end{proof}

\section{Construction of the GNS Representation Functor}\label{basicconstruction}
\subsection{Representable States}
Let $A$ be a $\ast$-algebra.

\begin{definition}
A linear map $\varphi : A \rightarrow \mathbb{C}$ is called a \emph{representable state} if there exists a $\ast$-module $M$ over $A$, with an element $m \in M$ such that
\begin{displaymath}
\varphi(a) = \langle am, m \rangle,
\end{displaymath}
for all $a \in A$.
\end{definition}
We will say that $M$ (or $m$) represents $\varphi$, or that $\varphi$ is a representable state on $A$. We require neither $\varphi$ nor $m$ to be normalized.

The annihilator of any cyclic vector representing $\varphi$ is determined by $\varphi$ itself. The specific formula for $Ann_{A}(m)$ given below is not important. What matters is that it is given in terms of $\varphi$ and not $M$.
\begin{proposition}\label{annihilatordetermined}
Let $(M, m)$ be a cyclic module representing $\varphi$. Then
\begin{displaymath}
Ann_{A}(m)  = \ker \beta,
\end{displaymath}
where $\beta : A \rightarrow Hom_{\bar{\mathbb{C}}}(A, \mathbb{C})$ is given by $a \mapsto \varphi((-)^{\ast}a) : A \rightarrow \mathbb{C}$, and $Hom_{\bar{\mathbb{C}}}$ is the functor of conjugate-linear maps.
\end{proposition}
\begin{proof}
This is obvious, but the following diagram chase easily adapts to any topos.

The module $M$ is cyclic, so we have an exact sequence
\begin{displaymath}
0 \longrightarrow Ann_{A}(m) \longrightarrow A \longrightarrow M \longrightarrow 0,
\end{displaymath}
with the projection $p : A \rightarrow M$ mapping $1$ to $m$. Apply $Hom_{\bar{\mathbb{C}}}(-, \mathbb{C})$ to that sequence, and construct the following diagram,
\begin{center}
\begin{tikzpicture}
\matrix (m) [matrix of math nodes, nodes in empty cells, row sep = 1.5 cm, column sep = 1.5 cm, text height = 1.5ex, text depth = .25ex]{
& & 0 & Hom_{\bar{\mathbb{C}}}(Ann_{A}(m), \mathbb{C}) \\
0 & \ker \beta & A & Hom_{\bar{\mathbb{C}}}(A, \mathbb{C}) \\
0 & Ann_{A}(m) & A & Hom_{\bar{\mathbb{C}}}(M, \mathbb{C}) \\
& 0 & 0 & 0 \\
};

\path[->] (m-2-1) edge (m-2-2)
		  (m-3-1) edge (m-3-2)
		  (m-4-2) edge (m-3-2)
		  (m-4-3) edge (m-3-3)
		  (m-4-4) edge (m-3-4)
		  (m-2-2) edge (m-2-3)
		  (m-2-3) edge node[auto] {$\beta$} (m-2-4)
		  (m-3-3) edge node[auto] {$\alpha$} (m-3-4)
		  (m-3-2) edge (m-3-3)
		  (m-3-4) edge node[auto] {$p^{\ast}$} (m-2-4)
		  (m-2-4) edge (m-1-4)
		  (m-3-1) edge (m-2-1)
		  (m-3-3) edge node[auto] {$id$} (m-2-3)
		  (m-3-2) edge node[auto] {$\gamma$} (m-2-2)
		  (m-2-3) edge (m-1-3);
\end{tikzpicture}
\end{center}
where $p^{\ast} = Hom_{\bar{\mathbb{C}}}(p, \mathbb{C})$, $\alpha$ is given by $a \mapsto \langle am, (-) \rangle_{M}$, and $\gamma$ results from the universal property of $\ker \beta$. The diagram commutes because
\begin{displaymath}
\beta(a) = \varphi((-)^{\ast}a) = \langle am , (-)m \rangle_{M} = p^{\ast}\alpha(a),
\end{displaymath}
by the definition of representability.

The rows and columns are exact. For the row including $\alpha$ this follows since the Hermitian form on $M$ is nondegenerate. For the column of $\gamma$ this follows because its composite with the inclusion $\ker \beta \rightarrow A$ is a monomorphism. The other cases are obvious.

By the four lemma of homological algebra, applied to the two middle rows, $\gamma$ is an isomorphism.
\end{proof}
\begin{remark}\label{reasonforadmissibility}
Note that the above proposition does not apply to non-cyclic modules. In fact if $(M, m)$ represents $\varphi$ then the cyclic module generated by $m$ may not be a $\ast$-module, because the Hermitian form on $Am$ inherited from $M$ may be degenerate.
\end{remark}

Representability has several useful characterizations.

\begin{theorem}\label{representabilityconditions}
Let $\varphi : A \rightarrow \mathbb{C}$ be a linear map. The following are equivalent:
\begin{enumerate}
\item There exists a cyclic $\ast$-module over $A$ which represents $\varphi$. This module is unique up to a unique cyclic isometry.
\item $\varphi$ is a representable state.
\item $\varphi$ is $\ast$-linear: $\varphi(a^{\ast}) = \overline{\varphi(a)}$.
\end{enumerate}
\end{theorem}
\begin{proof}
Clearly $1 \implies 2 \implies 3$. We prove $3 \implies 1$.

Uniqueness follows from proposition \ref{annihilatordetermined} and theorem \ref{cyclicmodules}: any two representing modules are uniquely isomorphic. These isomorphisms are unitary by cyclicity. In the diagram,
\begin{center}
\begin{tikzpicture}
\node (a) at (0,0) {$A$};
\node (b) at (-1.5, -1.5) {$M$};
\node (c) at (1.5, -1.5) {$M^{\prime}$};

\path[->] (a) edge (b)
		  (a) edge (c)
		  (b) edge (c);
\end{tikzpicture}
\end{center}
displaying the canonical cyclic isomorphism between two cyclic representations, the maps from $A$ (mapping $1$ to the cyclic vectors) are epimorphisms, and induce the same Hermitian form on $A$, namely
\begin{equation}\label{representationequation}
\langle a, b \rangle_{\varphi} = \varphi(b^{\ast}a),
\end{equation}
showing the horizontal map $M \rightarrow M^{\prime}$ must be isometric.

Existence follows from a variant of the GNS construction. Reconsider the bilinear form on the $A$-module $A$ given by equation \ref{representationequation}. It is Hermitian by $3$, but may be degenerate. To obtain nondegeneracy we divide $A$ by $A^{\perp}$, the radical of the Hermitian form $\langle -, - \rangle_{\varphi}$.

Left multiplication by $a$ on $A$ is adjointable with respect to $\langle -, - \rangle_{\varphi}$, with the adjoint being left multiplication by $a^{\ast}$. By lemma \ref{isotropiclemma} $A^{\perp}$ is a submodule of $A$.

Thus $A/A^{\perp}$ is an $A$-module. By construction, the Hermitian form $\langle -, - \rangle_{\varphi}$ factors through $A/A^{\perp}$:
\begin{displaymath}
\langle -, - \rangle_{\varphi} : A/A^{\perp} \times A/A^{\perp} \longrightarrow \mathbb{C},
\end{displaymath}
and is nondegenerate on $A/A^{\perp}$. Thus it gives $A/A^{\perp}$ the structure of a $\ast$-module, which clearly represents $\varphi$ through the cyclic vector $[1]$.
\end{proof}
\begin{remark}
Note that the norm of the cyclic vector satisfies $\| m \|^{2} = \varphi(1)$, so $\varphi$ must be defined, or at least uniquely definable, on a unital algebra, if we are to have any hope for uniqueness.
\end{remark}

\begin{definition}
The unique cyclic module representing $\varphi$ is called the GNS space associated to $\varphi$, and will be denoted by $GNS(\varphi)$. The cyclic vector representing $\varphi$ in $GNS(\varphi)$ will be denoted by $\Omega$, or $\Omega_{\varphi}$, if several different states are under consideration.
\end{definition}

The behavior of representable states under tensor products is predictable.

\begin{proposition}\label{representationoftensors}
Let $\varphi$ be a representable state on $A$ and $\psi$ a representable state on $B$. Then $\varphi \otimes \psi : A \otimes B \rightarrow \mathbb{C}$ is representable, and represented by $(M \otimes N, m \otimes n)$, for any representations $(M, m)$ and $(N, n)$ of $\varphi$ and $\psi$, respectively.
\end{proposition}
\begin{proof}
This follows immediately from the definition of the Hermitian form on $M \otimes N$, and the definition of representability.
\end{proof}

\begin{corollary}\label{gnsmonoidalproperty}
\begin{displaymath}
GNS(\varphi \otimes \psi) = GNS(\varphi) \otimes GNS(\psi)
\end{displaymath}
\end{corollary}
\begin{proof}
Immediate by propositions \ref{cyclictensorproducts} and \ref{representationoftensors}.
\end{proof}

We denote the set of representable states on $A$ by $\mathcal{S}_{r}(A)$. The following theorem establishes the functorial properties of $\mathcal{S}_{r}$, and the notion of pure and mixed states in our setting. By $\mathbf{Conv}_{\mathbb{C}}$ we denote the category of convex subsets of complex vector spaces and $\mathbb{C}$-affine maps between them.

\begin{theorem}
The construction $A \mapsto \mathcal{S}_{r}(A)$ is part of a functor $\mathcal{S}_{r} : \ast \mathbf{Alg}^{op} \rightarrow \mathbf{Conv}_{\mathbb{C}}$.
\end{theorem}
\begin{proof}
The dual space construction $A \mapsto A^{\ast}$ is a functor of the type we are looking for, and $\mathcal{S}_{r}(A) \subseteq A^{\ast}$, so we will construct our functor as a subfunctor of $(-)^{\ast}$.

We must check if this is well-defined, that is, if $\varphi \in \mathcal{S}_{r}(A)$, and $f : B \rightarrow A$ is a $\ast$-algebra map, then $f^{\ast}\varphi \in \mathcal{S}_{r}(B)$. But this is easy: if $(M, m)$ represents $\varphi$, then $(f^{\ast}M, m)$ represents $f^{\ast} \varphi$. Alternatively, it is trivial to check that $f^{\ast} \varphi$ is $\ast$-linear if $\varphi$ is.

What remains is to see that $\mathcal{S}_{r}(A)$ is a convex subset of $A^{\ast}$. So let $\varphi, \psi \in \mathcal{S}_{r}(A)$ be represented by $(M, m)$ and $(N, n)$ respectively. The state $t \varphi + (1-t) \psi$, for $t \in [0;1]$, is represented by
\begin{displaymath}
(M \oplus N, \sqrt{t}m + \sqrt{1-t} n),
\end{displaymath}
where $M \oplus N$ is the orthogonal direct sum of $M$ and $N$ (which is nondegenerate by proposition \ref{nondegeneratesums}).
\end{proof}

The category $\mathbf{Conv}_{\mathbb{C}}$ has finite products, and is therefore symmetric monoidal. We have also seen that $\ast \mathbf{Alg}$ is monoidal under the usual tensor product. The following natural transformations give $\mathcal{S}_{r}$ the structure of a lax monoidal functor $\ast \mathbf{Alg}^{op} \rightarrow \mathbf{Conv}_{\mathbb{C}}$.
\begin{align*}
\mathcal{S}_{r}(A) \times \mathcal{S}_{r}(B) & \longrightarrow \mathcal{S}_{r}(A \otimes B) \\
(\varphi, \psi) & \longmapsto \varphi \otimes \psi \\ \\
1 & \longrightarrow \mathcal{S}_{r}(\mathbb{C}) \\
\ast & \longmapsto id : \mathbb{C} \longrightarrow \mathbb{C}.
\end{align*}
Note that this structure is inherited from the natural structure on the dual space functor $(-)^{\ast}$. The verification of the following theorem is thus routine, and is omitted. 
\begin{theorem}\label{srismonoidal}
The above definitions make $\mathcal{S}_{r}$ into a symmetric lax monoidal functor.
\end{theorem}

\subsection{Positivity}
Let $A$ be a $\ast$-algebra.
\begin{definition}\hspace{1em}
\begin{itemize}
\item An element $b \in A$ is called positive if it is of the form $b = a^{\ast}a$ for some $a \in A$.
\item A linear map $A \rightarrow B$ of $\ast$-algebras is called positive if it maps positive elements to positive elements.
\item A linear map $A \rightarrow \mathbb{C}$ is called a positive state if it is positive and representable.
\end{itemize}
\end{definition}
Positive maps compose, and thus result in a category. Clearly $\ast$-homomorphisms are positive. Further examples will be given below.
\begin{proposition}
A state $\varphi$ is positive iff its GNS space is a pre-Hilbert space.
\end{proposition}
\begin{proof}
\begin{displaymath}
\langle a, a \rangle \geq 0 \iff \varphi(a^{\ast} a) \geq 0,
\end{displaymath}
since the left hand sides are equal. The result follows since the module under consideration is cyclic.
\end{proof}

\begin{corollary}\label{tensorproductsarepositive}
If $\varphi$ and $\psi$ are positive, then so is $\varphi \otimes \psi$.
\end{corollary}
\begin{proof}
By lemma \ref{flatnesslemma} the tensor product of pre-Hilbert spaces is a pre-Hilbert space.
\end{proof}

\begin{theorem}[Universality of the GNS Construction]\label{positiveuniversality}
Let $\varphi$ be a positive state. Then $GNS(\varphi)$ is initial among the pre-Hilbert $\ast$-modules representing $\varphi$.
\end{theorem}
\begin{proof}
Let $(M, m)$ be a representation of $\varphi$. Then $Am \subset M$ also represents $\varphi$, since it is obviously a $\ast$-submodule of $M$ (unlike in the indefinite case, cf.~remark \ref{submoduletrouble}), and is cyclic. Thus by theorem \ref{cyclicmodules} and proposition \ref{annihilatordetermined} there is a unique map
\begin{displaymath}
GNS(\varphi) \longrightarrow Am \hookrightarrow M
\end{displaymath} 
mapping $\Omega$ to $m$.
\end{proof}
In light of this theorem the classical GNS result can be restated as ``positive linear functionals on a $C^{\ast}$-algebra are representable''.

Let $\mathcal{S}_{p}(A)$ be the set of positive states on $A$.
\begin{theorem}\label{spismonoidal}
$\mathcal{S}_{p} \subseteq \mathcal{S}_{r}$ is a symmetric monoidal subfunctor.
\end{theorem}
\begin{proof}
The pullback of a positive state is positive, since maps of $\ast$-algebras are positive. The set $\mathcal{S}_{p}(A)$ is also obviously convex, since $\mathbb{R}_{\geq 0} \subseteq \mathbb{R}$ is convex. By corollary \ref{tensorproductsarepositive}, and the obvious fact that $id : \mathbb{C} \rightarrow \mathbb{C}$ is a positive state, the monoidal structure can be inherited from $\mathcal{S}_{r}$.
\end{proof}

The following lemma connects us to the more traditional versions of the GNS construction, and is needed for representing maps of positive states.
\begin{lemma}\label{radicalisnullspace}
Let $\varphi : A \rightarrow \mathbb{C}$ be a positive state. Then for the induced Hermitian form on $A$, we have $A^{\perp} = \{a \in A : \langle a, a \rangle = 0\}$.
\end{lemma}
\begin{proof}
Clearly $A^{\perp} \subseteq \{a \in A : \langle a, a \rangle = 0\}$. To see the other inclusion recall the general Cauchy-Schwartz inequality (or its proof), which is still valid in our setting: $|\langle a, b \rangle|^{2} \leq \langle a, a \rangle \langle b, b \rangle$, for any $a, b \in A$. Thus if $\langle a, a \rangle = 0$, then $\langle a, b \rangle = 0$ for any $b \in A$.
\end{proof}
\subsubsection{Complete Positivity}
In this section we recover a variant of the Stinespring factorization theorem.
\begin{definition}
$M_{n}(-) = (-) \otimes M_{n}(\mathbb{C}) : \ast\mathbf{Alg} \rightarrow \ast\mathbf{Alg}$.
\end{definition}
\begin{definition}
A linear map $\Phi: A \rightarrow B$ between $\ast$-algebras is completely positive if it is $\ast$-linear and $M_{n}(\Phi)$ is positive for all $n \in \mathbb{N}$.
\end{definition}
In the setting of $C^{\ast}$-algebras $\ast$-linearity is a consequence of ordinary positivity. In our case we list it as a separate requirement. Clearly, completely positive maps form a category which includes the $\ast$-homomorphisms.

Now let $\Phi : A \rightarrow B$ be completely positive, and let $\varphi : B \rightarrow \mathbb{C}$ be a positive state on $B$. Set $H = A \otimes B$, and let
\begin{align*}
V : B & \longrightarrow H &\textnormal{be given by} &b \longmapsto 1_{A} \otimes b \\
V^{\ast} : H & \longrightarrow B &\textnormal{be given by} &a\otimes b \longmapsto \Phi(a)b \\
\pi(a) : H & \longrightarrow H &\textnormal{be given by} &a^{\prime}\otimes b \longmapsto aa^{\prime} \otimes b.
\end{align*}
Declare $\pi(a)^{\ast} = \pi(a^{\ast})$, and finally define a bilinear form on $H$ by
\begin{displaymath}
\langle a_{1} \otimes b_{1}, a_{2} \otimes b_{2} \rangle_{H} = \langle \Phi(a_{2}^{\ast} a_{1})b_{1}, b_{2}\rangle_{\varphi}.
\end{displaymath}
By inspection, $\pi(a)$ and $\pi(a)^{\ast}$ are adjoint with respect to the Hermitian form on $H$ (which may be degenerate), and $\pi$ defines an $A$-module structure on $H$. By the $\ast$-linearity of $\Phi$, $V$ and $V^{\ast}$ are also adjoint, with $B$ endowed with the form $\langle a, b\rangle_{\varphi} = \varphi(b^{\ast}a)$. By construction we have
\begin{displaymath}
\Phi(a) = V^{\ast} \pi(a) V(1_{B}).
\end{displaymath}
The form $\langle -, - \rangle_{H}$ is positive semi-definite by the complete positivity of $\Phi$ and the positivity of $\varphi$. Indeed, for any $a_{1}, \ldots a_{n} \in A$ we have $[a_{i}^{\ast} a_{j}] \in M_{n}(A)$, a positive element, equal to $X^{\ast} X$, where $X \in M_{n}(A)$ is the matrix with first row $(a_{i})$, and the rest $0$. This means that $M_{n}(\Phi)([a_{i}^{\ast}a_{j}])$ is positive, hence -- by our definition of positivity -- of the form $L^{\ast} L$, for some $L \in M_{n}(B)$, and so
\begin{displaymath}
\langle \sum_{j} a_{j}\otimes x_{j} , \sum_{i} a_{i} \otimes x_{i} \rangle_{H} = \langle M_{n}(\Phi)([a_{i}^{\ast}a_{j}]) x, x\rangle_{B^{n}} = \langle Lx, Lx \rangle_{B^{n}} \geq 0,
\end{displaymath}
where $x = (x_{1}, \ldots, x_{n}) \in B^{n}$, and $B^{n}$ is the $n$-fold orthogonal sum of $(B, \langle -, - \rangle_{\varphi})$.

We are now ready to state the factorization theorem. Let $\Phi : A \rightarrow B$ be completely positive, and let $\varphi : B \rightarrow \mathbb{C}$ be a positive state, and let $i : B \rightarrow \underline{End}(GNS(\varphi))$ be its GNS representation.
\begin{theorem}[Stinespring Factorization Theorem]\label{stinepring}
There exists a pre-Hilbert $A$-module $H$ and an adjointable linear map $V : GNS(\varphi) \rightarrow H$ such that $V^{\ast} \pi V = i \Phi$, where $\pi$ is the representation of $A$ on $H$.
\end{theorem}
\begin{proof}
Factor out all the degeneracy in the above formulas. Lemma \ref{isotropiclemma} ensures everything remains well-defined.
\end{proof}
\begin{remark}
Conversely, one can easily compute that all maps of the form $V^{\ast}\pi V$, where $V$ is any adjointable map between pre-Hilbert spaces, are completely positive according to our definition.
\end{remark}
One can replace $B$ with an arbitrary pre-Hilbert $\ast$-module $L$ over $B$, but no real generality is gained.
\begin{corollary}\label{stinespringcorollary}
Let $L$ be a pre-Hilbert $\ast$-module over $B$, and let $\Phi : A \rightarrow B$ be completely positive. Then there exists a pre-Hilbert $\ast$-module $H$ over $A$, and an adjointable linear map $V : L \rightarrow H$ such that $V^{\ast} \pi V = i \Phi$, where $\pi$ is the representation of $A$ on $H$, and $i$ is the representation of $B$ on $L$.
\end{corollary}
\begin{proof}
Apply the previous theorem to the composite $i\Phi$, and note that by proposition \ref{hilbertproposition} and theorem \ref{representabilityconditions} we have $L = GNS(\varphi)$, for $\varphi : \underline{End}(L) \rightarrow \mathbb{C}$ given by $\varphi(f) = \langle f(v) , v \rangle_{L}$, for any choice of nonzero $v \in L$.
\end{proof}
\subsection{Categories of Physical Processes}
Let $1$ be the terminal category, and $1 \rightarrow \mathbf{Conv}_{\mathbb{C}}$ the functor which picks out the affine point.
\begin{definition}\hspace{1em}
\begin{itemize}
\item The unrestricted category of physical processes is the comma category $1 \downarrow \mathcal{S}_{r}$. It will be denoted by $\mathbf{Phys}_{r}$.
\item The category of positive physical processes is $1 \downarrow \mathcal{S}_{p}$. It will be called $\mathbf{Phys}_{p}$.
\item The category of physical processes (just so), $\mathbf{Phys}$, will be constructed below in definition \ref{definitionofphysa}, after the introduction of admissible morphisms.
\end{itemize}
\end{definition}
$\mathbf{Phys}_{r}$ is strong symmetric monoidal by theorem \ref{srismonoidal}, and purely formal properties of forming comma categories. The others are monoidal subcategories, with $\mathbf{Phys}_{p}$ being such by theorem \ref{spismonoidal}. For convenience, we will spell out the details of $\mathbf{Phys}_{r}$.

The objects of $\mathbf{Phys}_{r}$ are pairs $(A, \varphi)$, with $A$ a $\ast$-algebra, and $\varphi : A \rightarrow \mathbb{C}$ a representable state on $A$. A morphism
\begin{displaymath}
(A, \varphi) \longrightarrow (B, \psi)
\end{displaymath}
in $\mathbf{Phys}_{r}$ is a $\ast$-algebra homomorphism $f : B \rightarrow A$ such that $\psi = f^{\ast} \varphi = \varphi \circ f$.

As in the introduction, we will write $f : \varphi \rightarrow \psi$ for morphisms in $\mathbf{Phys}_{r}$, omitting the algebras. They can be recovered by applying the observables functor
\begin{displaymath}
\mathcal{O} : \mathbf{Phys}_{r} \longrightarrow \ast\mathbf{Alg}^{op},
\end{displaymath}
which is simply forgetting the state: $(A, \varphi) \mapsto A$.

The monoidal structure is defined by
\begin{displaymath}
(A, \varphi) \otimes (B, \psi) = (A \otimes B, \varphi \otimes \psi),
\end{displaymath}
with the obvious formula for morphisms.

Examples of physical processes abound. A vast supply of objects and morphisms will be constructed in theorem \ref{schrodingerpicture}, where it is shown how to lift Schr{\"o}dinger picture operators, observables, and states to $\mathbf{Phys}_{p}$. Using that theorem all $W^{\ast}$- or $C^{\ast}$-dynamical systems (with invertible dynamics) can be lifted into our formalism.

\subsection{Representations of Physical Processes}\label{representablemaps}
\subsubsection{Construction for Positive States}\label{constructionforpositivestates}
We will now construct a symmetric monoidal functor
\begin{displaymath}
\mathbf{Phys}_{p}^{op} \longrightarrow \ast \mathbf{Mod},
\end{displaymath}
whose object function is given by $\varphi \mapsto GNS(\varphi)$. It will serve as a foundation for our formalization of physics.

The construction, outlined in the introduction, follows immediately from theorem \ref{positiveuniversality}. Let $f : \varphi \rightarrow \psi$ be a morphism in $\mathbf{Phys}_{p}$. Then $\mathcal{O}(f)^{\ast} GNS(\varphi)$ represents $\psi$ and so we have a map
\begin{displaymath}
GNS(\psi) \longrightarrow \mathcal{O}(f)^{\ast} GNS(\varphi).
\end{displaymath}
We define $GNS(f)$ to be the composite of this map with the cartesian lift of $\mathcal{O}(f)$:
\begin{displaymath}
GNS(\psi) \longrightarrow \mathcal{O}(f)^{\ast} GNS(\varphi) \longrightarrow GNS(\varphi).
\end{displaymath}
Note that $GNS(f)$ lies over $f$, making $GNS$ fibered over $\ast\mathbf{Alg}$.

The fact that this construction defines a functor, which is furthermore strong symmetric monoidal in a natural way, follows from theorem \ref{positiveuniversality} and corollary \ref{coherencedestroyer}, applied repeatedly to every condition we have to check. The structures we must exhibit are uniquely specified by appeals to theorem \ref{positiveuniversality}, and any coherence laws are satisfied by corollary \ref{coherencedestroyer}. Since we will perform the construction in more generality, we leave the details to the reader.

Without positivity we have no analog of universality, $\mathcal{O}(f)^{\ast}GNS(\varphi)$ does not need to contain a cyclic module representing $\psi$, and so we must restrict the maps we can represent. This leads to the notion of admissibility.
\subsubsection{Construction in General}\label{representationofadmissiblemaps}
Now we consider a map $f : \varphi \rightarrow \psi$ in $\mathbf{Phys}_{r}$, with $\varphi$ not necessarily positive. To ease notation, write $\mathcal{O}(f) = f : A \rightarrow B$, with $A = \mathcal{O}(\psi)$ and $B = \mathcal{O}(\varphi)$. We also abbreviate $f^{\ast} = \mathcal{O}(f)^{\ast}$.

Recall that $GNS(\varphi) = B/B^{\perp}$ and $GNS(\psi) = A/A^{\perp}$, with $(-)^{\perp}$ denoting the radical of the induced Hermitian form. We wish to define a map $GNS(f) : GNS(\psi) \rightarrow GNS(\varphi)$, but so far we only have the following diagram.
\begin{center}
\begin{tikzpicture}
\matrix (m)[matrix of math nodes, row sep = 2cm, column sep = 2cm, text height = 1.5ex, text depth = .25ex]{
A/f^{-1}(B^{\perp}) & GNS(\psi) \\
GNS(\varphi)\\
};

\path[->] (m-1-1) edge node[auto] {$\pi$} (m-1-2)
		  (m-1-1) edge node[auto,swap] {$[f]$} (m-2-1)
		  (m-1-2) edge[densely dashed] node[auto] {?} (m-2-1);
\end{tikzpicture}
\end{center}
The map $[f]$ is a morphism of cyclic modules over $f$, and is given by $[x] \mapsto [f(x)]$. The horizontal map $\pi$ is a quotient projection (since we clearly have $f^{-1}(B^{\perp}) \subseteq A^{\perp}$, by the definition of $\psi$). To fill in the dashed map, we simply assume that $\pi$ is an isomorphism, leading to the following definition.
\begin{definition}\label{admissibilitydefinition}
The map $f : A \rightarrow B$ is called admissible for $\varphi$ if $A^{\perp} \subseteq f^{-1}(B^{\perp})$.
\end{definition}
\begin{proposition}\label{admissibilityproperties}
The following are equivalent:
\begin{enumerate}
\item $f$ is admissible for $\varphi$
\item $\pi$ is an isomorphism
\item The Hermitian form defined by $\psi$ on $A/f^{-1}(B^{\perp})$ is nondegenerate
\item $A/f^{-1}(B^{\perp})$ represents $\psi$
\item $f^{\ast}GNS(\varphi)$ contains a cyclic module representing $\psi$.
\end{enumerate}
\end{proposition}
\begin{proof}
The implications $1 \implies 2 \implies 3 \implies 4$ are trivial. We have $4 \implies 5$ since the image of $[f]$ in $f^{\ast}GNS(\varphi)$ is the sought after module.

Finally, we show $5 \implies 1$ as follows. By 5 and theorem \ref{representabilityconditions}(1), there is a cyclic map $GNS(\psi) \rightarrow f^{\ast}GNS(\varphi)$, and hence a cyclic map $GNS(\psi) \rightarrow GNS(\varphi)$ over $f$. But, by the construction of $GNS$ spaces, this must be a cyclic map $A/A^{\perp} \rightarrow B/B^{\perp}$ over $f$. Thus, by cyclicity, $x \in A^{\perp}$ implies $f(x) \in B^{\perp}$, which is 1.
\end{proof}

\begin{remark}\label{uniquenessremark}
Note that any cyclic map completing the triangle above will make it commute (by proposition \ref{coherencedestroyer}). This is implicit in the proof of the last implication above. Consequently, the content of proposition \ref{admissibilityproperties} is that there is only one reasonable formula for $GNS(f)$, i.e.~$[f]$, and it gives a well-defined map iff $f$ is admissible.
\end{remark}
\begin{definition}
Let $f : \varphi \rightarrow \psi$ be a morphism in $\mathbf{Phys}_{r}$, such that $\mathcal{O}(f)$ is admissible for $\varphi$. Then the GNS representation of $f$ is defined to be
\begin{align*}
GNS(f) : GNS(\psi) &\longrightarrow GNS(\varphi) \\
GNS(f)([x]) &= [f(x)].
\end{align*}
\end{definition}

Admissible homomorphisms have all the categorical properties we require.
\begin{proposition}\label{admissibilitypropertiestwo}\hspace{1em}
\begin{enumerate}
\item Composites of admissible maps are admissible
\item The tensor product of admissible maps is admissible
\item All maps between positive states are admissible
\end{enumerate}
\end{proposition}
\begin{proof}
1. is obvious by direct computation. 3. is obvious by proposition \ref{radicalisnullspace}.

To see 2. note that, since all modules over $\mathbb{C}$ are flat (being free), the tensor product of nondegenerate forms is nondegenerate. Consider a tensor of admissible maps $f \otimes f^{\prime} : \varphi \otimes \varphi^{\prime} \rightarrow \psi \otimes \psi^{\prime}$, over homomorphisms $f, f'$ of $\ast$-algebras, and compute $(f \otimes f^{\prime})^{\ast} GNS(\varphi \otimes \varphi^{\prime}) = f^{\ast}GNS(\varphi) \otimes (f^{\prime})^{\ast} GNS(\varphi^{\prime})$, using proposition \ref{gnsmonoidalproperty}. Both factors of the product contain cyclic modules representing $\psi$ and $\psi^{\prime}$, respectively, by proposition \ref{admissibilityproperties}(5). Therefore their tensor product -- a cyclic submodule of $(f \otimes f^{\prime})^{\ast} GNS(\varphi \otimes \varphi^{\prime})$ represents $\psi \otimes \psi^{\prime}$. So $f \otimes f^{\prime}$ is admissible by proposition \ref{admissibilityproperties}(5).
\end{proof}
\begin{remark}
If $f : \varphi \rightarrow \psi$ is a map in $\mathbf{Phys}_{r}$, and $\varphi$ is positive, then $\psi$ is as well. This makes proposition \ref{admissibilitypropertiestwo}(3) easier to apply.
\end{remark}
\begin{definition}\label{definitionofphysa}
We denote by $\mathbf{Phys} = \mathbf{Phys}_{a}$ the symmetric monoidal subcategory of $\mathbf{Phys}_{r}$ spanned by the admissible morphisms.
\end{definition}
$\mathbf{Phys}$ is well-defined by proposition \ref{admissibilitypropertiestwo}. Note that $\mathbf{Phys}_{p} \subseteq \mathbf{Phys}$ by proposition \ref{admissibilitypropertiestwo}(3).

\begin{theorem}\label{gnsconstruction}
The constructions
\begin{align*}
\varphi &\longmapsto GNS(\varphi)\\
f &\longmapsto GNS(f),
\end{align*}
for objects $\varphi \in \mathbf{Phys}$, and morphisms $f : \varphi \rightarrow \psi$ in $\mathbf{Phys}$, are part of a strong symmetric monoidal functor
\begin{displaymath}
GNS : \mathbf{Phys} \longrightarrow \ast \mathbf{Mod},
\end{displaymath}
fibered over $\ast\mathbf{Alg}$.
\end{theorem}
\begin{proof}
$GNS(f)$ is always cyclic, and hence preserves composition by corollary \ref{coherencedestroyer}. It is strong symmetric monoidal by corollary \ref{gnsmonoidalproperty}. All coherence diagrams commute by corollary \ref{coherencedestroyer}, since all the morphisms involved in these diagrams are obviously cyclic.

$GNS(f)$ is fibered over $\ast\mathbf{Alg}$ by its explicit construction.
\end{proof}

\subsubsection{The Covariant Representation}
Let $\ast \mathbf{Mod}_{p}$ denote the category pre-Hilbert $\ast$-modules, and $\ast\mathbf{Mod}_{adj}$ the category of Hilbert $\ast$-modules, with algebras acting by closable maps.
\begin{definition}
The covariant $GNS$ construction is the composite
\begin{center}
\begin{tikzpicture}
\node (a) at (0,0) {$\mathbf{Phys}_{p}$};
\node (b) at (4,0) {$\ast\mathbf{Mod}_{p}^{op}$};
\node (c) at (11,0) {$\ast\mathbf{Mod}_{adj}$};

\path[->] (a) edge node[auto] {$GNS^{op}$} (b)
		  (b) edge node[auto] {completion + adjoint} (c);
\end{tikzpicture}
\end{center}
It will be denoted by $GNS_{c}$.
\end{definition}
Thus $GNS_{c}(f)$ acts as the adjoint of $GNS(f)$ on the completion of the appropriate pre-Hilbert spaces. The existence of adjoints requires completeness, so we use it out of necessity. Topology does not internalize well, so this construction cannot reasonably be repeated in a topos (unlike its contravariant cousin, see section \ref{sins}). Despite this, it is the ``correct'' version for physical applications, as is evident in theorems \ref{dinaturality}, \ref{schrodingerpicture}, \ref{unitaryrepresentations} and section \ref{wavefunctioncollapsesection}.
\begin{theorem}\label{covariantgnsproperty}
$GNS_{c}$ is a symmetric monoidal functor $\mathbf{Phys}_{p} \rightarrow \ast\mathbf{Mod}_{adj}$ 
\end{theorem}
\begin{proof}
By definition, $GNS_{c}$ is a composite of such.
\end{proof}
There is a more topological variant of this definition, the details of which we leave to the reader. Define a monoidal subfunctor $\mathcal{S}_{b} \subseteq \mathcal{S}_{p}$ consisting of those states $\varphi \in \mathcal{S}_{p}(A)$, for which $A$ acts by bounded operators on $GNS(\varphi)$. This can be expressed using only $\varphi$. Then the covariant representation can be defined on $\mathbf{Phys}_{b} = 1 \downarrow \mathcal{S}_{b}$, with codomain the ordinary Hilbert modules -- with algebras acting by bounded, not just closable maps. Note that, since $GNS(f)$ is always isometric, $GNS_{c}(f)$ will always be coisometric, and hence bounded.

\section{Computations and Examples}\label{examplesection}
\subsection{Dinaturality}
Let $f : \varphi \rightarrow \psi$ be a morphism in $\mathbf{Phys}$. We wish to gain a preliminary understanding of the map 
\begin{displaymath}
GNS(f) : GNS(\psi) \longrightarrow GNS(\varphi).
\end{displaymath}
To facilitate this comparison, we will make use of the natural map, which maps vectors in the GNS space to the obvious states which they represent:
\begin{align*}
GNS(\varphi) & \longrightarrow \mathcal{S}_{r}(\mathcal{O}(\varphi)) \\
v & \longmapsto s_{\varphi}(v) \\
s_{\varphi}(v) = a &\longmapsto \langle a v, v \rangle_{GNS(\varphi)}.
\end{align*}
The representability of $s_{\varphi}(v)$ is guaranteed by theorem \ref{representabilityconditions}(3). The maps $s_{\varphi}$ constitute a dinatural transformation \cite[IX.4]{maclane}.
\begin{theorem}\label{dinaturality}
Let $U : \ast \mathbf{Mod} \rightarrow \mathbf{Set}$ map each module to its underlying set of elements, and let $\mathcal{S}_{r} : \ast\mathbf{Alg}^{op} \rightarrow \mathbf{Set}$ map every algebra to its set of representable states. Then $s : U \circ GNS \rightarrow \mathcal{S}_{r} \circ \mathcal{O}$ is a dinatural transformation, meaning the following diagram commutes:
\begin{center}
\begin{tikzpicture}
\matrix (m) [matrix of math nodes, row sep = 1cm, column sep = 3cm, text height = 1.5ex, text depth = .25ex]{
GNS(\psi) & GNS(\varphi) \\
\mathcal{S}_{r}(\mathcal{O}(\varphi)) & \mathcal{S}_{r}(\mathcal{O}(\psi)), \\
};

\path[->] (m-1-1) edge node[auto] {$GNS(f)$} (m-1-2)
		  (m-1-1) edge node[auto, swap] {$s_{\psi}$} (m-2-1)
		  (m-2-2) edge node[auto, swap] {$\mathcal{S}_{r}(\mathcal{O}(f))$} (m-2-1)
		  (m-1-2) edge node[auto] {$s_{\varphi}$} (m-2-2);
\end{tikzpicture}
\end{center}
for every morphism $f : \varphi \rightarrow \psi$ in $\mathbf{Phys}$.
\end{theorem}
\begin{proof}
Let $v \in GNS(\psi)$. Since the GNS space is cyclic, there is an element $x \in A$ such that $x \Omega = v$. Thus $v$ represents the state $s_{\psi}(v)$ given by
\begin{displaymath}
a \mapsto \langle av, v \rangle_{GNS(\psi)} =  \psi(x^{\ast} a x) = \varphi(f(x)^{\ast} f(a) f(x)).
\end{displaymath}
To calculate $GNS(f)(v)$ we look at its explicit construction and find that
\begin{displaymath}
GNS(f)(v) = [f(x)] \in GNS(\varphi),
\end{displaymath}
and so $w = GNS(f)(v)$ represents the state $s_{\varphi}(w)$ given by
\begin{displaymath}
b \mapsto \langle b w, w \rangle_{GNS(\varphi)} = \varphi(f(x)^{\ast} b f(x)),
\end{displaymath}
whose pullback by $f$ is clearly $s_{\psi}(v)$.
\end{proof}
Thus $GNS(f)$ acts essentially as $\mathcal{S}_{r}\mathcal{O}(f)^{-1}$ on \emph{presentations} of states, which are \emph{presented} in such a way as to make this operation well-defined. Recklessly abusing notation, writing $f_{\ast} = GNS(f)$ and $f^{\ast} = \mathcal{S}_{r}\mathcal{O}(f)$, we can say
\begin{displaymath}
v = f^{\ast} f_{\ast} v.
\end{displaymath}
\begin{corollary}\label{dinaturalitycorollary}
Let $\mathcal{C} \subset \mathbf{Phys}$ be the category of those $f$ for which $GNS(f)$ is unitary. Then $s : U \circ GNS_{c} \rightarrow S_{r} \circ \mathcal{O}$ is a natural transformation of functors on $\mathcal{C}$.
\end{corollary}
\begin{proof}
$GNS(f)$ is unitary iff it's invertible, and then $GNS_{c}(f) = GNS(f)^{-1}$. We can substitute this inverse into the dinaturality square above, obtaining a naturality square.
\end{proof}

\subsection{Antiunitary Processes}\label{antiunitaryprocesses}
Let $\mathbf{Vect}_{\mathbb{C}}$ be the category of complex vector spaces and linear maps between them. To accommodate antiunitary processes, such as time reversal \cite{roberts}, we will require the following device.
\begin{definition}
Let $V$ be a complex vector space. Its conjugate, $\overline{V}$, is defined by the universal property
\begin{displaymath}
\mathbf{Vect}_{\mathbb{C}}(\overline{V}, W) = \{\textnormal{Conjugate-linear maps } V \longrightarrow W\},
\end{displaymath}
for any complex vector space $W$.
\end{definition}
One easily proves that $\overline{V}$ exists, by direct construction. The sets underlying $V$ and $\overline{V}$ can be taken to coincide, and we will do so.
\begin{remark}\label{generalconjugationremark}
One can play this game for any endomorphism of any ring extension, not just complex conjugation on $\mathbb{C}/\mathbb{R}$.
\end{remark}
Since conjugation is the only automorphism of $\mathbb{C}$ over $\mathbb{R}$, we will abbreviate conjugate-linear to antilinear. The formal properties of vector space conjugation assemble into the following theorem
\begin{theorem}\label{conjugationtheorem}
Conjugation defines a symmetric monoidal, conjugate-closed, $\mathbf{Vect}_{\mathbb{C}}$-enriched involution
\begin{displaymath}
\overline{(-)} : \overline{\mathbf{Vect}_{\mathbb{C}}} \longrightarrow \mathbf{Vect}_{\mathbb{C}}.
\end{displaymath}
\end{theorem}
\begin{proof}
This is all trivial, as long as the terms are understood. We merely explain their meaning.

Since $\overline{V}$ is defined by a universal property, its existence automatically defines a functor
\begin{displaymath}
\mathbf{Vect}_{\mathbb{C}} \longrightarrow \mathbf{Vect}_{\mathbb{C}},
\end{displaymath}
with object function $V \mapsto \overline{V}$.

$\mathbb{R}$-bilinear forms can be antilinear (in both variables), and such forms are clearly represented by both $\overline{V \otimes W}$ and $\overline{V} \otimes \overline{W}$. Thus we have $\overline{V} \otimes \overline{W} = \overline{V \otimes W}$, making conjugation into a strong symmetric monoidal functor.

Since $\mathbf{Vect}_{\mathbb{C}}$ is monoidal closed, we can ask if conjugation is a closed functor. It's not, but the natural maps
\begin{align*}
\mathbf{Vect}_{\mathbb{C}}(V, W) &\longrightarrow \mathbf{Vect}_{\mathbb{C}}(\overline{V}, \overline{W}) \\
f &\longmapsto \overline{f}
\end{align*}
are antilinear, thus defining isomorphisms
\begin{displaymath}
\overline{\mathbf{Vect}_{\mathbb{C}}(V, W)} \rightarrow \mathbf{Vect}_{\mathbb{C}}(\overline{V}, \overline{W}).
\end{displaymath}
This is the meaning of conjugate-closed.

Since $\mathbf{Vect}_{\mathbb{C}}$ is symmetric monoidal closed, it is self-enriched, and since conjugation is symmetric monoidal we can extend the action of conjugation to $\mathbf{Vect}_{\mathbb{C}}$-enriched categories, such as $\mathbf{Vect}_{\mathbb{C}}$ itself, resulting in $\overline{\mathbf{Vect}_{\mathbb{C}}}$. Then conjugation is an enriched functor, as displayed in the statement of the theorem.

Such functors are rightfully called conjugate-enriched, and can be composed, just like contravariant functors. Conjugation thus understood is involutive (up to coherent natural isomorphism), since
\begin{displaymath}
\overline{\overline{V}} = V,
\end{displaymath}
which follows from the fact that an anti-antilinear map is just linear, since conjugation (on $\mathbb{C}$) is an involution.
\end{proof}
\begin{remark}
The last step of the proof shows the usefulness of the general perspective of remark \ref{generalconjugationremark}, utilizing the composition of $\sigma$- and $\rho$-linearity to $(\sigma \circ \rho)$-linearity.
\end{remark}
\begin{remark}\label{involutivityremark}
Due to the involutivity, $\overline{V}$ also represents antilinear maps into $V$.
\end{remark}

Theorem \ref{conjugationtheorem} allows us to conjugate essentially anything, in particular $\ast$-algebras and their modules. Extreme care must be taken, however, to distinguish conjugation of vector spaces and their maps and the function of complex conjugation on $\mathbb{C}$. Failure to do so will result in catastrophic error -- object types will stop matching.

As an example, let us conjugate a Hermitian form. The conjugate of
\begin{displaymath}
H \otimes \overline{H} \longrightarrow \mathbb{C}
\end{displaymath}
is
\begin{displaymath}
\overline{H} \otimes H \longrightarrow \overline{\mathbb{C}},
\end{displaymath}
which is not a Hermitian form, because $\overline{\mathbb{C}} \neq \mathbb{C}$ (even though they are canonically isomorphic). We correct this by composing with complex conjugation (the \emph{linear} function):
\begin{displaymath}
\overline{H} \otimes H \longrightarrow \overline{\mathbb{C}} \xrightarrow{\sigma} \mathbb{C}.
\end{displaymath}
The end result of this operation can be given by the following explicit formula:
\begin{displaymath}
\langle v, w \rangle_{\overline{H}} = \langle w, v \rangle_{H}.
\end{displaymath}

Conjugation of $\ast$-algebras presents no difficulties. Note that $\ast$ remains antilinear, by remark \ref{involutivityremark}. Moving on to $\ast$-modules consider a $\ast$-representation of $A$ on $H$, given by a $\ast$-homomorphism
\begin{displaymath}
A \longrightarrow \underline{End}(H),
\end{displaymath}
to the adjointable maps on $H$. We compute the conjugate of this representation:
\begin{displaymath}
\overline{A \longrightarrow \underline{End}(H)}
\end{displaymath}
as
\begin{displaymath}
\overline{A} \longrightarrow \underline{End}(\overline{H}),
\end{displaymath}
and note that conjugation maps adjointable maps in $H$ to adjointable maps in $\overline{H}$ (with respect to the conjugate form constructed above). This allows us to state the following proposition.
\begin{proposition}
Let $(H,v)$ represent $\varphi : A \rightarrow \mathbb{C}$. Then $(\overline{H}, v)$ represents $\overline{\varphi} : \overline{A} \rightarrow \mathbb{C}$
\end{proposition}
\begin{proof}
Calculate carefully. Note that $\overline{\varphi}$ is implicitly post-composed with conjugation, to make it a state on $\overline{A}$. For $a \in \overline{A}$ we have:
\begin{displaymath}
\langle av, v \rangle_{\overline{H}} = \langle v, av \rangle_{H} = \overline{\varphi(a)} = \overline{\varphi}(a).
\end{displaymath}
\end{proof}
\begin{remark}
Formally, one should write $\overline{a}$ for the action of $a \in \overline{A}$ on $\overline{H}$.
\end{remark}
\begin{corollary}\label{conjugaterepresentability}
$GNS(\overline{\varphi}) = \overline{GNS(\varphi)}$
\end{corollary}
\begin{proof}
Theorem \ref{conjugationtheorem} says that all categorically expressible algebra is preserved by conjugation. So homomorphisms, cyclicity of modules and maps, and the like are preserved. Hence this is immediate by theorem \ref{representabilityconditions}(1) and the preceding proposition.
\end{proof}

The following theorem sets up the proper definition of antilinear processes. Let $\overline{(-)}$ denote conjugation appropriate to the objects it's applied to (theorem \ref{conjugationtheorem} gives meaning to all legitimate instances of this operation). Using this we can state the relation between conjugation and the $GNS$ representation.
\begin{theorem}
The following diagram of symmetric monoidal functors fibered over $\overline{(-)} : \ast\mathbf{Alg} \rightarrow \ast\mathbf{Alg}$ commutes
\begin{center}
\begin{tikzpicture}
\matrix (m) [matrix of math nodes, row sep = 1.5 cm, column sep = 2 cm, text height = 1.5ex, text depth = .25ex]{
\mathbf{Phys}^{op} & \ast\mathbf{Mod} \\
\mathbf{Phys}^{op} & \ast\mathbf{Mod} \\
};

\path[->] (m-1-1) edge node[auto] {$GNS$} (m-1-2)
		  (m-1-1) edge node[auto,swap] {$\overline{(-)}$} (m-2-1)
		  (m-2-1) edge node[auto] {$GNS$} (m-2-2)
		  (m-1-2) edge node[auto] {$\overline{(-)}$} (m-2-2);
\end{tikzpicture}
\end{center}
\end{theorem}
\begin{proof}
This, similarly to theorem \ref{gnsconstruction}, follows immediately from theorem \ref{representabilityconditions}(1) and corollaries \ref{conjugaterepresentability} and \ref{coherencedestroyer}.
\end{proof}
\begin{definition}
An antilinear process $\varphi \rightarrow \psi$ in $\mathbf{Phys}$ is defined as a map $\overline{\varphi} \rightarrow \psi$.
\end{definition}
By the preceding theorem, such processes are represented by antilinear isometries, as expected. Note that $\mathcal{O}(\overline{\varphi}) = \overline{\mathcal{O}(\varphi)}$, so that the observables are formally changed by conjugation.

\subsection{Normalization}\label{normalizationsection}
In this section we mitigate the oddity that states can satisfy $\varphi(1) \neq 1$. If a state satisfies $\varphi(1) = \lambda \neq 1$, we will call it $\lambda$-normalized, and just normalized otherwise. $0$-normalized states will be called isotropic. We omit the proofs in this section, since they are all trivial. Restating the results below for positive states is left to the reader.

We first analyze the states on the initial $\ast$-algebra, $\mathbb{C}$.
\begin{proposition}
All linear maps $\varphi : \mathbb{C} \rightarrow \mathbb{C}$ are representable.
\end{proposition}
Note that the zero state is representable for all $\ast$-algebras, not just $\mathbb{C}$.
\begin{definition}
$I_{\lambda}$ is the unique state on $\mathbb{C}$ such that $\varphi(1) = \lambda$.
\end{definition}
Note that $I_{1}$ is the monoidal unit.
\begin{lemma}
$I_{\lambda} \otimes I_{\mu} = I_{\lambda \mu}$
\end{lemma}
There are no processes going between states of different normalizations.
\begin{lemma}
If $f : \varphi \rightarrow \psi$ is a morphism in $\mathbf{Phys}$, then $\varphi(1) = \psi(1)$.
\end{lemma}
We can now understand the various roles played by non-normalized states. Let $\mathbf{Phys}_{\lambda}$ be the full subcategory of $\mathbf{Phys}$ containing the $\lambda$-normalized states. 
\begin{theorem}\label{normalizationtheorem}\hspace{1em}
\begin{enumerate}
\item $\mathbf{Phys}$ is the disjoint union of the $\mathbf{Phys}_{\lambda}$:
\begin{displaymath}
\mathbf{Phys} = \coprod_{\lambda \in \mathbb{C}} \mathbf{Phys}_{\lambda}
\end{displaymath}
\item The monoidal structure on $\mathbf{Phys}$ restricts to
\begin{displaymath}
\otimes : \mathbf{Phys}_{\lambda} \times \mathbf{Phys}_{\mu} \longrightarrow \mathbf{Phys}_{\lambda \mu}.
\end{displaymath}
\item For $\lambda \neq 0$, $I_{\lambda} \otimes (-) : \mathbf{Phys}_{\mu} \rightarrow \mathbf{Phys}_{\lambda\mu}$ is an equivalence.
\item $I_{\lambda}$ is terminal in $\mathbf{Phys}_{\lambda}$.
\item $I_{0}\otimes (-)$ maps every state to a zero state.
\end{enumerate}
\end{theorem}
Let $\mathbb{C}$ be the multiplicative monoid of complex numbers, considered as a discrete monoidal category.
\begin{corollary}\label{normalizationcorollary}
The functor $\mathbf{Phys} \rightarrow \mathbb{C}$ given by $\varphi \mapsto \varphi(1)$ is a symmetric monoidal fibration, trivial over $\mathbb{C}^{\ast} \subseteq \mathbb{C}$.
\end{corollary}
Thus $\mathbf{Phys}$ is monoidally equivalent to $\mathbf{Phys}_{0} + \mathbb{C}^{\ast} \times \mathbf{Phys}_{1}$, where $\mathbb{C}^{\ast}$ is the discrete monoidal category of nonzero complex numbers. The equivalence is given by the inclusion of $\mathbf{Phys}_{0}$ on the first term, and by $(\lambda, \varphi) \mapsto I_{\lambda} \otimes \varphi$ on the second.

We are left with only two interesting subcategories of $\mathbf{Phys}$: the monoidal subcategory $\mathbf{Phys}_{1}$, of normalized states, and the mysterious monoidal ideal $\mathbf{Phys}_{0}$, of isotropic states.

\subsection{Examples}
\subsubsection*{Commutative $C^{\ast}$-algebras and Positive States}
By the Riesz-Markov theorem any state $\varphi : C(X) \rightarrow \mathbb{C}$ is a Radon measure $\mu$ on $X$, with the identification being given by $\varphi(f) = \int_{X} f(x) \, d\mu(x)$. The Hermitian form $\langle -, - \rangle_{\varphi}$ is then given by
\begin{displaymath}
\langle f, g \rangle_{\varphi} = \int_{X} f(x) \bar{g}(x) \, d\mu(x),
\end{displaymath}
which is clearly the standard $L^{2}$ inner product, as long as $\varphi$ is positive. It is thus easy to see that $GNS(\varphi) \subseteq L^{2}(\mu)$ is the standard image of $C(X)$ in $L^{2}$. By Lusin's theorem the completion of $GNS(\varphi)$ is the whole of $L^{2}(\mu)$.

$GNS(\varphi)$ continues to be dense in $L^{2}(\mu)$ as long as we assume that $X$ is locally compact and $\sigma$-compact. If $\mu$ is not Radon, then we must assume that $X$ is metrizable. In general $GNS(\varphi)$ is the norm closure of $C_{0}(X)$ in $L^{2}(\mu)$. This norm closure can omit the constant functions, even when they are square-integrable with respect to $\mu$.

Now let $f : X \rightarrow Y$ be a continuous map, and set $\nu = f_{\ast} \mu$. What is $GNS(f) : GNS(\psi) \rightarrow GNS(\varphi)$, with $\psi = f^{\ast} \varphi$? By its explicit construction we see that it is simply pullback $f^{\ast} : L^{2}(\nu) \rightarrow L^{2}(\mu)$. The isometricity of $GNS(f)$ comes down to the adjunction formula:
\begin{displaymath}
\int_{X} f^{\ast} g \, d\mu = \int_{X} g \circ f \, d\mu = \int_{Y} g \, df_{\ast}\mu = \int_{Y} g \, d\nu.
\end{displaymath}
In this example $GNS_{c}(f)$ can be understood as integration along the fibers of $f$, or as the pushforward of measures having $\mu$-densities in $L^{2}$.
\subsubsection*{Endomorphisms of a pre-Hilbert Space}
Let $V$ be a pre-Hilbert space. Any $v \in V$ determines a state $\varphi_{v} : \underline{End}(V) \rightarrow \mathbb{C}$ by the formula
\begin{displaymath}
\varphi_{v}(f) = \langle f(v), v \rangle_{V}.
\end{displaymath}
Clearly $v \in V$ represents $\varphi_{v}$. If $v$ is nonzero, then $(V, v)$ is a cyclic $\ast$-module for $\underline{End}(V)$ by proposition \ref{hilbertproposition}. So by the uniqueness clause in theorem \ref{representabilityconditions} we have $GNS(\varphi_{v}) = V$. We have already encountered this example in the proof of corollary \ref{stinespringcorollary}.

It is worth recalling remark \ref{boundednessremark} here: if $V$ is a Hilbert space, then by uniform boundedness the adjointable maps $\underline{End}(V)$ are exactly the bounded ones.
\section{Recovering Traditional Physics, part I}\label{physics1}
We now start recovering the classical formalism of physics. In this section we consider only the notions which do not require the use of differential calculus. This shortcoming can be remedied by internalization (cf.~section \ref{sins}).
\subsection{Lifting the Schr{\"o}dinger Picture}\label{pictureequivalence}
We have hitherto been working firmly in the Heisenberg picture, using algebras and their homomorphisms to represent physics. This is the more fundamental picture, due to classical mechanics. Here take the first steps toward recovering the Schr{\"o}dinger picture. We can already attach morphisms of Hilbert spaces to homomorphisms of algebras, through the GNS functor. We now investigate how much of this can be reversed.

Let $H$ be a faithful pre-Hilbert $\ast$-module over $A$, and let $U : H \rightarrow H^{\prime}$ be an adjointable isometric linear map to some other pre-Hilbert space. Set $B = UAU^{\ast} \subseteq \underline{End}(H^{\prime})$. Then the map $f : A \rightarrow B$ given by
\begin{displaymath}
f(a) = UaU^{\ast},
\end{displaymath}
is a homomorphism of $\ast$-algebras. The map $f$ is well-defined by the faithfulness of $H$. Note that $B$ is a $\ast$-algebra, but not a subalgebra of $\underline{End}(H^{\prime})$ unless $U$ is unitary.

Recall that every vector $\psi \in H$ defines a state $s(\psi) \in \mathcal{S}_{r}(A)$ by the formula
\begin{displaymath}
s(\psi)(a) = \langle a\psi, \psi \rangle_{H}.
\end{displaymath}
In the following theorem we abuse notation, and write $\psi$ for both the vector and the state it represents. This will not cause confusion, since one can recover the proper meaning by analyzing the types of our expressions.
\begin{theorem}[Lifting of the Schr{\"o}dinger Picture]\label{schrodingerpicture}
In the situation above, for any state $\psi \in H$ we have $f : U\psi \rightarrow \psi$ in $\mathbf{Phys}$, and $GNS(f) = U\vert_{GNS(\psi)}$, i.e.~the following diagram commutes:
\begin{center}
\begin{tikzpicture}
\node (a) at (0,0) {$GNS(\psi)$};
\node (b) at (5, 0) {$GNS(U\psi)$};
\node (c) at (0,-1.6) {$H$};
\node (d) at (5,-1.6) {$H^{\prime}$};
\path[->] (a) edge node[auto] {$GNS(f)$} (b)
		  (c) edge node[auto] {$U$} (d);
		  
\path[{Hooks[right]}->] (a) edge (c)
					  (b) edge (d);
\end{tikzpicture}
\end{center}
\end{theorem}
\begin{proof}
By the pre-Hilbert condition and theorem \ref{positiveuniversality}, $GNS(\psi) = A\psi \subseteq H$, with $\psi$ seen as a state on $A$, and $GNS(U\psi) = UAU^{\ast}U\psi = UA\psi$, with $U\psi$ seen as a state on $B$. Finally, by the construction of $f$, $U$ restricts to a cyclic morphism $GNS(\psi) \rightarrow GNS(U\psi)$ over $f$. So if $GNS(f)$ maps $\psi$ to $U\psi$, we will be done, invoking corollary \ref{coherencedestroyer}. But this is obvious, since by the isometricity of $U$ we have
\begin{displaymath}
f^{\ast}(U\psi)(a) = \langle f(a) U\psi, U\psi \rangle_{H^{\prime}} = \langle UaU^{\ast}U\psi, U\psi \rangle_{H^{\prime}} = \langle a\psi, \psi \rangle_{H} = \psi(a).
\end{displaymath}
$GNS(f)$ must then map $\psi$ to $U\psi$, by its construction for positive states in section \ref{constructionforpositivestates}.
\end{proof}
Perhaps the following is the more natural statement.
\begin{corollary}\label{schrodingerpicturecorollary}
If $H$ and $H^{\prime}$ are Hilbert spaces, and $U$ is unitary, then setting $F(b) = U^{\ast} b U$, for $b$ in some given $B$, yields $F : \psi  \rightarrow U\psi$ in $\mathbf{Phys}$, and $GNS_{c}(F) = U\vert_{GNS(\psi)}$, i.e.~the following diagram commutes (we take $A$ to be $U^{\ast}BU$):
\begin{center}
\begin{tikzpicture}
\node (a) at (0,0) {$GNS(\psi)$};
\node (b) at (5, 0) {$GNS(U\psi)$};
\node (c) at (0,-1.6) {$H$};
\node (d) at (5,-1.6) {$H^{\prime}$};
\path[->] (a) edge node[auto] {$GNS_{c}(F)$} (b)
		  (c) edge node[auto] {$U$} (d);
		  
\path[{Hooks[right]}->] (a) edge (c)
					  (b) edge (d);
\end{tikzpicture}
\end{center}
\end{corollary}
\begin{proof}
Theorem \ref{schrodingerpicture} is applicable to $U^{\ast}$, and gives $GNS(F) = U^{\ast}$. The claim follows by the definition of $GNS_{c}$.
\end{proof}
This is the primary reason for considering the $GNS_{c}$ construction. Its formal properties are in all other respects inferior to those of the $GNS$ functor, since it internalizes poorly, and the analog of theorem \ref{dinaturality} requires invertibility, as seen in corollary \ref{dinaturalitycorollary}.
\begin{remark}\label{coisometryremark}
It is tempting to change the hypothesis in the corollary to ``$U$ is a coisometry'', but this cannot be done due to normalization -- one cannot lift maps connecting vectors (i.e.~states) of different normalizations, by the results of section \ref{normalizationsection}. We will address this issue in section \ref{nonunitarygnssection}.
\end{remark}

We leave the reader wondering about the naturality and uniqueness of the lift constructed in theorem \ref{schrodingerpicture}.
\begin{problem}
Let $\mathbf{Phys}_{fa} \subseteq \mathbf{Phys}_{p}$ be the category of faithful states, with morphisms $f$ such that $GNS(f)$ is adjointable. Is the composite
\begin{displaymath}
\mathbf{Phys}_{fa} \xrightarrow{GNS} \ast\mathbf{Mod} \longrightarrow \textnormal{pre-}\mathbf{Hilb},
\end{displaymath}
where the last arrow is the forgetful functor, an opfibration? The category $\textnormal{pre-}\mathbf{Hilb}$ is the category of pre-Hilbert spaces and adjointable isometric maps.
\end{problem}
In other words: is $f$ in theorem \ref{schrodingerpicture} uniquely determined, and $B$ its minimal codomain? This is obvious if we restrict our attention to unitary maps.
\subsection{Probability, Wave Functions, and Eigenvalues}
Let $\mathbf{Prob}_{L}$ be the category of probability spaces, and measurable, probability preserving maps between them. We will denote such spaces by $(X, \mu)$, where $\mu$ is the probability measure on $X$.

Let $L^{\infty} : \mathbf{Prob}_{L} \rightarrow \mathbf{Phys}_{p}$ assign to each space $(X, \mu)$ the $\ast$-algebra $L^{\infty}(\mu)$, with $L^{\infty}(f) : L^{\infty}(\nu) \rightarrow L^{\infty}(\mu)$, for $f : (X, \mu) \rightarrow (Y, \nu)$, being given by the pullback of functions along $f$. The algebra $L^{\infty}(\mu)$ is equipped with the expectation value state $\mathbb{E}_{\mu} : L^{\infty}(\mu) \rightarrow \mathbb{C}$, given by
\begin{displaymath}
\mathbb{E}_{\mu}(f) = \int_{X} f(x) \, d\mu(x).
\end{displaymath}
$L^{\infty}$ is clearly lax monoidal.

Next, let $LL^{2} : \mathbf{Prob}_{L}^{op} \rightarrow \ast\mathbf{Mod}$ assign to each probability space $(X, \mu)$ the image of $L^{\infty}(\mu)$ in $L^{2}(\mu)$. This functor is also easily seen to be lax monoidal.

Similarly, let $\mathbf{Prob}_{C}$ be the category of compact Radon probability spaces, and continuous probability preserving maps between them. Let $C : \mathbf{Prob}_{C} \rightarrow \mathbf{Phys}_{p}$ be the functor which assigns to each space $X$ the $\ast$-algebra $C(X)$ of complex-valued continuous functions on $X$, with $C(f)$, for $f : X \rightarrow Y$, being again given by pullback of functions along $f$. As before, $C(X)$ is equipped with the expectation value state.

Finally, let $CL^{2}: \mathbf{Prob}_{C}^{op} \rightarrow \ast \mathbf{Mod}$ assign to each Radon space $(X, \mu)$ the image of $C(X)$ in $L^{2}(\mu)$. Like before, this functor is lax monoidal.

\paragraph{}The probabilistic interpretation of quantum theory is based upon theorems of the following form.
\begin{theorem}\label{probabilisticinterpretation}
The following diagrams commute up to natural monoidal isomorphisms, fibered over $\ast \mathbf{Alg}$:
\begin{center}
\begin{tikzpicture}
\matrix (m) [matrix of math nodes, nodes in empty cells,  row sep = 1cm, column sep = 3cm, text height = 1.5ex, text depth = .25ex]{
\mathbf{Prob}_{L}^{op} \\
& \ast \mathbf{Mod} \\
\mathbf{Phys}^{op} \\
};

\path[->] (m-1-1) edge node[auto] {$LL^{2}$} (m-2-2)
		  (m-1-1) edge node[auto, swap] {$(L^{\infty})^{op}$} (m-3-1)
		  (m-3-1) edge node[auto, swap] {$GNS$} (m-2-2);
\end{tikzpicture}
\end{center}
\begin{center}
\begin{tikzpicture}
\matrix (m) [matrix of math nodes, nodes in empty cells,  row sep = 1cm, column sep = 3cm, text height = 1.5ex, text depth = .25ex]{
\mathbf{Prob}_{C}^{op} \\
& \ast \mathbf{Mod} \\
\mathbf{Phys}^{op} \\
};

\path[->] (m-1-1) edge node[auto] {$CL^{2}$} (m-2-2)
		  (m-1-1) edge node[auto, swap] {$C^{op}$} (m-3-1)
		  (m-3-1) edge node[auto, swap] {$GNS$} (m-2-2);
\end{tikzpicture}
\end{center}
\end{theorem}
\begin{proof}
The $LL^{2}$ and $CL^{2}$ functors take values in cyclic modules and cyclic maps. In both cases the constant function $1$ represents the expectation value:
\begin{displaymath}
\langle f 1, 1 \rangle_{L^{2}} = \int_{X} f(x) \, d\mu(x) = \mathbb{E}_{\mu}(f).
\end{displaymath}
Thus, by theorem \ref{representabilityconditions} and corollary \ref{coherencedestroyer}, $LL^{2}$ and $CL^{2}$ coincide with $GNS$ up to unique isomorphism, which then must be natural by cyclicity.
\end{proof}
\begin{remark}\hspace{1em}
\begin{itemize}
\item The monoidal structures on $LL^{2}$ and $CL^{2}$ can also be constructed as part of the proof of the above theorem.
\item We can also use the algebras $L = \bigcap_{p \geq 1} L^{p}$ to represent probability measures. This is usually bigger than $L^{\infty}$ due to, for example, Gaussian random variables, and is not a Banach space in general. The resulting GNS space is not the $L^{2}$ space of the probability measure, and, for general reasons, $L$ cannot act on in by bounded operators.
\item The theorem remains true if we replace probability measures by finite signed measures. Complex measures, on the other hand, cannot be accommodated. One would need to replace Hilbert spaces by quadratic complex spaces.
\end{itemize}
\end{remark}

It is important to understand that the above theorem is only one of a huge family of theorems. The category of probability spaces can be replaced by any number of similar categories, and we have only given diagrams for the two most important cases. The proof always come down to the same simple argument: the $L^{2}$ space contains an obvious representation of the state in question.

This diversity is the result of our liberal approach. $\mathbf{Phys}_{p}$ contains, inadvertently in some sense, various categories of structured $\ast$-algebras, such as $C^{\ast}$-algebras, von Neumann algebras, and $\ast$-algebras of purely algebraic origin. The reader wishing to distinguish them must merely consider a variant of the construction of $\mathbf{Phys}$, suiting the specific application.

\subsubsection{Eigenvalue-Eigenvector Link}
Here is a prototypical application of theorems of this sort. Let $a \in \mathcal{O}(\varphi)$ be a normal observable of some positive state $\varphi$. Normality means that $[a, a^{\ast}] = 0$ or, equivalently, that the $\ast$-algebra generated by $a$ in $\mathcal{O}(\varphi)$ is commutative. One imagines this algebra, denoted by $\langle a \rangle$, to be the algebra of functions on some probability space, with the probability measure given by the restriction of $\varphi$ to $\langle a \rangle$. This gives an object $P_{\varphi}(a)$ in $\mathbf{Phys}_{p}$.

Typically $P_{\varphi}(a)$ can be completed into some algebra in the image of $L^{\infty}$ or $C$. One then has the following theorem.

\begin{theorem}[Eigenvalue-Eigenvector Link]\label{eelink}
Suppose that the canonical map $\varphi \rightarrow P_{\varphi}(a)$, induced by the inclusion $\langle a \rangle \subseteq \mathcal{O}(\varphi)$, admits a factorization
\begin{displaymath}
\varphi \xrightarrow{R} (C(X), \mathbb{E}_{\mu}) \longrightarrow P_{\varphi}(a),
\end{displaymath}
or
\begin{displaymath}
\varphi \xrightarrow{R} (L^{\infty}(X, \mu), \mathbb{E}_{\mu}) \longrightarrow P_{\varphi}(a),
\end{displaymath}
with the second arrow being over an inclusion $\langle a \rangle \subseteq C(X)$ or $\langle a \rangle \subseteq L^{\infty}(X, \mu)$.

Then $a$ is canonically a random variable on $X$ and the following are equivalent for any $\lambda \in \mathbb{C}$:
\begin{enumerate}
\item $a\Omega_{\varphi} = \lambda \Omega_{\varphi}$
\item $a = \lambda$ almost everywhere on $X$
\item $\mathbb{P}(a = \lambda) = 1$
\end{enumerate}
\end{theorem}
\begin{proof}
The equivalence $2 \Leftrightarrow 3$ is obvious. To see the equivalence $1\Leftrightarrow2$ compute $GNS(R)$:
\begin{displaymath}
CL^{2}(X, \mu) \longrightarrow GNS(\varphi)
\end{displaymath}
or
\begin{displaymath}
LL^{2}(X, \mu) \longrightarrow GNS(\varphi).
\end{displaymath}
These are morphisms of cyclic modules representing $\varphi$ for $a$. Thus $a\Omega_{\varphi} = \lambda \Omega_{\varphi}$ is equivalent to $a \cdot 1 = \lambda \cdot 1$ in $CL^{2}$ or $LL^{2}$, where $1$ is the constant function on $X$. But this last condition is equivalent to $a = \lambda$ a.e.~by basic measure theory.
\end{proof}
\begin{remark}
The statement of this theorem is slightly awkward, again, due to our liberal inclusion of any kind of $\ast$-algebra in our categories. In the setting of pure $C^{\ast}$-algebras on can give a much sharper statement, using the full $L^{2}$ space and not requiring a given factorization (since it can always be constructed by spectral theory).
\end{remark}

\subsubsection*{Digression: $GNS^{p}$ and the massless 2d quantum scalar field}
Theorem \ref{probabilisticinterpretation} suggests that the GNS construction is the noncommutative analogue of the $L^{2}$ space. It is well known that the massless quantum scalar field in 2 dimensions cannot be defined in the same manner as in higher dimensions \cite[$\mathsection$1.5]{witten}. One wonders whether the field ``really does not exist'' or, as Witten's constructions suggest, is merely located outside the ``$L^{2}$-realm''. This leads to the following problem.
\begin{problem}\label{2dquantumscalarproblem}
Define the $p$-analog of the GNS construction, such that for Radon measures on compact Hausdorff spaces we have $GNS^{p}(\mu) = L^{p}(\mu)$. Define the massless 2d quantum scalar field in some $GNS^{0}$ space.
\end{problem}
The theory of noncommutative $L^{p}$ spaces for von Neumann algebras is well established \cite{pisierxu} (somewhat less so for $p=0$), and may be relevant here. But the assumption of traciality is problematic.
\subsubsection{Generalized Eigenvalue-Eigenvector Link}
The eigenvalue-eigenvector link can be derived in considerably greater generality, by substituting for Gelfand duality the duality between algebras and affine schemes. No real measure theory is needed -- we will only need to deal with analogues of Dirac delta measures.

Let $\varphi$ be a state with algebra of observables $A = \mathcal{O}(\varphi)$. Let $a \in A$ be a normal element. Theorems \ref{probabilisticinterpretation} and \ref{eelink} say that the number $\varphi(a)$ is to be interpreted as the expectation value, in the sense of probability theory, of $a$ in the state $\varphi$.

The $\ast$-algebra generated by $a$, $B = \mathbb{C}[a,a^{\ast}]$ is commutative. We will denote its inclusion in $A$ by $i: B \hookrightarrow A$.

Passing to the geometric picture, we obtain an affine scheme $X = Spec(B)$ over $\mathbb{C}$, with chosen real form $X_{\mathbb{R}}$. By the adjunction $\Gamma \dashv Spec$, between global sections and the spectrum functor, the global sections of the structure sheaf $\mathcal{O}_{X}$ correspond to complex scheme maps $X \rightarrow \mathbb{A}^{1}_{\mathbb{C}}$. In addition, for $X = Spec(B)$, we have $\mathcal{O}_{X}(X) = B$. Thus $a \in A$ is a complex-valued function on $X$, and we may talk about its values at the points of $X$.

Since we are in the algebraic category, we will have to deal with the fact that the type of value $a$ has depends on the point it is evaluated on: the value of $a$ at $x \in X$ is an element of the residue field $\mathcal{O}_{X,x}/m_{x}$, which is an extension of $\mathbb{C}$. For this reason, we restrict our attention to the $\mathbb{C}$-points of $X$, for which this extension is trivial.

We can now formalize the statement that self-adjoint observables are real-valued.
\begin{proposition}
Self adjoint elements $x \in B$ determine maps $X_{\mathbb{R}} \rightarrow \mathbb{A}^{1}_{\mathbb{R}}$.
\end{proposition}
\begin{proof}
This is just an algebraic geometry consequence of corollary \ref{commutativestaralgebras}.
\end{proof}
This means, in addition, that self-adjoint observables are determined by their values on the real part of $X$, i.e.~$X_{\mathbb{R}}$. Their ``analytic continuation'' to $X$ is automatic.

The state $\varphi$ restricts from $A$ to $B$, giving us a measure-like structure on $X$:
\begin{displaymath}
\mathcal{O}_{X}(X) = B \xrightarrow{i^{\ast}\varphi} \mathbb{C}.
\end{displaymath}
We will abuse terminology, and call linear maps $\mathcal{O}_{X}(X) \rightarrow \mathbb{C}$ measures on $X$. We are interested in measures supported by single points on $X$ -- the ``Dirac delta measures''.
\begin{definition}
Let $X$ be a scheme over $\mathbb{C}$, and $x \in X$ a $\mathbb{C}$-point. The Dirac delta at $x$, denoted $\delta_{x}$, is the localization (i.e.~evaluation) map $\mathcal{O}_{X}(X) \rightarrow \mathcal{O}_{X,x}/m_{x} = \mathbb{C}$.
\end{definition}

Regular functions separate points on affine $X$, and so we have the following lemma.
\begin{lemma}\label{uniquenessofdeltasupport}
If $\delta_{x} = \delta_{y}$ on an affine scheme $X$ over $\mathbb{C}$, then $x = y$.
\end{lemma}
\begin{proof}
The Dirac delta measures are ring homomorphisms, so when they are equal, they determine the same maximal ideal in $\mathcal{O}_{X}(X)$, and hence the same $\mathbb{C}$-point of $X$.
\end{proof}
The lemma fails for projective varieties, since then $\mathcal{O}_{X}(X) = \mathbb{C}$.

Measures naturally push forward under maps of spaces, and the same is true in our setting.
\begin{definition}
Let $\varphi$ be a measure on $X$, and $f : X \rightarrow Y$ a map of schemes over $\mathbb{C}$. Then $f_{\ast}\varphi$ defined by
\begin{displaymath}
\mathcal{O}_{Y}(Y) \xrightarrow{f^{\ast}} \mathcal{O}_{X}(X) \xrightarrow{\varphi} \mathbb{C},
\end{displaymath}
is a measure on $Y$.
\end{definition}
Since $X = Spec(B)$ is the ``space of possible values'', or ``possible (pure) states'' of $a \in B = \mathbb{C}[a, a^{\ast}]$, the following principle is an algebraic reformulation of the condition $\mathbb{P}(a = \lambda) = 1$.
\begin{principle}[Definition of ``having a definite value'']
The observable $a \in A$ has value $\lambda \in \mathbb{C}$ in the state $\varphi : A \rightarrow \mathbb{C}$ if
\begin{displaymath}
i^{\ast}\varphi = \delta_{x},
\end{displaymath}
for some $\mathbb{C}$-point $x \in X$ satisfying $a(x) = \lambda$, where $a : X \rightarrow \mathbb{A}^{1}_{\mathbb{C}}$ is the map constructed above, and $\lambda \in \mathbb{A}^{1}_{\mathbb{C}}$ is the $\mathbb{C}$-point corresponding to $\lambda \in \mathbb{C}$.
\end{principle}
\begin{remark}
In the setting of probability spaces, the above definition is easily seen to be equivalent to \ref{eelink}(2-3).
\end{remark}
We can make the definition more concrete by pushing forward to $\mathbb{A}^{1}_{\mathbb{C}}$:
\begin{proposition}
The observable $a$ has value $\lambda$ in $\varphi$ iff $a_{\ast}i^{\ast}\varphi = \delta_{\lambda}$.
\end{proposition}
\begin{proof}
If $a(x) = \lambda$ then $a_{\ast} i^{\ast}\varphi = a_{\ast}\delta_{x} = \delta_{a(x)} = \delta_{\lambda}$. Conversely, if $a_{\ast}i^{\ast}\varphi = \delta_{\lambda}$, then $i^{\ast}\varphi : B \rightarrow \mathbb{C}$ is a ring homomorphism, by explicit inspection on all elements of $B$ (recall that $\varphi$ is $\ast$-linear), and so represents a $\mathbb{C}$-point $x \in X = Spec(B)$. Then $i^{\ast}\varphi = \delta_{x}$ by our definition of the Dirac delta. Finally $a(x) = \lambda$ by lemma \ref{uniquenessofdeltasupport}.
\end{proof}
We can now generalize the eigenvalue-eigenvector link to our entire setting.
\begin{theorem}[Generalized Eigenvalue-Eigenvector Link]\label{generalizedeigenvalueeigenvectorlink}\hspace{1em}
\begin{enumerate}
\item[a)] If any (hence every) cyclic vector representing $\varphi$ is a $\lambda$-eigenvector of $a$, then the observable $a$ has value $\lambda$ in $\varphi$.
\item[b)] If $i$ is admissible for $\varphi$, and the observable $a$ has value $\lambda$ in $\varphi$, then any (hence every) cyclic vector representing $\varphi$ is a $\lambda$-eigenvector of $a$.
\end{enumerate}
\end{theorem}
\begin{proof}
Any cyclic vector $\Omega$ representing $\varphi$ is part of the unique cyclic module representing $\varphi$, so we may use whichever representation we like.

If $\Omega$ is an $\lambda$-eigenvector of $\Omega$, then the unique cyclic $\ast$-module representing $i^\ast \varphi$ is one-dimensional, and one again finds that $i^{\ast}\varphi$ is a ring homomorphism, by explicit computation, giving $i^{\ast}\varphi = \delta_{x}$, for some $\mathbb{C}$-point $x \in X$. And again, $\lambda$ is the only possible value of $a(x)$, by lemma \ref{uniquenessofdeltasupport}.

If $i^{\ast}\varphi = \delta_{x}$, then $GNS(i^{\ast}\varphi) = L^{2}(\delta_{x}) = \mathbb{C}$, with $B$ acting by evaluation (localization). In particular $a$ acts as multiplication by $a(x) = \lambda$, by assumption. By admissibility we have the map of $\ast$-modules over $i$:
\begin{displaymath}
GNS(i^{\ast}\varphi) \xrightarrow{GNS(i)} GNS(\varphi).
\end{displaymath}
Denoting by $\Omega$ and $\Omega^{\prime}$ the cyclic vectors of $GNS(\varphi)$ and $GNS(i^{\ast}\varphi)$, respectively, we have
\begin{displaymath}
a\Omega = a GNS(i)(\Omega^{\prime}) = GNS(i)(a \Omega^{\prime}) = GNS(i)(\lambda \Omega^{\prime}) = \lambda GNS(i)(\Omega^{\prime}) = \lambda \Omega.
\end{displaymath}
\end{proof}
\begin{corollary}
Let $\varphi$ be a positive state. Then $a$ has value $\lambda$ in $\varphi$ iff any vector in a pre-Hilbert module representing $\varphi$ is a $\lambda$-eigenvector of $a$.
\end{corollary}
\begin{proof}
This follows from theorem \ref{positiveuniversality}, proposition \ref{admissibilityproperties}, and the preceding theorem.
\end{proof}
\subsection{Symmetries and Group Representations}\label{symmetries}
Let $G$ be any symmetry groupoid. The equivariant $GNS$ construction is the categorical exponential
\begin{displaymath}
(\mathbf{Phys}^{op})^{G} \xrightarrow{GNS^{G}} \ast \mathbf{Mod}^{G},
\end{displaymath}
where $(\mathbf{Phys}^{op})^{G}$ is the category of functors $G \rightarrow \mathbf{Phys}^{op}$, and similarly for $\ast \mathbf{Mod}^{G}$. Note that, in general $(\mathcal{C}^{op})^{\mathcal{D}} = (\mathcal{C}^{\mathcal{D}^{op}})^{op}$.

The covariant construction does not require fussing about with opposites:
\begin{displaymath}
\mathbf{Phys}^{G} \xrightarrow{GNS_{c}^{G}} \ast\mathbf{Mod}_{adj}^{G}.
\end{displaymath}
These constructions include symmetry groups (seen as one element groupoids) acting on single states, groupoids of symmetries between different states, and even general categories. We will use all of them below.

For the record, we state:
\begin{theorem}\label{unitaryrepresentations}
Let $G$ be a group, and let $\varphi \in \mathbf{Phys}_{p}^{G}$ be a $G$-symmetric, positive state. Then the covariant GNS construction, $GNS_{c}(\varphi)$, is a unitary representation of $G$.
\end{theorem}
\begin{proof}
Pedantically speaking, one should write $GNS_{c} \circ \varphi : G \rightarrow \ast \mathbf{Mod}_{adj}$. This object simply is, among other things, a unitary representation of $G$.
\end{proof}

Such theorems can be multiplied at will. For example:
\begin{theorem}
The equivariant $GNS$ constructions are naturally symmetric monoidal.
\end{theorem}
\begin{proof}
Let $\mathcal{C}$ be any category. Then $(-)^{\mathcal{C}} : \mathbf{Cat} \rightarrow \mathbf{Cat}$ is a right 2-adjoint, and hence preserves any algebraic structures in $\mathbf{Cat}$. This includes symmetric monoidal categories, and so any 2-functor of the form $(-)^{G}$ lifts to symmetric monoidal categories
\begin{displaymath}
(-)^{G} : \mathbf{SymMonCat} \longrightarrow \mathbf{SymMonCat}.
\end{displaymath}
Its values on $GNS$ and $GNS_{c}$ are the natural structures we are looking for.
\end{proof}

Clearly $G \mapsto \mathbf{Phys}^{G}$ is a functor $\mathbf{Gpd}^{op} \rightarrow \mathbf{SymMonCat}$, likewise $G \mapsto \ast \mathbf{Mod}_{adj}^{G}$. Using the Grothendieck construction we obtain the monoidal fibrations
\begin{align*}
\mathbf{Phys}^{S} &= \int \mathbf{Phys}^{(-)} \\
\ast \mathbf{Mod}_{adj}^{S} &= \int \ast \mathbf{Mod}_{adj}^{(-)},
\end{align*}
of states with some arbitrary symmetry groupoid, and of $\ast$-modules with some $G$-action. The covariant GNS construction becomes a morphism of monoidal fibrations:
\begin{center}
\begin{tikzpicture}
\matrix (m) [matrix of math nodes, nodes in empty cells, row sep = 1 cm, column sep = 1 cm, text height = 1.5ex, text depth = .25ex] {
\mathbf{Phys}^{S} && \ast \mathbf{Mod}_{adj}^{S} \\
& \mathbf{Gpd} \\
};

\path[->] (m-1-1) edge node[auto] {$GNS_{c}^{S}$} (m-1-3)
		  (m-1-1) edge (m-2-2)
		  (m-1-3) edge (m-2-2);
\end{tikzpicture}
\end{center}
This structure allows a systematic investigation of how symmetries restrict and extend for states, observables, and representations.
\subsubsection{Time Reversal}\label{timereversal}
Time reversal provides an excellent excuse for the usage of groupoids of symmetries. Let $M$ be linear Minkowski space. Since we are reversing time, we assume $M$ is time oriented. Traditionally time reversal is an element of the Lorentz group $O(M)$, but this makes applying the formalism of section \ref{antiunitaryprocesses} impossible. Instead we split the Lorentz group into pieces.

Let $\widetilde{O}$ be the following groupoid. Its objects are $M$ and $\overline{M}$, which is $M$ with reversed time orientation. The morphisms are just the orthochronous isometries. Clearly, the maps $\overline{M} \rightarrow M$ are simply the time reversing Lorentz transformations. The full group $O(M)$ is divided into pieces in $\widetilde{O}$.
\begin{remark}
The construction $O \mapsto \widetilde{O}$ can be made systematic, and should be seen as a nonlinear/noncommutative variant of the globular Dold-Kan correspondence.
\end{remark}
Clearly, a state with $O(M)$ symmetry, as traditionally understood, is just a functor
\begin{displaymath}
F : \widetilde{O} \longrightarrow \mathbf{Phys},
\end{displaymath}
such that $F(\overline{M}) = \overline{F(M)}$. This makes time reversing Lorentz transformations into antilinear processes in a natural manner.

We can make this last condition less arbitrary by being more arbitrary with the construction. Let
\begin{displaymath}
1 \longrightarrow O^{+} \longrightarrow O \longrightarrow \mathbb{Z}_{2} \longrightarrow 1
\end{displaymath}
be the exact sequence where $O \rightarrow \mathbb{Z}_{2}$ maps Lorentz transformations to $-1$ if they reverse time, and $1$ if not. We can construct a splitting of this sequence by choosing coordinates on $M$ and sending $-1 \in \mathbb{Z}_{2}$ to the map $(t, x, y, z) \mapsto (-t, x, y, z)$\footnote{This choice is not optimal in odd spacetime dimensions, where the semidirect product can be chosen direct \cite[section 5.5]{pinphysics}.}. This results in a group homomorphism
\begin{displaymath}
h : \mathbb{Z}_{2} \rightarrow Aut(O^{+}),
\end{displaymath}
which classifies the above extension. This turns $\widetilde{O}$ into a $\mathbb{Z}_{2}$-equivariant groupoid, with the generator acting by $M \mapsto \overline{M}$ on objects and by $h$ on arrows.

By theorem \ref{conjugationtheorem} $\mathbf{Phys}$ is already $\mathbb{Z}_{2}$-equivariant, with the generator acting by conjugation. An $O(M)$ symmetry with distinguished time reversal can be defined as a strictly $\mathbb{Z}_{2}$-equivariant functor
\begin{displaymath}
\widetilde{O} \longrightarrow \mathbf{Phys}.
\end{displaymath}
The distinguished time reversal amounts to picking a $\mathbb{Z}_{2}$-fixed point in $\mathbf{Phys}$, i.e.~a specific isomorphism $\overline{\varphi} \rightarrow \varphi$.

\subsubsection{Inhomogeneous Time}\label{categoriesoftime}
The above discussion of symmetries includes time evolution only if it is homogeneous. Then we consider functors
\begin{displaymath}
\mathbb{R} \longrightarrow \mathbf{Phys},
\end{displaymath}
where $\mathbb{R}$ is the additive group of real numbers, considered as a one object groupoid.

Inhomogeneous time evolution can be modeled as well, by considering appropriate ``categories of time''.
\begin{definition}\hspace{1em}
\begin{enumerate}
\item The category of \emph{homogeneous time}, $\mathbf{Time}_{h}$ is the one object groupoid corresponding to the additive group of the real numbers.
\item The category of \emph{inhomogeneous time}, $\mathbf{Time}$ is the pair groupoid corresponding to $\mathbb{R}$.
\item The category of \emph{thermodynamical time}, $\mathbf{Time}_{th}$ is the poset of the real numbers, considered as a category.
\item The category of \emph{restricted thermodynamical time} $\mathbf{Time}_{th}^{t_{0}}$ is the poset of real numbers $\geq t_{0}$, considered as a category.
\end{enumerate}
\end{definition}
The relationships between these categories of time are summarized by the diagram of functors
\begin{center}
\begin{tikzpicture}
\node (a) at (-1,0) {$\mathbf{Time}_{th}^{t_{0}} \subseteq \mathbf{Time}_{th}$};
\node (b) at (3,0) {$\mathbf{Time}$};
\node (c) at (6,0) {$\mathbf{Time}_{h},$};

\path[->] (a) edge node[auto] {$i$} (b)
		  (b) edge node[auto] {$p$} (c);
\end{tikzpicture}
\end{center}
where $i$ is the obvious inclusion, and is actually a localization of $\mathbf{Time}_{th}$, inverting all arrows. The functor $p$ collapses the distinct time objects into one, and maps the unique morphism $t \rightarrow t^{\prime}$ to $t^{\prime} - t$.

The categories of states equipped with various notions of time evolution correspond to the exponentials
\begin{displaymath}
\mathbf{Phys}^{\mathbf{Time}},
\end{displaymath}
with $\mathbf{Time}$ carrying an appropriate subscript.

Homogeneous time determines a single state $\varphi$ and an additive group of automorphisms $U(t) : \varphi \rightarrow \varphi$, $t \in \mathbb{R}$.

Inhomogeneous time determines a state $\varphi(t)$ for every time $t\in\mathbb{R}$, and invertible maps $U(t, t^{\prime}) : \varphi(t) \rightarrow \varphi(t^{\prime})$ subject to $U(t,t) = id$ and
\begin{displaymath}
U(t^{\prime}, t'')U(t, t^{\prime}) = U(t, t'').
\end{displaymath}
Thermodynamical time is similar to inhomogeneous time, but $U(t, t^{\prime})$ is only given for $t \leq t^{\prime}$, and need not be invertible. Restricted time simply restricts $t$ to $t \geq t_{0}$ for objects and morphisms.

In section \ref{nonunitarygnssection} we will construct a statistical version of $\mathbf{Phys}$, called $\mathbf{Phys}_{M}$. In that category thermodynamical time can truly come into its own, with
\begin{displaymath}
\mathbf{Phys}_{M}^{\mathbf{Time}_{th}^{t_{0}}}
\end{displaymath}
generalizing the notion of a quantum dynamical semigroup (cf.~\cite{holevo}).
\subsection{Composite Systems}\label{compositesystems}
In this section we assume states are normalized, working exclusively with $\mathbf{Phys}_{1}$. The monoidal structure on normalized states can be characterized as the most general notion of composite satisfying the following axioms.
\begin{axioms}[Axioms for Composite Systems]
A state $\varphi \boxtimes \psi$ will be called a \emph{composite of $\varphi$ and $\psi$} if we are given the following structure and properties:
\begin{enumerate}
\item \emph{Composition}: there are morphisms
\begin{align*}
p_{\varphi} : \varphi \boxtimes \psi &\longrightarrow \varphi \\
p_{\psi} :\varphi \boxtimes \psi &\longrightarrow \psi,
\end{align*}
in $\mathbf{Phys}_{1}$, meaning that $\varphi \boxtimes \psi$ contains a copy of both $\varphi$ and $\psi$.
\item \emph{Noninteraction}: these copies do not affect each other, meaning:
\begin{displaymath}
\varphi \boxtimes \psi(p_{\varphi}(a)p_{\psi}(b)) = \varphi(a) \psi(b),
\end{displaymath}
for all $a \in \mathcal{O}(\varphi)$ and $b \in \mathcal{O}(\psi)$.
\item \emph{Probabilistic Independence}: these copies are independent, in the sense of noncommutative probability theory. For any $a \in \mathcal{O}(\varphi)$ and $b \in \mathcal{O}(\psi)$ we have
\begin{displaymath}
[p_{\varphi}(a), p_{\psi}(b)] = 0,
\end{displaymath}
in $\mathcal{O}(\varphi \boxtimes \psi)$. Here $[x, y]$ denotes the commutator of $x$ and $y$.
\end{enumerate}
\end{axioms}
\begin{remark}\hspace{1em}
\begin{itemize}
\item These axioms are not completely independent. Composition implies all instances of noninteraction in which $a=1$ or $b=1$.
\item Without the normalization assumption the composition axiom can't be satisfied. By theorem \ref{normalizationtheorem}(1-2) if either $\varphi$ or $\psi$ is not normalized then one of $p_{\varphi}$ or $p_{\psi}$ cannot exist. If neither is normalized then neither can exist.
\end{itemize}
\end{remark}
\begin{theorem}\label{compositesystemcharacterization}
$\varphi \otimes \psi$ is initial among the composites of $\varphi$ and $\psi$.
\end{theorem}
\begin{proof}
$\varphi \otimes \psi$ clearly satisfies requirements 1-3. That it is initial follows from theorem \ref{tensorproductuniversalproperty}: by probabilistic independence the maps $\mathcal{O}(p_{\varphi})$ and $\mathcal{O}(p_{\psi})$ factor uniquely through $\mathcal{O}(\varphi) \otimes \mathcal{O}(\psi)$, and by noninteraction the pullback of $\varphi \boxtimes \psi$ along this factorization must be $\varphi \otimes \psi$. This gives a unique structure preserving map $\varphi \otimes \psi \rightarrow \varphi \boxtimes \psi$ in $\mathbf{Phys}_{1}$.
\end{proof}
\begin{corollary}
The initial composite satisfies the following additional axioms:
\begin{enumerate}
\item[4.] \emph{Process Covariance}: $\otimes$ is a functor:
\begin{displaymath}
\otimes : \mathbf{Phys}_{1} \times \mathbf{Phys}_{1} \longrightarrow \mathbf{Phys}_{1}.
\end{displaymath}
This means that processes can be composed, in addition to states.
\item[5.] \emph{Naturality of Composition}: the components $p_{\varphi}$ and $p_{\psi}$ form natural transformations $\otimes \rightarrow \pi_{i}$, where $\pi_{i}$ is the projection
\begin{displaymath}
\pi_{i} : \mathbf{Phys}_{1} \times \mathbf{Phys}_{1} \longrightarrow \mathbf{Phys}_{1}.
\end{displaymath}
This means that the initial composite is uniform, and not dependent on the details of any states.
\item[6.] \emph{No Further Relations}: $\otimes$ is initial in the category of functors with the structures and properties above.
\end{enumerate}
\end{corollary}
\begin{proof}
These are obvious, with point 6 being a weakening of theorem \ref{compositesystemcharacterization}.
\end{proof}
Theorem \ref{compositesystemcharacterization} and corollary \ref{normalizationcorollary} characterize the monoidal product on $\mathbf{Phys}$ for all non-isotropic states. The composites of isotropic states remain mysterious.

\section{Statistical Physics and Non-Unitary Processes}\label{statisticalphysics}

%
\subsection{Non-unitary GNS}\label{nonunitarygnssection}
To define noncommutative Markov processes, we must extend the notion of admissibility.
\begin{definition}\label{linearadmissibilitydefinition}
A $\ast$-linear map $\Phi : A \rightarrow B$ between $\ast$-algebras is admissible for a state $\varphi \in \mathcal{S}_{r}(B)$ if $A^{\perp} \subseteq \Phi^{-1}(B^{\perp})$.
\end{definition}
Here $A^{\perp}$ is computed for the Hermitian form induced by $\psi = \Phi^{\ast}\varphi$. Note that $\psi$ is representable by theorem \ref{representabilityconditions}(3). Unlike before, the inclusion $\Phi^{-1}(B^{\perp}) \subseteq A^{\perp}$ is no longer trivial, since $\Phi$ is not multiplicative.

Just like in section \ref{representationofadmissiblemaps} we define the linear map
\begin{displaymath}
GNS_{M}(\Phi) : GNS(\psi) = A/A^{\perp} \rightarrow B/B^{\perp} = GNS(\varphi)
\end{displaymath}
by the formula $[x] \mapsto [\Phi(x)]$. No analogue of proposition \ref{admissibilityproperties} is available, and the formula looks like an arbitrary choice.

$GNS_{M}(\Phi)$ is no longer isometric or cyclic, but it can be computed in interesting cases, due to the following proposition.
\begin{proposition}\label{overf}
$GNS_{M}(\Phi) : GNS(\psi) \rightarrow GNS(\varphi)$ satisfies the following identity:
\begin{displaymath}
GNS_{M}(\Phi)(a\Omega_{\psi}) = \Phi(a) \Omega_{\varphi},
\end{displaymath}
for all $a \in A$.
\end{proposition}
\begin{proof}
$GNS_{M}(\Phi)(a \Omega_{\psi}) = GNS_{M}(\Phi)([a]) = [\Phi(a)] = \Phi(a) \Omega_{\varphi}$.
\end{proof}
Note that this property looks like ``being a linear map over $\Phi$'', but it applies only to the cyclic vector. We do not, in general, have $GNS_{M}(\Phi)(av) = \Phi(a)GNS_{M}(\Phi)(v)$ for arbitrary $v \in GNS(\psi)$ and $a \in A$. Note also that this proposition applies to $a = 1$, showing that $GNS_{M}(\Phi)$ is cyclic iff it's unital.

An analogue of proposition \ref{admissibilitypropertiestwo} is available.
\begin{proposition}\label{linearadmissibilityproperties} \hspace{1em}
\begin{enumerate}
\item The composite of admissible maps is admissible
\item The tensor product of admissible maps is admissible
\item Completely positive maps between positive states are admissible
\end{enumerate}
\end{proposition}
\begin{proof}
We deal with the complications of not being a homomorphism on a case-by-case basis.
\paragraph{Ad 1.} This is still obvious, as before.
\paragraph{Ad 2.} Since $GNS(\varphi \otimes \psi) = GNS(\varphi) \otimes GNS(\psi)$ we have
\begin{displaymath}
(A \otimes B)^{\perp} = \ker (A \otimes B \longrightarrow GNS(\varphi \otimes \psi)) = A^{\perp} \otimes B + A \otimes B^{\perp}.
\end{displaymath}
Now let $\Phi : C \rightarrow A$ and $\Psi : D \rightarrow B$ be admissible for states $\varphi \in \mathcal{S}_{r}(A)$ and $\psi \in \mathcal{S}_{r}(B)$, respectively. Then $\Phi \otimes \Psi$ is clearly $\ast$-linear, and
\begin{align*}
(\Phi \otimes \Psi)^{-1}(A \otimes B)^{\perp} &= (\Phi \otimes \Psi)^{-1}(A^{\perp} \otimes B + A \otimes B^{\perp}) \\
&= (\Phi \otimes \Psi)^{-1}(A^{\perp} \otimes B) + (\Phi \otimes \Psi)^{-1}(A \otimes B^{\perp}) \\
&= \Phi^{-1}(A^{\perp}) \otimes D + C \otimes \Psi^{-1}(B^{\perp}) \\
&\supseteq C^{\perp} \otimes D + C \otimes D^{\perp} \\
&= (C\otimes D)^{\perp},
\end{align*}
where we use admissibility of $\Phi$ and $\Psi$ in the penultimate step.

Note that here we heavily rely on linear algebra over fields, especially the flatness of any vector space.
\paragraph{Ad 3.} We use point 1 together with the Stinespring factorization \ref{stinespringcorollary}. Any completely positive map $\Phi : A \rightarrow B$ fits into a commutative square as follows:
\begin{center}
\begin{tikzpicture}
\matrix (m)[matrix of math nodes, row sep = 1.5 cm, column sep = 2.5cm, text height = 1.5ex, text depth = .25ex]{
\underline{End}(H) & \underline{End}(L) \\
A & B \\
};

\path[->] (m-2-1) edge node[auto] {$\pi$} (m-1-1)
		  (m-2-1) edge node[auto] {$\Phi$} (m-2-2)
		  (m-1-1) edge node[auto] {$V^{\ast} - V$} (m-1-2)
		  (m-2-2) edge node[auto] {$i$} (m-1-2);
\end{tikzpicture}
\end{center}
Here $L$ is any pre-Hilbert $B$-module, $H$ is some pre-Hilbert $A$-module depending on $L$, $V : L \rightarrow H$ is an adjointable linear map, and $\pi$ is a homomorphism of $\ast$-algebras. The reader may wish to review the construction of these objects, given before theorem \ref{stinepring}.

Now let $\varphi : B \rightarrow \mathbb{C}$ be a positive state, and set $L = GNS(\varphi)$. Then $GNS(i)$ is, as a function of sets, the identity on $GNS(\varphi)$, by proposition \ref{hilbertproposition}. Thus it suffices to show that $i\Phi$ is admissible. But $\pi$ is admissible, so by point 1 we only need to check that $\Psi = V^{\ast}(-)V : \underline{End}(H) \rightarrow \underline{End}(L)$ is admissible.

For this we use lemma \ref{radicalisnullspace}. To do so, we must show that $\Psi$ preserves positive vectorial states $\varphi_{v} : \underline{End}(L) \rightarrow \mathbb{C}$, i.e.~those given by
\begin{displaymath}
\varphi_{v}(f) = \langle fv, v \rangle_{L}.
\end{displaymath}
We compute
\begin{displaymath}
\Psi^{\ast}\varphi_{v}(f) = \langle V^{\ast} f Vv, v \rangle_{L} = \langle f Vv, Vv \rangle_{H},
\end{displaymath}
which is non-negative, since $H$ is a pre-Hilbert space.

Next we set $v = \Omega_{\varphi} \in L$, and check the admissibility of $\Psi$ for $\varphi_{v}$ using lemma \ref{radicalisnullspace}. We see that
\begin{displaymath}
\underline{End}(H)^{\perp} = \{f : \Psi^{\ast}\varphi_{v}(f^{\ast}f) = 0\},
\end{displaymath}
which is exactly those $f \in \underline{End}(H)$ for which $fVv = 0$. On the other hand
\begin{displaymath}
\underline{End}(L)^{\perp} = \{g : \varphi_{v}(g^{\ast} g) = 0 \},
\end{displaymath}
which is those $g \in \underline{End}(L)$ for which $gv = 0$. Thus if $f \in \underline{End}(H)^{\perp}$ then $\Psi(f) = V^{\ast}f V \in \underline{End}(L)^{\perp}$, which means $\Psi$ is admissible for $\varphi_{v}$. 
\end{proof}

Now let $\ast\mathbf{Alg}_{M}$ be the category of $\ast$-linear maps between $\ast$-algebras, and let
\begin{displaymath}
\mathcal{S}_{M} : \ast\mathbf{Alg}_{M}^{op} \longrightarrow \mathbf{Set}
\end{displaymath}
be the functor assigning to every algebra its set of representable states. This is well-defined by theorem \ref{representabilityconditions}(3).
\begin{definition}
$\mathbf{Phys}_{M}$ is the subcategory of $1\downarrow \mathcal{S}_{M}$ spanned by the admissible morphisms.
\end{definition}
This is well-defined by proposition \ref{linearadmissibilityproperties}, which also implies the next theorem.
\begin{theorem}
$\mathbf{Phys}_{M}$ is a symmetric monoidal category.
\end{theorem}
Before stating that $GNS_{M}$ is symmetric monoidal, we must determine its codomain. For now, we declare it to be $\mathbf{Herm}$, the category of nondegenerate Hermitian vector spaces and all linear maps between them. This category is symmetric monoidal by lemma \ref{flatnesslemma}.

This choice neglects a lot of structure, such as the module structure on $GNS(\varphi)$, and the property described in proposition \ref{overf}. Because of this we cannot say that $GNS_{M}$ is fibered over $\ast\mathbf{Alg}_{M}$.
\begin{theorem}\label{nonunitarygns}
The constructions
\begin{align*}
\varphi &\longmapsto GNS(\varphi) \\
\Phi &\longmapsto GNS_{M}(\Phi),
\end{align*}
for objects $\varphi \in \mathbf{Phys}_{M}$ and morphisms $\Phi : \varphi \rightarrow \psi$ in $\mathbf{Phys}_{M}$, are part of a strong symmetric monoidal functor
\begin{displaymath}
GNS_{M} : \mathbf{Phys}_{M} \longrightarrow \mathbf{Herm}.
\end{displaymath}
\end{theorem}
\begin{proof}
For any state $\varphi$ the module $GNS(\varphi)$ is cyclic, and so we can use proposition \ref{overf} for computations. That $GNS_{M}$ is a functor is then obvious.

Note that $\mathbf{Phys}_{M}$ has the same objects as $\mathbf{Phys}$ (literally). It just has more morphisms. Thus the monoidal structure is already there, and we merely have to check that our transformation
\begin{equation}\label{monoidalnaturality}
GNS_{M}(\varphi) \otimes GNS_{M}(\psi) \longrightarrow GNS_{M}(\varphi \otimes \psi)
\end{equation}
remains natural. The coherence conditions don't involve maps outside of $\mathbf{Phys}$, and so are still automatically satisfied.

The isomorphism (\ref{monoidalnaturality}), constructed abstractly in theorem \ref{gnsconstruction}, is easily computed by cyclicity. It's the map
\begin{displaymath}
a \Omega_{\varphi} \otimes b \Omega_{\psi} \longmapsto a\otimes b \Omega_{\varphi \otimes \psi},
\end{displaymath}
where $a \in \mathcal{O}(\varphi)$ and $b \in \mathcal{O}(\psi)$ are acting on the appropriate cyclic vectors.

Now consider $\Phi : \varphi^{\prime} \rightarrow \varphi$ and $\Psi : \psi^{\prime} \rightarrow \psi$ in $\mathbf{Phys}_{M}$ and compute:
\begin{align*}
GNS_{M} (\Phi \otimes \Psi) (a \otimes b \Omega_{\varphi \otimes \psi}) &= \Phi \otimes \Psi (a \otimes b) \Omega_{\varphi^{\prime} \otimes \psi^{\prime}} \\
&= \Phi(a) \otimes \Psi(b) \Omega_{\varphi^{\prime} \otimes \psi^{\prime}} \\
&\mapsto \Phi(a) \Omega_{\varphi^{\prime}} \otimes \Psi(b) \Omega_{\psi^{\prime}} \\
&= GNS_{M}(\Phi)(a \Omega_{\varphi^{\prime}}) \otimes GNS_{M}(b \Omega_{\psi^{\prime}}) \\
&= GNS_{M}(\Phi) \otimes GNS_{M}(\Psi)(a \Omega_{\varphi^{\prime}} \otimes b \Omega_{\psi^{\prime}}),
\end{align*}
where we first use proposition \ref{overf}, and check naturality for the inverse of (\ref{monoidalnaturality}).
\end{proof}

\subsection{The Covariant Representation}
A new problem arises when trying to take the adjoint of $GNS_{M}$. The maps $GNS_{M}(\Phi)$ are not isometric, and so are not guaranteed to have an adjoints upon passing to Hilbert completions. We deal with this in a manner similar to what we suggested after theorem \ref{covariantgnsproperty}.

Let $\textnormal{pre}\mathbf{Hilb} \subseteq \mathbf{Herm}$ be the monoidal subcategory of pre-Hilbert spaces and bounded maps between them. Define
\begin{displaymath}
\mathbf{Phys}_{M, pb} = GNS_{M}^{-1}(\textnormal{pre}\mathbf{Hilb}),
\end{displaymath}
giving a monoidal subcategory of $\mathbf{Phys}_{M}$ spanned by the positive states and processes with bounded $GNS$ representations between them. This category contains $\mathbf{Phys}_{p}$, and is therefore already quite rich. We will see in section \ref{wavefunctioncollapsesection} that it is a proper extension of $\mathbf{Phys}_{p}$.

\begin{definition}
The covariant $GNS_{M}$ construction, $GNS_{M,c}$ is defined as the composite
\begin{center}
\begin{tikzpicture}
\node (a) at (0,0) {$\mathbf{Phys}_{M,pb}$};
\node (b) at (4,0) {$\textnormal{pre}\mathbf{Hilb}^{op}$};
\node (c) at (10,0) {$\mathbf{Hilb}$};

\path[->] (a) edge node[auto] {$GNS_{M}^{op}$} (b)
		  (b) edge node[auto] {completion + adjoint} (c);
\end{tikzpicture}
\end{center}
\end{definition}
By the definition of the tensor product of Hilbert spaces, we have the following theorem.
\begin{theorem}
$GNS_{M,c} : \mathbf{Phys}_{M, pb} \longrightarrow \mathbf{Hilb}$ is a symmetric monoidal functor.
\end{theorem}

\subsection{Gelfand Duals of Markov Processes}\label{gelfandduals}
In this section we extend Gelfand duality to Markov processes and completely positive maps. We follow \cite{stochasticgelfand}, albeit with more pedestrian notation.

The Gelfand dual of a completely positive unital map is a Markov process in Radon measures. To see this consider a (completely) positive map
\begin{displaymath}
\Phi : C(Y) \longrightarrow C(X),
\end{displaymath}
and compute
\begin{displaymath}
\Phi(f)(x) = \int_{X} \Phi(f) \, d\delta_{x} = \int_{Y} f \, d\Phi^{\ast}(\delta_{x}),
\end{displaymath}
where $\delta_{x}$ is the Dirac delta measure at $x$ (i.e.~evaluation at $x$). This shows that $\Phi$ is the dual of the Markov process given by
\begin{align*}
X & \longrightarrow M(Y) \\
x & \longmapsto \Phi^{\ast}(\delta_{x}),
\end{align*}
where $M(Y)$ is the space of Radon probability measures on $Y$.

Conversely, given a Markov process $F : X \rightarrow M(Y)$ we obtain a completely positive map
\begin{align*}
C(Y) & \longrightarrow C(X) \\
f & \longmapsto (x \mapsto \int_{Y} f \, dF(\delta_{x})).
\end{align*}
These identifications clearly generalize Gelfand duality, and are compatible with composition. To see the second claim, recall that multiplication in $M$ is given in components $m_{X} : M(M(X)) \rightarrow M(X)$ by
\begin{displaymath}
\int_{X} f \, d(m_{X}(\lambda)) = \int_{M(X)}\int_{X} f(x) \, d\nu(x) \, d\lambda(\nu).
\end{displaymath}
The composition of two Markov processes $F : X \rightarrow M(Y), G : Y \rightarrow M(Z)$ is given by
\begin{displaymath}
X \xrightarrow{F} M(Y) \xrightarrow{M(G)} M(M(Z)) \xrightarrow{m_{Z}} M(Z).
\end{displaymath}
Now consider two positive maps $\Psi : C(Z) \rightarrow C(Y), \Phi: C(Y) \rightarrow C(X)$, with duals $G, F$ respectively. The dual of their composite is
\begin{displaymath}
x \mapsto \Psi^{\ast}(\Phi^{\ast}(\delta_{x})) = \Psi^{\ast}(F(x)) = m_{Z}G_{\ast}(F(x)) = G(F(x)).
\end{displaymath}
To see the penultimate equality consider any Radon measure $\mu$ in place of $F(x)$, and compute:
\begin{displaymath}
\int_{Z} f \, d(m_{Z}G_{\ast}\mu) = \int_{M(Z)} \int_{Z} f(z) \, d\nu(z) \, d(G_{\ast}\mu)(\nu) = \int_{Y} \int_{Z} f(z) \, dG(y)(z) \, d\mu(y),
\end{displaymath}
demonstrating that $\Psi^{\ast}(\mu) = m_{Z}(G_{\ast} \mu)$. The last equality uses the well known adjunction formula: $\int g^{\ast} f \, d\mu = \int f \circ g \, d\mu = \int f \, dg_{\ast}\mu$.

The Radon measure monad is lax monoidal, with the monoidal structure given by
\begin{align*}
M(X) \times M(Y) &\longrightarrow M(X \times Y) \\
(\mu, \nu) &\longmapsto \mu \otimes \nu \\
1 & \longrightarrow M(1) \\
\ast &\longmapsto \delta_{1}
\end{align*}
One easily verifies that the composition and unit on $M$ are monoidal transformations. Because of this, for completely formal reasons \cite{zawadowski}, the Kleisli category $\mathbf{CptHaus}_{M}$ for $M$ is monoidal, with the monoidal product given by
\begin{displaymath}
(X \xrightarrow{F} M(Z)) \otimes (Y \xrightarrow{G} M(T)) = X \times Y \xrightarrow{F \times G} M(Z) \times M(T) \longrightarrow M(Z \times T),
\end{displaymath}
where the last arrow is the monoidal product on $M$.

The identification of completely positive maps with Markov processes is monoidal. Given $\Phi : C(T) \rightarrow C(Y), \Psi: C(Z) \rightarrow C(X)$, with duals $F, G$ respectively, the dual of $\Phi \otimes \Psi : C(Z) \otimes C(T) \rightarrow C(X) \otimes C(Y)$ is, under the identification $C(X) \otimes C(Y) \simeq C(X \times Y)$,  $F \otimes G$. To see this note that under the isomorphism $C(X) \otimes C(Y) \simeq C(X \times Y)$ the measure $\delta_{(x,y)}$, for $(x,y) \in X \times Y$, corresponds to the functional $\delta_{x} \otimes \delta_{y}$ on $C(X) \otimes C(Y)$, and compute
\begin{displaymath}
(\Phi \otimes \Psi)^{\ast}(\delta_{(x,y)}) \simeq (\Phi \otimes \Psi)^{\ast} (\delta_{x} \otimes \delta_{y}) = \Phi^{\ast}(\delta_{x}) \otimes \Psi^{\ast}(\delta_{y}) = F(x) \otimes G(y) = F\otimes G(x,y),
\end{displaymath}
demonstrating that the dual of $\Phi \otimes \Psi$ is $F \otimes G$.

These computations demonstrate the following theorem.
\begin{theorem}[Theorem 5.1 in \cite{stochasticgelfand} ]\label{gelfanddualityextensiontheorem}
Gelfand duality extends to a monoidal equivalence
\begin{displaymath}
\mathbf{CptHaus}_{M} = \{\textnormal{commutative } C^{\ast}\textnormal{-algebras with positive unital maps}\}^{op},
\end{displaymath}
where $\mathbf{CptHaus}_{M}$ is the Kleisli category of the Radon probability measure monad, i.e.~the category of Markov processes in $\mathbf{CptHaus}$.
\end{theorem}
Equivalently $\mathbf{CptHaus}_{M}$ is the category of Markov processes with Radon measure kernels between compact Hausdorff spaces.
\begin{remark}
We have restricted ourselves to probability measures, since only then is $M(X)$ a compact Hausdorff space. Finite measures give a locally compact Hausdorff space, and require working with locally compact spaces from the beginning. Since we are focusing on unital algebras, we will not pursue this generalization here.
\end{remark}
\begin{corollary}\label{gelfanddualityextensioncorollary}
The category of compact Radon probability spaces and Markov processes between them is monoidally equivalent to the category of states on commutative $C^{\ast}$-algebras and positive unital maps between them.
\end{corollary}
\begin{proof}
The first category is the coslice $1/\mathbf{CptHaus}_{M}$ and the second is the slice $\{C^{\ast}$-algebras and positive maps between them$\}/\mathbb{C}$. They are clearly dual to each other, through the above monoidal equivalence.
\end{proof}

\subsection{Quantum Markov Processes}
The discussion above allows us to generalize the relationship between the GNS construction and probability theory (theorem \ref{probabilisticinterpretation}) to the case of Markov processes. We begin by extending the functor $CL^{2}$ to our new setting.

Let $F : X \rightarrow Y$ be a Markov process between probability spaces. Then by corollary \ref{gelfanddualityextensioncorollary} we obtain a completely positive unital map
\begin{displaymath}
C(F) : C(Y) \longrightarrow C(X),
\end{displaymath}
which furthermore preserves the expectation values on $C(X)$ and $C(Y)$.

We define
\begin{displaymath}
CL^{2}(F) : L^{2}(Y) \longrightarrow L^{2}(X)
\end{displaymath}
by the formula
\begin{equation}\label{cl2}
CL^{2}(F)(f)(x) = \int_{Y} f \, dF(\delta_{x}),
\end{equation}
with the right hand side seen as an element of $L^{2}(X)$. One easily sees that this is well defined, and monoidal. Indeed, the formula (\ref{cl2}) is just the composite of pulling back by the Gelfand dual of $C(F)$ with the projection to the $GNS$ space. As such it is immediately obvious that $CL^{2}(F)$ is given by the same formula defining $GNS_{M}$, leading us to the following theorem.

\begin{theorem}\label{completeprobabilisticcompatibility}
The following prism of symmetric monoidal functors commutes up to natural monoidal isomorphism
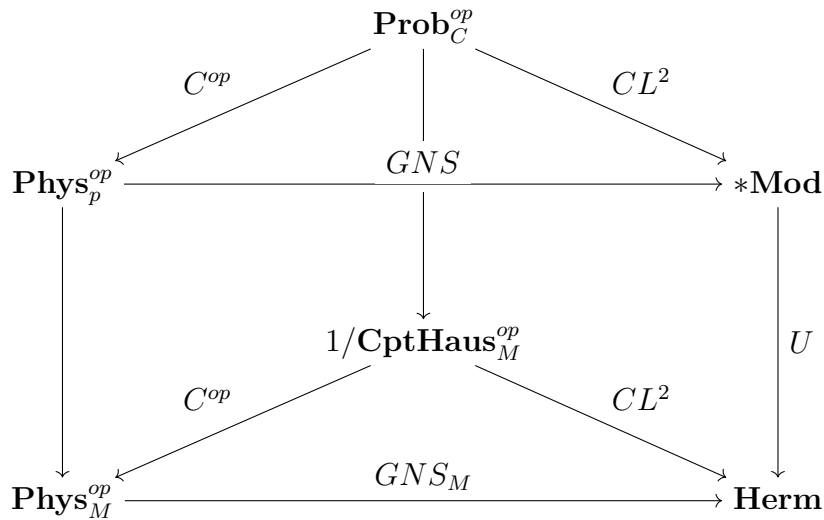
\begin{center}
\begin{tikzpicture}[overline/.style={preaction={draw=white, -, line width=6pt}}]
\matrix (m) [matrix of math nodes, nodes in empty cells, row sep = 1.5cm, column sep = 2.5cm, text height = 1.5ex, text depth = .25ex]{
& \mathbf{Prob}_{C}^{op}\\
\mathbf{Phys}_{p}^{op}  && \ast \mathbf{Mod} \\
& 1/\mathbf{CptHaus}_{M}^{op}\\
\mathbf{Phys}_{M}^{op} && \mathbf{Herm} \\
};

\path[->] (m-2-1) edge  (m-4-1)
		  (m-2-3) edge node[auto] {$U$} (m-4-3)
		  (m-4-1) edge node[auto] {$GNS_{M}$} (m-4-3)
		  (m-1-2) edge node[auto] {$CL^{2}$} (m-2-3)
		  (m-1-2) edge (m-3-2)
		  (m-2-1) edge[overline] node[auto, preaction = {fill=white}] {$GNS$} (m-2-3)
		  (m-1-2) edge node[auto,swap] {$C^{op}$} (m-2-1)
		  (m-3-2) edge node[auto] {$CL^{2}$} (m-4-3)
		  (m-3-2) edge node[auto,swap] {$C^{op}$} (m-4-1);
\end{tikzpicture}
\end{center}
where $U$ is the obvious forgetful functor, and the unlabeled arrows are inclusions.
\end{theorem}
\begin{proof}
The top triangle commutes by theorem \ref{probabilisticinterpretation}. The back left square commutes by corollary \ref{gelfanddualityextensioncorollary}. The front square commutes by the definitions of $GNS$ and $GNS_{M}$ (the formula for $GNS$ is a necessary consequence of proposition \ref{admissibilityproperties}, see remark \ref{uniquenessremark}). The bottom triangle commutes by theorem \ref{probabilisticinterpretation} and the explicit constructions of $GNS_{M}$ and $CL^{2}$, as mentioned above. The back right square then commutes theorem \ref{probabilisticinterpretation}, and the commutativity of the bottom triangle.

These isomorphisms are given by easily computed explicit formulas. We leave checking their coherence to the reader.
\end{proof}
\begin{remark}
We have omitted the $L^{\infty}$ version of this theorem. It would require the duality of section \ref{gelfandduals} for von Neumann algebras. Such a generalization should not present any serious difficulty -- for compact Hausdorff spaces and Radon measures, the $L^{2}$space is (topologically) cyclic for both $C(X)$ and $L^{\infty}$-algebras.
\end{remark}
This theorem raises more questions than it answers.
\begin{problem}\hspace{1em}
\begin{enumerate}
\item What does the Stinespring factorization theorem mean for ordinary (commutative) Markov processes?
\item What does the Kleisli category structure on ordinary Markov processes mean for completely positive maps?
\item Is $\mathbf{Phys}_{M}$ a Kleisli category for some monad on $\mathbf{Phys}$?
\end{enumerate}
\end{problem}
\paragraph{Example.}
Let $\textnormal{Time} =\mathbb{R}_{\geq 0}$ be the the order of the nonnegative real numbers, considered as a category. The functor category
\begin{displaymath}
\mathbf{Phys}_{M}^{\textnormal{Time}},
\end{displaymath}
represents a vast generalization of the category of quantum dynamical semigroups (cf.~\cite{holevo}). The preceding theorem shows that this notion completely subsumes the notion of a Markov semigroup, defined as semigroups of maps in $1/\mathbf{CptHaus}_{M}$.

\subsection{Conditioning}\label{wavefunctioncollapsesection}
Since $GNS$ maps for admissible morphisms are not usually cyclic, the proper reformulation of theorem \ref{schrodingerpicture} is not obvious. Corollary \ref{stinespringcorollary} makes this superfluous to a certain extent, giving a definite form to the most interesting maps under investigation -- the completely positive ones. Here we will simply note some obvious examples, which exhibit enough cyclicity for computation. The reader should think of these computations as extending corollary \ref{schrodingerpicturecorollary} to coisometries.

\subsubsection{State Vector Collapse}
Consider an inclusion $i : H \rightarrow H^{\prime}$ of Hilbert spaces. Its adjoint $i^{\ast}$ is the orthogonal projection $H^{\prime} \rightarrow H$. Both $\Phi = i^{\ast}(-)i$ and $\Psi = i (-) i^{\ast}$ are completely positive maps between $\underline{End}(H)$ and $\underline{End}(H^{\prime})$, with $\Phi \circ \Psi = 1_{\underline{End}(H)}$, and $\Psi \circ \Phi = ii^{\ast} (-)ii^{\ast}$ being a conditioning operator by the self-adjoint projection $P = ii^{\ast}$. By proposition \ref{hilbertproposition} both $H$ and $H^{\prime}$ are cyclic for their endomorphism algebras, and any nonzero vector is cyclic. By remark \ref{boundednessremark} $\underline{End}(H)$ is just the usual algebra of all bounded operators on $H$, and $H$ is algebraically cyclic over it -- no closure required.

Now let $v \in H$, and $\varphi_{v} : \underline{End}(H) \rightarrow \mathbb{C}$ be the state given by $\varphi_{v}(a) = \langle av, v \rangle_{H}$. By the above $GNS(\varphi_v) = H$, with cyclic vector $\Omega_{\varphi_{v}} = v$. Next let $\psi = \Phi^{\ast} \varphi_{v}$, and compute
\begin{displaymath}
\Phi^{\ast}\varphi_{v}(a) = \varphi_{v}(\Phi(a)) = \langle i^{\ast} aiv, v\rangle_{H} = \langle aiv, iv\rangle_{H^{\prime}},
\end{displaymath}
to see that $\psi = \varphi_{iv}$. Thus $GNS_{M}(\Phi) : H \rightarrow H^{\prime}$. By proposition \ref{overf}, it acts as
\begin{displaymath}
a iv \longmapsto \Phi(a) v = i^{\ast}aiv,
\end{displaymath}
for $a \in \underline{End}(H^{\prime})$, and $v$ kept fixed, and so $GNS_{M}(\Phi) = i^{\ast}$. Similarly, $\Psi^{\ast} \varphi_{w} = \varphi_{i^{\ast}w}$, for $w \in H^{\prime}$, and $GNS_{M}(\Psi) = i$, acting as 
\begin{displaymath}
ai^{\ast}w \longmapsto \Psi(a)i^{\ast}w = iai^{\ast}ii^{\ast}w = iai^{\ast}w,
\end{displaymath}
for $a \in \underline{End}(H)$.

This gives $GNS_{M}(\Psi) \circ GNS_{M}(\Phi) = GNS_{M}(\Phi \circ \Psi) = ii^{\ast} = P$, the orthogonal projection onto $H$. Note that, since $GNS_{M}$ is contravariant, the morphism ``$\Phi \circ \Psi$'', as an arrow of $\mathbf{Phys}_{M}$, corresponds to the algebra homomorphism $\Psi \circ \Phi$. The covariant $GNS_{M}$ functor yields $GNS_{M,c}(\Phi) = i$, and $GNS_{M,c}(\Psi) = i^{\ast}$. This gives the expected identity $GNS_{M,c}(\Phi \circ \Psi) = P$. This is another indication of the physical naturality of the $GNS_{c}$ construction.

We can also compute the effect of conditioning on an arbitrary $\ast$-algebra $A$. Let $P \in A$ be a self-adjoint projection, let $\varphi$ be a positive state on $A$, let $\Phi : A \rightarrow A$ be $\Phi(a) = PaP$, and set $\psi = \Phi^{\ast} \varphi$. Clearly, $\Phi$ is completely positive. Note that
\begin{displaymath}
\psi(a) = \varphi(\Phi(a)) = \varphi(PaP) = \langle aP\Omega_{\varphi}, P\Omega_{\varphi} \rangle_{GNS(\varphi)},
\end{displaymath}
and hence $\psi$ is represented by $P\Omega_{\varphi} \in GNS(\varphi)$. Thus, by theorem \ref{positiveuniversality}, $GNS(\psi) = AP\Omega_{\varphi} \subseteq GNS(\varphi)$, 

As above, $GNS_{M}(\Phi) : GNS(\psi) \rightarrow GNS(\varphi)$ acts as
\begin{displaymath}
aP\Omega_{\varphi} = a\Omega_{\psi} \longmapsto PaP\Omega_{\varphi} = Pa\Omega_{\psi},
\end{displaymath}
which means it is the composite
\begin{displaymath}
GNS(\psi) \hookrightarrow GNS(\varphi) \xrightarrow{P} GNS(\varphi).
\end{displaymath}
Note that $GNS(\psi)$ is not, in general, contained in the image of $P$, so this is a nontrivial map. If $A = End(H)$, and $\varphi = \varphi_{v}$ for some $v \in H$, we would have $GNS(\psi) = H = GNS(\varphi)$, as long as $\psi$ is nonzero. For general $A$, $GNS(\psi)$ may be a proper subspace of $GNS(\varphi)$. Returning to the current situation, $GNS_{M,c}(\Phi)$ is given by the action of $P^{\ast} = P$ followed by the orthogonal projection onto (the Hilbert completion of) $GNS(\psi)$.

I submit to the reader that these computations provide a reasonable mathematical interpretation of the notion of ``state vector collapse'', with $P$ given by a suitable spectral projection. What we have shown is that it is an unnormalized conditioning operation. The lack of normalization is not a problem -- rescaling is a Markov process in our setting.

\subsubsection{Scattering}
We can define scattering processes, such as annihilation $e^{+} + e^{-} \rightarrow \gamma + \gamma$, directly as Markov processes, for any well-defined scattering matrix.

Let $\alpha$ and $\beta$ denote two types of particles (possibly composite). The scattering process $\alpha \rightarrow \beta$ is given by the composite
\begin{displaymath}
H_{\alpha} \xrightarrow{i_{\alpha}} \mathcal{F}(H) \xrightarrow{S} \mathcal{F}(H) \xrightarrow{p_{\beta}} H_{\beta},
\end{displaymath}
where $H_{\alpha}, H_{\beta} \subseteq \mathcal{F}(H)$ are the Hilbert spaces of states of the $\alpha$ and $\beta$ particles, $\mathcal{F}$ is the Fock space functor, $H$ is an arbitrary Hilbert space (usually a uniform mixture of elementary particles), $S$ is a unitary operator (called the scattering matrix), and the $i$ and $p$ maps are inclusions and projections, respectively. Since all these maps are inclusions, projections, or are unitary, this composite defines an arrow $S_{\alpha \beta} : \alpha \rightarrow \beta$ in $\mathbf{Phys}_{M}$ over the completely positive map $(p_{\beta}Si)^{\ast} (-) p_{\beta}Si : \underline{End}(H_{\beta}) \rightarrow \underline{End}(H_{\alpha})$, such that $GNS_{M,c}(S_{\alpha \beta})$ is the composite displayed above.

Note that in decomposing $S$ into its matrix elements we lose the full algebra of observables on Fock space, and must restrict to the observables preserving the $\alpha$ and $\beta$ particles.

\subsubsection{Corollary: The ``Penrose Problem''}
We finish this section by indulging in wild quantum gravity speculation. Below we unapologetically ignore the specific content of Penrose's ideas \cite{penrose}, and gratuitously appropriate his name nonetheless.

The critical point I wish to communicate here is that the basic idea of gravity collapsing the state of a system \emph{could be right}. What's more, we are in a position to look for its mathematical realization. We formulate the search as follows.
\begin{problem}[``Penrose Problem'']\label{penroseproblem}
Which bordisms can be monoidally represented by conditioning maps?
\end{problem}
More formally let $\mathbf{Bord}$ be some category of structured bordisms, such as timelike Lorentzian bordisms. Are there any symmetric monoidal functors
\begin{displaymath}
\mathbf{Bord} \longrightarrow \mathbf{Phys}_{M},
\end{displaymath}
which map a bordism to a conditioning process? Clearly, such bordisms cannot be invertible. But, with the extra structure afforded by a metric, there are plenty of such morphisms, even for topologically trivial bordisms. Expanding and collapsing spacetimes are both obvious examples. Dualizing the the TQFT wisdom that
\begin{quote}
(\ldots) \emph{the absence of topology change implies unitary time evolution}.
\begin{flushright}
John Baez, \cite{quandaries} (emphasis original)
\end{flushright}
\end{quote}
we can say that
\begin{quote}
\emph{The presence of a dynamical metric allows non-unitary time evolution.}
\end{quote}
We even know that this allowance is utilized by quantum field theory in curved spacetime \cite{wald}.

At the physical level of rigor, we can formulate our question as follows: dynamical spacetime appears to represent a flux of information. Can this information be used to condition states evolving in that spacetime? Can there be a gravity-induced outflux beyond what is required by the canonical commutation relations?

The answer to the second question appears to be yes -- consider Hawking radiation. I consider it to be an exact result in an approximate theory, hence worthy of mathematical consideration. By hand-waving CPT arguments \cite{hawking} we expect influxes to be possible as well.

It would be interesting to investigate this obviously information theoretic aspect of bordism representations to Verlinde's ideas on entropic gravity \cite{verlinde}.

\subsection{Remarks on Measurement and Interpretation}\label{interpretation}
Having constructed state vector collapse as a legitimate dynamical object, it is only natural to return to the problem of interpreting quantum theory. In this section we return to the axiomatic postulates of the introduction, treating states and processes synthetically. This determines abstract categories called $\mathbf{Phys}$ and $\mathbf{Phys}_{M}$, which should not be confused with their specific models constructed earlier. We proceed through a series of remarks.
\begin{enumerate}
\item There is no recognized measurement problem in classical mechanics. This is only possible due to assigning probability a purely epistemic role, claiming it to be a quantification of our ignorance (and exclusively ignorance).

\item This makes mixed states completely fictitious. If there is a classical system whose (mixed) states do not obey Choquet theory, then this claim would be invalidated. Mixed states would need to be treated as independently existing entities. As far as ontology is concerned, the probabilistic combination $\frac{1}{2} \varphi + \frac{1}{2} \psi$ is just as problematic as any complex superposition.

\item Quantum theory makes such epistemic dodging impossible. Due to the commutation relations, $[p,q] = i \hbar$, the states required by epistemic interpretations do not exist. There is no probability space on which $p$ and $q$ are both scalar variables.

\item Hidden variable theories push back on the epistemic front, postulating an unobservable (even in principle \cite[2.5-2.6]{durr}) exact state. This alters the mathematical formalism, and will not be discussed here.

\item No two interpretations can disagree on the statistics of measurement, since that would lead to empirical differences. What role is left? It seems that it is \emph{exclusively} the probabilistic aspect of quantum theory that is problematic. To challenge this claim one would need to produce a non-epistemic interpretation of statistical mechanics which does not extend to quantum theory.

\item If the observable functor, $\mathcal{O}$, is not faithful, then measurement statistics can fail to distinguish two distinct processes. From a realist perspective, this gives rise to \emph{essentially non-quantitative} ``laws'' of physics, and deeply muddies the problem of measurement. We give three examples of increasing sophistication.

\paragraph{Failure of Gelfand Duality.}
Consider a non-Hausdorff space $X$, seen as a space of states of some system. Then the observables $X \rightarrow \mathbb{R}$ factor through the Hausdorffization, which collapses certain states in $X$. Since continuous maps $X \rightarrow Y$ serve (by analogy) as physical processes, we see that observables can miss differences among them.
\paragraph{Random Processes in $\mathbf{Phys}_{M}$.}
The category $\mathbf{Phys}_{M}$ can be constructed inside the topos of presheaves on probability spaces (cf.~section \ref{sins}). Then there are stochastic (i.e.~internal) functors
\begin{displaymath}
(\cdot \longrightarrow \cdot) \rightarrow \mathbf{Phys}_{M},
\end{displaymath}
which are empirically indistinguishable. In particular, one cannot tell if collapse actually occurs during measurement or not.

To get an approximate idea of how stochastic functors work, consider random variables taking values in the arrows of $\mathbf{Phys}_{M}$, without fixed domains and codomains. This approximation is unfortunately not technically viable, since there is no natural $\sigma$-algebra on the arrows of $\mathbf{Phys}_{M}$, or even on the hom-sets.
\paragraph{Gauge Theories.}
If our discussion in appendix \ref{appendixa} is on the right track, then $\mathbf{Phys}$ for gauge theories should look something like the 2-category $\mathbf{Gpd}$, of all groupoids. Then $\mathcal{O}$ is simply $\mathbf{Gpd}(-, \mathbb{R})$, with $\mathbb{R}$ discrete. This functor is obviously not faithful.

This unfaithfulness seems to have the effect of necessitating the consideration of ghosts, despite the fact that they ``have no physical significance''.

\item These examples suggest that there is plenty of purely mathematical ambiguity to go around, even before any serious interpretation is required. In particular, the notion of measurement, classical or quantum, is sorely lacking in conceptual development and mathematical structure.

\item The standard form of a measurement, given by a process
\begin{displaymath}
\textnormal{System} \otimes \textnormal{Apparatus} \longrightarrow \textnormal{System} \otimes \textnormal{Apparatus},
\end{displaymath}
is inadequate in two ways:
\begin{enumerate}
\item If the system is either the universe, or the apparatus, then the form above is simply wrong. There is nothing outside the universe, and self-measurement does not involve two copies of oneself. Yet we measure the universe and ourselves regularly. What happens? Why would measurement be a distinguished type of physical process? If it's not distinguished from mere time evolution, how would \emph{we} distinguish it? Such a distinction would be a prima facie formal object, with direct impact on our empirical pronouncements -- a truly miraculous entity.

\item As we have seen in section \ref{compositesystems}, $\otimes$ is the \emph{noninteracting} composite. Spatial compositions in laboratories are not of this kind. In particular, the composites have significantly fewer possible states: two bricks and $\textnormal{brick} \otimes \textnormal{brick}$ differ because of fermionic statistics. The ultimate significance of this is unclear to me.
\end{enumerate}

\item It would be interesting to consider complete interpretations as fully formal structures, taking the form of phenomenological reduction functors
\begin{displaymath}
\mathbf{Phys}_{M} \longrightarrow \mathbf{Pheno},
\end{displaymath}
taking values in phenomenological categories, constructed out of phenomena -- the direct objects of experience, which do not require any additional interpretation. Every person does seem to have such a metacategory (cf.~\cite[I.1]{maclane}) at hand. Can it be made a mathematical object? Is there a mathematical theory of subjectivity?
\item If probability is ontologically traumatic, then it's exit could be even more so. The imaginary Planck constant $i\hbar$ is a parameter controlling degree of noncommutativity among observables. To similarly introduce a parameter \cancel{a}, controlling associativity, would wipe out our access to probabilistic structures. Observables would have expectation values, but not distributions! Frequentism would somehow necessarily fail (magic!), and the existence of conserved quantities could depend on the choice of observables.

\item The preceding is a general phenomenon. Whenever we have a functor of categories of spaces $F: \mathcal{C}\rightarrow \mathcal{D}$, which is not a morphism of sites, the geometries of $X$ and $F(X)$ can differ greatly. This applies in particular to noncommutative and nonassociative geometries. For a riveting discussion of how the geometry of the affine line depends on commutativity see \cite{madore}.
\end{enumerate}

\section{Sins of Omission}\label{sins}
Two items are conspicuously missing from the preceding work. They are differential geometry and classical mechanics. Their inclusion, via internalization in some topos $E$, will now be briefly sketched. The full details will appear in forthcoming work. What follows is an outline which may interest experts. Before that, some remarks on the difficulties still to be addressed.

Internal topology and integration/measure theory seem to require approaches radically differing from classical mathematics. Because of this the covariant representation, $GNS_{c}$, is missing. Without completeness, or something like it, the adjoints required by $GNS_{c}$ are not guaranteed to exist. The Markov representation $GNS_{M}$ should not pose difficulties, but the whole probabilistic framework is missing, because of the lack of integration theory.

The primary difficulty of (locally) internalizing the $GNS$ functor is an ample supply of nondegenerate Hermitian forms, closed under the tensor product. These would be supplied by lemma \ref{flatnesslemma}, if not for the fact that fields are a useless concept in a topos. Over general rings, tensor products of bilinear forms seem to invoke essentially all possible homological complications. Even in the Cahiers topos, one would need to verify the flatness of Hilbert spaces (what a concept!) to prove the existence of interesting infinite dimensional examples. General convenient vector spaces are not bornologically flat, so there is no reason to expect flatness after embedding in the Cahiers topos.

The way forward seems to require developing the homological algebra of $C^{\infty}$-rings, or at least the $C^{\infty}$-analogues of relative affine schemes, coherence (for rings), and maps locally of finite type. Then the $C^{\infty}$-finitely generated examples would be interesting. In well adapted models of synthetic differential geometry the $C^{\infty}$-structure on the ring object $\mathbf{R}$ is visible internally, since $C^{\infty}(\mathbb{R}) = E(\mathbf{R}, \mathbf{R})$. The $C^{\infty}$-rings in $E$ are then the models of an internal algebraic theory, a notion which is well understood \cite[D5.3]{johnstone}.

The intended application of these constructions is setting the stage for the construction of the moduli space of vacua. The constructions below can be understood as endowing $\mathbf{Phys}$ with a smooth structure, giving rise to a ``space of all theories''. Naive attempts to construct a ``subspace of vacua'' within that space are met with stiff technical resistance. For a discussion of these issues we refer to appendix \ref{appendixb}.

\subsection{Internalizing $GNS$}
Let $E$ be a model of synthetic differential geometry \cite{sdg}, with ring $\mathbf{R}$. Choose a quadratic extension $\mathbf{R} \rightarrow \mathbf{C}$, to be treated as an analogue of the usual extension $\mathbb{C}/\mathbb{R}$. In particular, we demand an involution $\overline{(-)}$ of $\mathbf{C}$, such that $x \overline{x}$ lies in $\mathbf{R}$ for all $x \in \mathbf{C}$, and is positive if $\mathbf{R}$ happens to be ordered (we worked specifically to be able to omit any positivity requirement). In well adapted models $\mathbf{R}$ is typically an $\mathbb{R}$-algebra (i.e.~a $\Delta^{\ast}\mathbb{R}$-algebra), and we may set $\mathbf{C} = \mathbf{R} \otimes_{\mathbb{R}} \mathbb{C}$.

The construction of $GNS : \mathbf{Phys}^{op} \rightarrow \ast\mathbf{Mod}$ from this data is very simple, and can be carried out internally to $E$. The swiftest method is appealing to stack semantics \cite{shulmansemantics}. The procedure has two steps. First one writes down the formula defining the $GNS$ functor -- including the domain and codomain -- over $\mathbf{Set}$. This is not trivial, since there are many such formulas whose meanings diverge in other toposes, and the intentionally correct one must be chosen. This formula is in essence a procedure for constructing a morphism of $\mathbf{SymMonCat} = \mathbf{SymMonCat}(\mathbf{Set})$ (we use large sets on the right).

Stack semantics allows the same procedure over $E$. Naively one would expect the result to be in $\mathbf{SymMonCat}(E)$, or its locally internal analogue. But the internal logic of $E$ may have certain opinions that do not match reality. There could be an internal functor $F : \mathbf{C} \rightarrow \mathbf{D}$ such that
\begin{displaymath}
\vdash_{E} \textnormal{``}F \textnormal{ is an equivalence''},
\end{displaymath}
meaning the internal logic of $E$ says that $F$ is an equivalence, but $F$ is not \emph{actually} an equivalence. The inverses may exist locally in $E$, but fail to assemble into a globally defined object. To fix this discrepancy, and gain the flexibility of freely using internal equivalences we simply add the missing equivalences. This means localization.
\begin{theorem}\label{localizationtheorem}
Let $E$ be a small topos. Then there is a 2-adjunction $F \dashv U$,
\begin{center}
\begin{tikzpicture}
\node (b) at (6,0) {$\mathbf{Stacks}(E)$};
\node (a) at (0,0) {$\mathbf{Cat}(E)$};

\path[->] (a) edge[bend left = 10] node[auto] {$F$} (b)
		  (b) edge[bend left = 10] node[auto] {$U$} (a);
\end{tikzpicture}
\end{center}
which exhibits small stacks over $E$ as the reflective 2-localization of internal categories in $E$ at the \emph{local equivalences} -- the internal functors which $E$ asserts to be equivalences.
\end{theorem}
\begin{remark}\label{localizationremark}\hspace{1em}
\begin{itemize}
\item $F$ is the Grothendieck construction followed by stackification, and $U$ is splitting followed by sheafification. That this makes sense follows from the proof of lemma 4 in \cite[chapter 5]{awodey}.
\item If there are enough points, then the local equivalences are the stalk-wise equivalences.
\item This theorem extends to any $\kappa$-ary superextensive site, linking stacks over $\mathcal{C}$ and internal categories in $Sh(\mathcal{C})$. One wonders whether superextensivity is required.
\end{itemize}
\end{remark}
Following this philosophy we take the defining formula for the $GNS$ functor, and replace any instance of $\mathbf{Set}$ with ``the stack of objects of $E$'', better known as the codomain fibration $E^{\cdot \rightarrow \cdot} \longrightarrow E$ (call it $\mathbb{E}$). To our horror, we realize that the result is not quite right.

What should replace the category of sets is what I will call $\mathbb{E}_{lc}$ -- the ``stack of locally constant objects'' of $E$. It's the full substack of $\mathbb{E}$ generated by the global sections. One way to construct it is as the stackification of a presheaf of categores on $E$ whose objects are always the objects of $E$, and whose morphisms at stage $X$ between $A$ and $B$ are given by $E/X(\pi^{\ast}A , \pi^{\ast}B)$, where $\pi : X \rightarrow 1$.

This has the effect of working with families of objects which are locally trivial. $\mathbb{E}_{lc}(X)$ consists of those families in $E/X = \mathbb{E}(X)$ which become trivial over some covering of $X$. They are glued from product families via a cocycle. The inclusion $\mathbb{E}_{lc} \subseteq \mathbb{E}$ is fully faithful, so we do not lose any of the morphisms.

\begin{remark}\label{wildphysics}
If we use the full stack of objects then physical oddities can occur. In particular the existence of constants of nature can depend on the value of other constants of nature! Think of the residue fields in the base of a non-trivial family of schemes. Classical physics also becomes ``richer'' (or ``infested with junk''), encompassing exotic structures other than Poisson algebras.
\end{remark}

Ultimately, the result is a morphism of monoidal stacks over $E$. For aesthetic reasons we may wish to push the entire setting into internal categories in some colossal topos. ``Internal categories of physical processes'' sounds much more elegant than ``stacks of processes''.

Over well-adapted models the result includes at least the finite dimensional $C^{\ast}$-algebras, and their full moduli theory. The Cahiers topos includes all the convenient vector spaces \cite{kockreyes} as $\mathbf{R}$-modules, and so one hopes for a lot more, but they cannot be used to construct examples until we prove them to be $\mathbf{R}$-flat. I do not expect all convenient vector spaces to be flat, and if Hilbert spaces are not flat, then very few interesting examples will exist.

In this manner have arrived in a paradisal world, where everything is smooth. Both functors and families of objects and morphisms can be differentiated, and these two modes of differentiation lead to the traditional differential equations of quantum theory (Heisenberg and Schr{\"o}dinger) and to classical mechanics, respectively. We give only examples.

\subsection{Infinitesimal Symmetries}
Consider a $G$-equivariant $\ast$-algebra, i.e.~a functor $A : G \rightarrow \ast\mathbf{Alg}$. We treat $G$ as a one object groupoid, and hence, by the Grothendieck construction, as a prestack over $E$ (its stackification  consists of $G$-torsors \cite{bunge}, so we keep prestacks around for simplicity). In this picture, $A$ is a morphism of prestacks, and, unwinding the definitions, we see that $A$ amounts to a traditionally defined equivariant object in a fibration (cf.~\cite{vistoli}). Below we write $A$ for both the functor and the image in $\ast\mathbf{Alg}$ of the single object of $G$, a particular $\ast$-algebra in $E$.

Let $D = \{x \in \mathbf{R} : x^2 = 0\} \subseteq \mathbf{R}$ be the first order infinitesimals. In the synthetic setting differentiation is reduced to composing with $D$. Since we are working with prestacks, this amounts to evaluation, by the ``Yoneda lemma for fibrations'' \cite{streicher}. Evaluating $A(D)$, we find the following: $G(D) = TG$ is just the tangent bundle of $G$, and the rest of the structure amounts to a homomorphism
\begin{displaymath}
TG \longrightarrow End(A)(D) = End_{D}(A \times D),
\end{displaymath}
where the codomain is the endomorphisms of $A$ over $D$, that is commuting diagrams
\begin{center}
\begin{tikzpicture}
\matrix (m) [matrix of math nodes, nodes in empty cells, row sep = 1.5cm, column sep = 1.5cm, text height = 1.5ex, text depth = .25ex]{
A \times D && A \times D \\
& D \\
};
\path[->] (m-1-1) edge node[auto] {$f$} (m-1-3)
		  (m-1-1) edge node[auto, swap] {$\pi$} (m-2-2)
		  (m-1-3) edge node[auto] {$\pi$} (m-2-2);
\end{tikzpicture}
\end{center}
where $f$ is a $\ast$-algebra homomorphism, and the $\pi$ are projections to $D$. The Kock-Lawvere axiom shows that this data amounts to a $\ast$-derivation $A \rightarrow A$, recovering the usual the notion of infinitesimal symmetry. In particular we obtain a morphism of Lie algebras
\begin{displaymath}
Lie(G) \longrightarrow \ast Der(A).
\end{displaymath}
Since all we are really doing is composition, we can compose everything with the $GNS$ functor.

\begin{theorem}\label{infinitesimalgns}
Let $X \in Lie(G)$ act as the inner derivation $[Q, -]$ on $A$, for some $Q \in A$, and let $\varphi$ be a $G$-equivariant state over $A$. Then $GNS(X)$ acts on $GNS(\varphi)$, and
\begin{displaymath}
GNS(X) = Q \textnormal{ iff } Q\Omega = 0
\end{displaymath}
\end{theorem}
Thus infinitesimal generators coincide in the Heisenberg and Schr{\"o}dinger pictures only if the representing vector is invariant under the \emph{chosen} generator. This invariance can always be sabotaged, since $Z(A)$ always includes $\mathbf{C}$. Choices matter, and in this case are classified by Hochschild cohomology $HH^{0}(A) = Z(A)$. This is the second indication -- after lemma \ref{flatnesslemma} (see remark \ref{derivedremark}) -- that we should pass to a derived (i.e.~higher categorical) formalism.

\begin{remark}
Morally speaking, theorem \ref{infinitesimalgns} shows that $GNS_{c}$ would, had we enough modules isomorphic to their duals at our disposal, map the Heisenberg equation to the Schr{\"o}dinger equation. This infinitesimal result would complete the equivalence of these pictures, as it is traditionally understood.
\end{remark}

This discussion can be extended to groupoids. For simplicity, let's consider the pair groupoid for the affine line $\mathbb{A}^{1}$, which is just the base ring $\mathbf{R}$ as an object. The objects are $\mathbb{A}^{1}$ itself, and there is a unique morphism $t \rightarrow t^{\prime}$ for any two points, which we will identify with translation by $t^{\prime} - t$. We will call this groupoid $P(\mathbb{A}^{1})$.

Differentiating, we see that $P(\mathbb{A}^{1})(D)$ has as objects tangent vectors to the objects of $P(\mathbb{A}^{1})$. This means tangent vectors to  $\mathbb{A}^{1}$, which are naturally just vectors in $\mathbb{A}^{1}$. The specific object (point) to which these vectors are attached is determined by restriction $1 \rightarrow D \rightarrow \mathbb{A}^{1}$. The morphisms of these ``infinitesimal families of objects'' are again tangent vectors, with the unique morphism $v \rightarrow w$ identified with the translation by $w-v$.

All this data maps to $\ast\mathbf{Alg}(D)$, determining infinitesimal families of $\ast$-algebras $A_{v}$, for $v \in \mathbb{A}^{1}$, and isomorphisms $(w-v) : A_{v} \rightarrow A_{w}$ of $\ast$-algebras over $D$. Since everything is $\mathbf{R}$-linear, this is determined completely by any nontrivial map $v : A_{0} \rightarrow A_{v}$, which is a derivation along a deformation of $A$. If the deformation is trivial, i.e.~time acts on observables but not their algebra, we get a time-dependent family of derivations of $A$, just as expected.

The possibility of deformation arises since we allowed infinitesimal movement of the algebra itself, not just of its elements. The very notion of multiplication moved, along with a movement of the elements. This leads to a discussion of the classical limit.

\subsection{The Classical Limit}
Consider the affine line $\mathbb{A}^{1}$ as a discrete category in $E$. It again defines a stack over $E$, and we define $\hbar$-families of $\ast$-algebras to be functors $\mathbb{A}^{1} \rightarrow \ast\mathbf{Alg}$. The classical limit of such a family is its restriction to infinitesimal $\hbar$. Thus we are led to study maps
\begin{displaymath}
D \longrightarrow \ast\mathbf{Alg}.
\end{displaymath}
Because $\mathbb{E}_{lc}$ is full, maps $X \rightarrow \ast\mathbf{Alg}$ are those $\ast$-algebras in $E/X$ which become trivial -- as objects, but not algebras! -- over some covering of $X$. Maps retain arbitrary dependence on the fibers. In well adapted models this construction includes vector bundles over manifolds equipped with not-locally-trivial $\ast$-algebra structures.

Since $D$ is amazingly tiny in the sense of Lawvere \cite[appendix 4]{sdg}, maps $D \rightarrow \ast\mathbf{Alg}$ are simply $\ast$-algebra structures on $\pi : A \times D \rightarrow D$ in $E/D$, which extend the given structure on $A$ (thought of as sitting in the fiber over $0 \in D$). The Kock-Lawvere axiom shows that these are exactly the $\ast$-Hochschild cocycles on $A$.

The monoidal structure on $\ast\mathbf{Alg}$ restricts to a product of Hochschild cocycles, which includes the traditional product of Poisson structures. Classical and quantum composite systems are thus fully compatible.

In this way we include a very general version of deformation quantization. In particular the quantization of singular phase spaces can utilize symmetric Hochschild cocycles, in addition to the antisymmetric ones (which correspond to Poisson brackets). This may have bearing on the quantization of principal connections with isotropy and spacetimes with isometries (cf.~remark \ref{symmetrichochschild} in appendix \ref{appendixa}).

\subsection{Compatibility}
All of these considerations are functorial. In particular, the inhomogeneous Heisenberg, Schr{\"o}dinger, and Hamilton equations, as well as a ``classical Schr{\"o}dinger equation'' all derive from a single object in 
\begin{displaymath}
\mathbf{Phys}^{\mathbb{A}^{1} \times P(\mathbb{A}^{1})},
\end{displaymath}
where the first factor controls the value of $\hbar$ and the second is the pair groupoid of $\mathbb{A}^{1}$, representing inhomogeneous time evolution.

\appendix
\section{On The Notion of Gauge Theory}\label{appendixa}
The ideas presented here are not really new, but deserve being intensely stressed, for they deeply challenge any claim to understanding the general notion of gauge theory, especially quantum gauge theory.

These ideas are present implicitly or explicitly in the thinking of several authors, including Freed and Deligne \cite[$\mathsection$4.2]{freedfield}, Schreiber and Schulman \cite{schreiber-shulman} (among many), as well as Benini, Schenkel and Szabo \cite{benini, schenkel}, and very likely many others.

\subsection{The Problem}
Consider a classical theory with space of states $X$, carrying an action of a group $G$. Is $G$ a group of gauge equivalences, or an ordinary symmetry group? In the physicists' practice the distinction is always clear. However, there does not seem to be a mathematical criterion for establishing such a distinction. Yang-Mills theory, for example, together with the claim that the connection field is an empirically measurable observable, appears to be a perfectly fine mathematical structure. It is simply not a gauge theory, and does not have a well posed first order initial value problem.

In general, phase space-based approaches to gauge invariance are doomed. Despite appearances, gauge theories are not a special class of constrained systems. Gauge symmetry is not a property inconveniencing the construction of phase space, but a structure, and attempts to infer it from anything else cannot succeed.

Locality is also not a very promising candidate, since it requires saying ``spacetime'' and ``Cauchy surface''. Whatever string theory turns out to be, it will probably be out of luck with this kind of definition. And we definitely want it to make the list! Thus we seek a more conceptual understanding of gauge symmetry.

That is a rather tall order, since we are faced with the following dumbfounding claims:
\begin{itemize}
\item Diffeomorphisms in General Relativity are gauge equivalences.
\item But: isometries are actual symmetries, not just gauge equivalences (think of the Poincar\'e group, and Killing vector fields in general).
\item The automorphisms of a principal bundle are gauge equivalences.
\item But: the fiberwise constant automorphisms of a trivial principal bundle are actual symmetries (how else would electric charge be conserved?).
\end{itemize}
As the reader can see, there's a lot of backtracking going on. It gets worse. Consider a nontrivial principal $G$-bundle $P$ over a spacetime $M$. Then we have an exact sequence
\begin{displaymath}
1 \longrightarrow \mathscr{G}_{P} \longrightarrow Aut(P) \longrightarrow Diff(M),
\end{displaymath}
where $\mathscr{G}_{P}$ is the group of $M$-automorphisms of $P$, and $Aut(P)$ consists of all the $G$-automorphisms of $P$. The last map is typically not onto, and does not split over its image (which consists of the maps $f$ such that $f^{\ast}P \simeq P$).

Considering the above, one would like to say things such as ``gauge theory is isometry invariant''. For example, it is said that ``Yang-Mills theory is Lorentz invariant''. But this is problematic in two respects. First, the relevant symmetry group is $Aut(P)$, not $Diff(M)$, so isometries are not even in a position to act on our fields. This can be fixed by considering all principal bundles instead of just $P$, or just the trivial ones.

The second problem is deeper: let $\phi \in Aut(P)$ map to an isometry of $M$ in the sequence above. Then, since our sequence does not split, we have an automorphism acting on our fields, which has a ``symmetry part'', but does not have a ``gauge part''. And it certainly could have a ``gauge part'', since $\mathscr{G}_{P}$ is included in $Aut(P)$. It appears that gauge equivalences and actual symmetries cannot, in general, be neatly separated. This is especially true if indiscriminate symmetry gauging is allowed (cf.~\cite[$\mathsection$2.8]{freedfield}).

It is therefore difficult to accept the claim, commonly made in the community \cite{seiberg}, that gauge equivalences are ``redundancies in the description'' or ``do-nothing transformations'', and that they have no physical significance. In the presence of gauge equivalences, without further constraints, one cannot simply pass to a reduced phase space. The Aharonov-Bohm effect and Dijkgraaf-Witten theory cannot be understood, indeed cannot exist, if we simply divide out the gauge equivalences. It is also clear that gauge transformations are not plain symmetries. Pushing a metric around by diffeomorphisms certainly does not alter the state of a system in any physically relevant way.

What are we to make of this situation? I tentatively propose the following definition, which includes all Yang-Mills theories, General Relativity, as well as a multitude of (limits of) string theories among ``theories with gauge equivalences''.
\begin{definition}\label{definitionofgaugesymmetry}
A classical theory with gauge equivalences is a theory whose space of states has an additional structure\footnote{This excludes the natural structure of $\omega$-groupoid that the space of states possesses in virtue of being a space.} of a $k$-groupoid.
\end{definition}
The case $k=0$ is trivial, requiring no additional structure, and so one should really speak of $k$-gauge theories, including non-examples as the degenerate case. For $k>1$ we allow weak groupoids. In the examples below we have $k=1$, but since the Kalb-Ramond field in string theory is a connection on a principal 2-bundle, one expects stringy examples with $k>1$.
\subsubsection*{Example 1: General Relativity}
The state space of general relativity is the groupoid of all Lorentzian manifolds and their isometries. We use the notion of ``state space'' loosely. We want to preserve the ability to couple the theory to other fields, and so we disregard the equations of motion. Of course, Einstein spacetimes form a subgroupoid.

\subsubsection*{Example 2: Yang-Mills Theory}
The state space of Yang-Mills theory, with structure group $G$ on a spacetime $M$, is the groupoid of $G$-principal bundles with $G$-connection and connection-preserving isomorphisms between them. This example is slightly ambiguous, since it is not clear whether to include all bundle morphisms or just the equivariant ones. This choice affects, for example, color charge conservation on topologically nontrivial spacetimes (as we see below).

\subsubsection*{Example 3: In General?}
Let $A$ be an algebra of observables with a group $G$ of symmetries acting on it. Then the groupoid of states is the action groupoid $\mathcal{S}(A)//G$ (also known as the weak quotient), where $\mathcal{S}(A)$ is the space of states, with the action of $G$ given by the fact that $\mathcal{S}$ is a functor. Unlike the previous examples $\mathcal{S}(A)$ includes mixed states, causing further complications.

\begin{remark}[The Problem of Emergence]
There are multiple contexts in which gauge symmetry is emergent. Definition \ref{definitionofgaugesymmetry} would then dictate the discontinuous change in $\dim X$, the categorical dimension of $X$, classifying any perturbation removing emergent gauge symmetry as a singular perturbation.
\end{remark}

\subsubsection*{Digression on The Geometry of Groupoids}
To really work with definition \ref{definitionofgaugesymmetry}, one must define the notion of a smooth map into $X$, which should be understood as a smooth family of objects and morphisms of $X$. This leads immediately to the notion of a stack, since stacks are a higher localization of internal categories (including groupoids). We will not make this precise here, but will merely assert that a groupoid with a localizable notion of morphism into it (from a space) automatically defines a stack, with the original groupoid being the global sections of that stack. An idea of how this works can be extracted from theorem \ref{localizationtheorem} and remark \ref{localizationremark}.

This allows us to speak of the geometry of $X$, in particular the sheaves on $X$, $Sh(X)$. This topos naturally contains all the geometric invariants of $X$ which can be defined as sheaves on the site of spaces. This means that objects like the differential forms, $\Omega^{\ast}_{X}$, are canonically defined. If infinitesimals are available, then other constructions, such as tangent vectors and vector fields can be defined. These can't be pulled back from the site of spaces. In such cases one can prove the usual relation $\Omega^{1}(X) = \Gamma(T^{\ast}X)$, which is not usually taken as a definition for stacks.

Most importantly gauge invariance is inherently baked in to the formalism. In examples 1 and 2 there are smooth action functionals
\begin{displaymath}
S : X \longrightarrow \mathbb{R},
\end{displaymath}
whose differentials $dS \in \Omega^{1}(X)$ are legitimate 1-forms. The stacks of solutions are the substacks $\{x \in X : dS(x) = 0 \}\subseteq X$. For the bare Einstein-Hilbert action the global sections of the solution stack are simply the groupoid of Einstein manifolds and their isometries.

This notion of solution is automatically gauge invariant since it is really a 2-pullback in a 2-category. The mystery of why ``imposing gauge invariance'' -- a colimit construction -- commutes with imposing the equations of motion -- a limit construction -- is resolved. Gauge invariance is encoded in the 2-cells of a 2-category, and the equations of motion are a 2-categorical construction.

This kind of stacky geometry, including measure theory, will be explored in depth in upcoming work \cite{stackygr}. The treatment of noncompact spacetimes requires delicate analysis.
\subsection{In Pursuit of Proper Language}
As I have already stressed, the ideas behind definition \ref{definitionofgaugesymmetry} are not new. I would like to build on the idea of that definition, and give gauge theories a distinguished structural place among all theories, and clarifying the notion of ``theory'' in general. I begin with the following distinction:
\begin{definition}\label{definitionofsymmetry}
Let $x \in X$ be a state in a classical theory with gauge equivalences. Then:
\begin{itemize}
\item The \emph{symmetries of $x$} are by definition the groupoid 
\begin{displaymath}
Aut_{X}(x) = X(x,x)
\end{displaymath}
of self-equivalences of $x$.
\item The \emph{gauge equivalences} are maps $x \rightarrow y$ in $X$.
\end{itemize}
\end{definition}
\subsubsection*{Example 1: General Relativity}
Diffeomorphisms $f : M \rightarrow M^{\prime}$ are exhibited among Lorentzian manifolds by maps of the form $(M, g) \rightarrow (M^{\prime}, f_{\ast}g)$. The automorphisms of $(M, g)$ are therefore exactly the isometries.

\subsubsection*{Example 2: Yang-Mills Theory}
Analogously to gravity, the gauge transformations act by pushforward, and are counted as gauge equivalences iff they change the connection. The isotropy group of a connection is not counted among the equivalences! It consists of genuine symmetries according to our definition.

This leads to our first real problem: one must decide if the bundle morphisms in this example are to be equivariant. The decision may be obvious, but consider the following question: do we want QCD to enjoy color charge conservation on topologically nontrivial manifolds? If so, then the equivariant maps are too little -- one must also include the right $G$-action as a symmetry, since this seems to be the only way to include ``constant gauge transformations'' as symmetries on general spacetimes. Without this global symmetry, features such as the Higgs mechanism would fail.

\subsubsection*{Example 3: In General?}
Here we come to the crux of the matter. We must confront the effects of definition \ref{definitionofsymmetry} on the notion of symmetry in ordinary theories (those with $k = 0$). They are quite curious: for non-gauge theories the definition dictates that time-invariant states, such as vacua, would have time translation symmetry, but that same ``symmetry'' would act as a mere gauge transformation on non-ivariant (e.g.~excited) states.

The symmetries of a lagrangian field theory \cite[$\mathsection$2.6]{freedfield} would likewise be classified as gauge, except at their fixed points. In this respect, either definition \ref{definitionofsymmetry} or the construction of this example is problematic. Perhaps this is just a linguistic deficiency, or a historical lack of appreciation for groupoids, as opposed to groups.

\begin{remark}\label{symmetrichochschild}
This discussion suggests that perturbative quantization should take the automorphism group into account. From a geometric perspective states with automorphisms are singular points in the space of states, in the sense that $G$-fixed points in some $G$-space $X$ are usually singular points in the quotient space $X/G$.

Since working with definition \ref{definitionofgaugesymmetry} amounts to replacing $X/G$ with the action groupoid (weak quotient), or more properly its stackification, the quotient stack $[X/G]$, this suggests that all points with nontrivial automorphisms should be considered singular in any groupoid.

The opportunity for special treatment of these states is clearly visible in the general formalism of deformation quantization -- singularities allow the appearance of nontrivial symmetric Hochschild cocycles. All isotropic (reducible) connections are such singularities, and the space of connections is full of them \cite{fuchs}.
\end{remark}

\subsubsection*{The Necessity of Higher Categories}
At this point, the reader would be right to protest in confusion. What prevents us from setting $x = y$ in definition \ref{definitionofsymmetry}, and completely confusing the supposed distinction? Insisting that $x \neq y$ in the second case is tenable, but goes against the philosophy of category theory. It seems that a decisive discussion of these matters requires the systematic use of higher-categorical formalism. Such a formalism is currently only available in the form of homotopy type theory \cite{hott}. There we find the general notion of an identity type, and a distinction between definitional equality and propositional equality. The Atiyah-Singer index theorem is an example of a propositional equality -- two differently constructed numbers are proven to be equal. By contrast renaming variables is an example of a definitional equality -- such equalities have no mathematical content, and their use in deductive reasoning is limited to bookkeeping. In this perspective gauge equivalences arise from propositional equalities, and symmetries from propositional equalities between definitionally equal states.

\subsection{Dependent Fields}
In constructing a field theory, the specification of spacetime and any additional structure on it (like orientation, spin-structure, etc.) is prior to the introduction of any other fields. This is so because the spacetime determines what fields can be introduced. Fields are dependent on spacetime. There can be multiple levels of dependency: spinor fields depend on the metric field and the orientation, which in turn depend on spacetime. In gauge theory this dependence is subtler: the sections of a bundle $\Gamma(P[V])$, associated to a principal bundle $P$, are the equivariant maps $P \rightarrow V$. These depend on $P$ for the specification of their domain.

So it seems that we must introduce a general notion of a dependent field -- a field definable only in the presence of other fields, and parametrically dependent upon them. Again, one can frame this using type-theoretic language: if fields are understood as types, then dependent fields are dependent types. The prime ``field'' would be spacetime. After that one can introduce general tensorial fields, such as the metric. After the introduction of a metric, and an orientation, spinor fields become available (giving a possibly empty space of fields, if there is no spin structure on spacetime). Dependently on spacetime, one can introduce principal bundle ``fields'' (thought of as maps to the classifying stack $BG$), then, dependently on those, connection fields and fields associated to representations of the structure group. This leads to the following definition.

\begin{definition}
Let $X$ the the groupoid of states of a classical theory with gauge equivalences. A dependent field for this theory is a functor $F : \Phi \rightarrow X$.
\end{definition}
Note the strange direction! It is critical to what follows.

\subsubsection*{Example 1: Scalar Fields}
Let $X$ be the groupoid of spacetimes and isometries. Then $C^{\infty} : X^{op} \rightarrow \mathbf{Set}$, which assigns to each spacetime $M$ its ring of smooth functions $C^{\infty}(M, \mathbb{R})$ determines a dependent field through the Grothendieck construction,
\begin{displaymath}
\Phi = \int C^{\infty},
\end{displaymath}
which means that $\Phi$ is the category of pairs $(M, \phi)$, with $\phi$ a scalar field on $M$, with the obvious projection to $X$. $F : \Phi \rightarrow X$ is simply forgets $\phi$.

A similar construction encompasses all ordinary natural fields, such as tensor fields.

\subsubsection*{Example 2: Bundles}
Extraordinary fields include bundles. Let $G$ be a group, with $BG$ its classifying stack, the dependent field
\begin{displaymath}
X \downarrow BG,
\end{displaymath}
is simply the category of principal $G$-bundles over spacetimes. The comma category is constructed from the forgetful functor $X \hookrightarrow \mathbf{Spc}$ into spaces and the single object inclusion $\{BG\} \hookrightarrow \mathbf{Spc}$. Here $F$ is again the projection $X \downarrow BG \rightarrow X$.

Bundles with additional structure, such as a connection, can easily be included here. This allows us to add spin structures to manifolds as ``fields'' and dependently on that, spinor fields, extending example 1.

\subsection{The Pathology of Dependent Symmetry}
In this picture symmetries also become dependent. The sequence
\begin{displaymath}
1 \longrightarrow \mathscr{G}_{P} \longrightarrow Aut(P) \longrightarrow Diff(M),
\end{displaymath}
displays the nontrivial dependence of gauge symmetries on diffeomorphisms. The fibers $Aut_{f}(P)$, for $f \in Diff(M)$, can be empty or not, and are glued together in a nontrivial manner, owing to the non-splitness of the sequence.

The fact that the sequence is not exact at $Diff(M)$ is precisely the statement that the dependent field $F : \mathbf{Spc} \downarrow BG \rightarrow \mathbf{Spc}$ mapping bundles to their underlying spaces is not full on automorphisms.

The lack of fullness has real consequences, and is typically considered pathological. The treatment of the energy-momentum tensor in \cite[$\mathsection$2.9]{freedfield} is plagued by it. The definition of ``weak diffeomorphism invariance'' given there looks very awkward, but is natural in our setting: it is exactly fullness of $F$ on the $D$-points, where $D$ is the object of first order infinitesimals. Such points are also called $Spec(\mathbb{R}[\varepsilon])$-points, a traditional paraphrase of the notion in algebraic geometry.

Flat connections are then the first order functorial splittings of dependent fields, and since they are not guaranteed to exist, arbitrary connections, which are ``functors not preserving composition'' are recommended in ibid.~instead. Regardless of this effort, the application of Noether's theorem is frustrated \cite[2.183-2.190]{freedfield}.

Noninfinitesimal splittings correspond to strictly equivariant groupoids, by the nonlinear Dold-Kan correspondence, or the ``layer-cake philosophy'' \cite{baezcohomology}. An example of this correspondence can be seen in our discussion of time reversal in section \ref{timereversal}. Such a situation should be called ``removable'' or trivial dependency. In such cases we can make contact with the formalism of section \ref{symmetries} simply by lifting all symmetries to the domain of the dependent field, treating it as a new state space, and applying proposition \ref{gaugeobservables} below.

\subsection{The Problem of Dynamics}
The problem of ``frozen time'' is well known in General Relativity. Here it strikes us in general form, in the guise of a question.
\begin{problem}
What is a morphism of gauge theories?
\end{problem}
It seems that functors have already been exhausted by dependent fields. In addition, spacetimes with time translation symmetry, the critical example of a dynamical process, are already completely internal to the space of states. It looks like there is simply no need, or even scope for morphisms $F : X \rightarrow Y$ implementing dynamical changes in the states of $X$.

Consider, however, the category $X^{\ast}$ Lorentzian manifolds with chosen timelike tangent vectors and their isomorphisms. The vector is an initial condition for a massive particle in spacetime. Then ``free fall for $t$ seconds'' \emph{does} define a functor
\begin{center}
\begin{tikzpicture}
\node (a) at (0,0) {$X^{\ast}$};
\node (b) at (-3,0) {$X^{\ast}$};
\node (c) at (-1.5,-1.5) {$X$};

\path[->] (b) edge node[auto] {$t$} (a)
		  (a) edge (c)
		  (b) edge (c);
\end{tikzpicture}
\end{center}
over the groupoid $X$ of spacetimes. It seems that physical processes can still occur between dependent fields.

\subsection{What's an Observable?}
Consider a functor $F : X \rightarrow Y$ on the space of states. We will think of it as a $Y$-valued observable.
\begin{proposition}\label{gaugeobservables}\hspace{1em}
\begin{enumerate}
\item If $x\rightarrow y \in X$, then $F(x) \rightarrow F(y) \in Y$.

\emph{This means that $F$ is gauge invariant.}
\item $F(x)$ is a representation of the symmetries of $x$.

\emph{So we retain group theory.}
\end{enumerate}
\end{proposition}

A detailed analysis of what this definition means for General Relativity will be presented in \cite{stackygr}.
\subsubsection*{Relation with the Traditional Treatment}
The contemporary treatment of observables for gauge theories deviates slightly from this idea of observable, with the observables on $X$ being given by the groupoid cohomology $H^{\ast}(X,E)$, where $E$ is a representation of $X$. Such an object is a smooth sheaf of vector spaces on $X$, which is just a morphism $X \rightarrow \mathbf{Vect}$, to the stack of vector spaces (traditionally called the classifying stack of vector bundles). 

In Yang-Mills theory one usually takes $X$ to be the stack of families of principal $G$-bundles (which is not the classifying stack $BG$ \cite{benini}), and $E = \mathbb{R}$, the trivial representation (i.e.~the representation induced from the constant sheaf $\mathbb{R}$ over the trivial groupoid).

The first group, $H^{0}(X,E)$, is the invariant sections of $E$ over $X$. Therefore the group $H^{0}(X, \mathbb{R})$ does consist of functors $X \rightarrow \mathbb{R}$, with $\mathbb{R}$ considered discrete. A variant of our idea is included in contemporary thinking. In any case, these are exactly $\mathbb{R}$-valued functions on the isomorphism classes of $X$, and so exactly what we would expect a scalar observable to be in both Yang-Mills theory and General Relativity.

The higher groups are more mysterious, encoding ghost fields (which are global in this perspective \cite{schenkel}). Their physical significance has been questioned, but at the very least they control possible gauge invariant couplings of the theory to other fields. If one imagines a ``space of all theories'', then the ghost fields would be crucial in determining the theory's ultimate location in that space.

In the functorial perspective ghosts become invisible, as they should be. Their role in the notion of observable amounts to describing maps $X \rightarrow B^{n}E$, where $B$ is the delooping functor. Applying the Yoneda lemma, we see that such maps are part of the ``observable functor'' $Hom(X,-)$, necessary in reconstructing $X$ from the structure of its observables.

The physical significance of delooping, while obscure, can be illuminated somewhat. The ghost observables $H^{2}(X,E)$ control the extensions of $X$ by $E$, and hence, in our terminology, (some class of) dependent fields. Thus, our preceding remarks were correct -- the role of ghost fields, at least in part, is to control the possible couplings between the base theory and its dependent fields\footnote{A similar understanding of ghost fields was also communicated to me by Alexander Schenkel.}.

\subsection{Quantization}\label{gaugequantization}
Let $F \dashv U$ be an adjoint equivalence of categories
\begin{center}
\begin{tikzpicture}
\node (b) at (6,0) {$\mathcal{D}^{op}$};
\node (a) at (0,0) {$\mathcal{C}$};

\path[->] (a.north east) edge[bend left = 10] node[auto] {$F$} (b.north west)
		  (b.south west) edge[bend left = 10] node[auto] {$U$} (a.south east);
\end{tikzpicture}
\end{center}
The reader should think of it as a space-algebra duality, such as Gelfand duality, or the duality between affine schemes and commutative rings. We will treat $\mathcal{C}$ as spaces, and $\mathcal{D}$ as algebras.

Equivalences of categories preserve all limits and colimits, and the theories of categories and groupoids are finite limit theories (typed equational theories). This means that the constructions $\mathcal{C} \mapsto \mathbf{Cat}(\mathcal{C})$ and $\mathcal{C} \mapsto \mathbf{Grpd}(\mathcal{C})$, of internal categories and groupoids, are functorial with respect to finite limit preserving functors. In particular $F \dashv U$ induces equivalences $\mathbf{Cat}(F) \dashv \mathbf{Cat}(U)$ and $\mathbf{Grpd}(F) \dashv \mathbf{Grpd}(U)$ bewteen internal structures in $\mathcal{C}$ and $\mathcal{D}$.
\begin{center}
\begin{tikzpicture}
\node (b) at (6,0) {$\mathbf{Cat}(\mathcal{D}^{op})$};
\node (a) at (0,0) {$\mathbf{Cat}(\mathcal{C})$};
\node (y) at (6,-3) {$\mathbf{Grpd}(\mathcal{D}^{op})$};
\node (x) at (0,-3) {$\mathbf{Grpd}(\mathcal{C})$};

\path[->] (a) edge[bend left = 10] node[auto] {$\mathbf{Cat}(F)$} (b)
		  (b) edge[bend left = 10] node[auto] {$\mathbf{Cat}(U)$} (a)
		  (x) edge[bend left = 10] node[auto] {$\mathbf{Grpd}(F)$} (y)
		  (y) edge[bend left = 10] node[auto] {$\mathbf{Grpd}(U)$} (x);
\end{tikzpicture}
\end{center}
Writing out the diagrammatic definitions of categories and groupoids, we see that they correspond to coalgebroids and Hopf algebroids, respectively. Note that the arrow reversal also applies to morphisms -- the internal functors -- and also to internal natural transformations.

Since we have declared that the state spaces of gauge theories are essentially (higher) internal groupoids, and we are asking for a theory of deformation quantization of such structures, we are naturally led to consider deformations of Hopf algebroids. Such a theory has been formulated (e.g.~\cite{xu}), but its relation to the traditional BV-BRST approach remains to be understood.
\begin{problem}
Are Hopf algebroid deformations equivalent to the BV-BRST formalism?
\end{problem}
If we want to quantize stacks of states, that is take into account the geometry of a given groupoid, we must be mindful of theorem \ref{localizationtheorem}. The proper structures are a 2-localization of the category of Hopf algebroids, by the class of internal morphisms dual to the local equivalences. Localization can have drastic effects on how things look: a group $G$, seen as a one object groupoid \emph{is} the category of $G$-torsors when seen as a stack \cite{bunge}.
\begin{problem}
What's a noncommutative stack?
\end{problem}
In other words we are interested in computing the localization of Hopf algebroids by morphisms dual to the local equivalences in purely algebraic terms, and subsequently allowing everything to be noncommutative.

\section{Chasing The Moduli Theory of Vacua}\label{appendixb}
In this appendix we work abstractly, over some base topos $E$ of ``spaces'', with an \emph{ordered} ring $R$. $GNS$ is then a stack morphism over $E$, as sketched in section \ref{sins}. We write $\mathcal{O}_{X}$ for the pullback of $R$ along the geometric morphism $E/X \rightarrow E/1 = E$ induced by the unique map $X \rightarrow 1$ in $E$.
\subsection{The Stack of Vacua}
It appears that most, if not all ``path integral arguments'' and ``duality theorems'' are at their cores simply isomorphisms of vacuum states of certain theories. Thus we wish to study the notion of a vacuum. For this reason one of the central, long term aims of the studying the category $\mathbf{Phys}$, is the construction of the stack of vacua
\begin{displaymath}
\mathbf{Vac} \longrightarrow \mathbf{Phys}^{\mathbf{Time}}.
\end{displaymath}
Here $\mathbf{Time}$ is the groupoid of homogeneous time. Vacua are to be understood as ``minimal energy states''. This is deliberately ambiguous, due to the problems below.

The wording assumes that every time evolution has a Hamiltonian, to be able to define ``energy''. More importantly, \textbf{the notion of vacuum state is predicated upon the notion of time evolution}. In theories of emergent spacetime the concept of vacuum seems ambiguous. Mere stability -- as in the string landscape -- seems insufficient, as there are theories with time-invariant states which are not vacua, and yet there is no \emph{global} generator of time evolution whose expectation value we could wish to minimize. Clearly, there is conceptual work to be done here.

The category $\mathbf{Vac}$ should be thought of as ``the space of all vacua'', with the projection
\begin{displaymath}
\pi : \mathbf{Vac} \longrightarrow (\ast\mathbf{Alg}^{op})^{\mathbf{Time}} \longrightarrow \ast\mathbf{Alg}^{op}
\end{displaymath}
giving us the observables (at any time, every time, or some specific time -- this usually doesn't matter) to which a given vacuum belongs. This structure does not, at least ``morally'', contain the string landscape, as explained in the introduction.

\begin{remark}
We have not formally required the purity of our states. This has the effect that $\mathbf{Vac}$ will include classical mixed states of minimal energy, which are not usually considered vacua. For an illuminating discussion see \cite[$\mathsection$1.1]{witten}. The technical problems discussed below make this objection temporarily moot.
\end{remark}

The study of $\mathbf{Vac}$ is the study of how quantum vacua behave in families. The following issues stand out as extremely interesting.
\subsection*{Is $\mathbf{Vac}$ a Stack?}
This might seem obvious. But there are caveats, which I believe should influence the form of the definition. First, we assumed the existence of a Hamiltonian pointwise, that is for states over the point in the category of spaces. As every time evolution is a homology class of Hamiltonians, the existence of a Hamiltonian for an entire family is a homological problem. We may -- and will -- simply demand specifying a solution in the definition.

The other problem is more concerning. The notion of state we have been using until now appears to be too generous for $\mathbf{Vac}$ to form a stack. It arises -- again -- from the nonuniqueness of the Hamiltonian. The following situation may arise: there could be a family of candidate vacua $\varphi : X \rightarrow \mathbf{Phys}^{\mathbf{Time}}$, and a central observable $f \in \pi(\varphi)$, whose expectation value $\langle f \rangle \in \mathcal{O}_{X}$ changes sign, as a function on $X$ ($f$ is a section of $\mathcal{O}_{X}$ which is just a map $X \rightarrow R$ in $E$).

In this situation the very notion of ``vacuum'' does not make sense. Since we can add $f$ to any Hamiltonian for this family, the possible values of energy $\langle H + af \rangle = \langle H \rangle + a\langle f \rangle$, for $a \geq 0$, are not linearly ordered in $\mathcal{O}_{X}$, and minimality does not make sense. Consequently, defining $\mathbf{Vac}$ by demanding minimal energy for every generalized element, or ``generalized vacuum'', $\varphi$ does not seem to make sense.

More abstractly, we can explain this by noting that $\mathcal{O}_{X}$ will in general be only \emph{partially} ordered, even if $R$ is linearly ordered. Energy will always carry a free $\mathcal{O}_{X}$-action, since $\mathcal{O}_{X}$ sits in the center of every $\ast$-algebra over $X$, and energy is always a torsor over the center of the algebra. Such torsors are also partially ordered, but, unlike for linearly ordered $R$, minimality in partially ordered torsors can always be ruined by the phenomenon described above.

There are several ways to proceed. One would be to demand something similar to ``$\neg\neg$minimality'' internally to $E$, in the sense of asserting that it's not true that there are lower energy states than the one under consideration. Another would be pointwise (or maybe even stalkwise) minimality. One imagines sections of $\mathbf{Phys}$ which are vacuum states pointwise. For the time being, we leave this issue unresolved.

\subsection*{Is $\pi$ a Representable Morphism of Stacks?}
Here we assume that $\mathbf{Vac}$ is a stack. Consider a space $X \in E$ and any pullback square
\begin{center}
\begin{tikzpicture}
\matrix (m) [matrix of math nodes, nodes in empty cells, row sep = 1cm, column sep = 1cm, text height = 1.5ex, text depth = .25ex]{
Y & \mathbf{Vac} \\
X & \ast \mathbf{Alg}^{op} \\
};

\path[->] (m-2-1) edge (m-2-2)
		  (m-1-2) edge node[auto] {$\pi$} (m-2-2)
		  (m-1-1) edge (m-1-2)
		  (m-1-1) edge (m-2-1);
		  
\end{tikzpicture}
\end{center}
in the category of stacks over $E$. Representability means that $Y$ is actually a \emph{space} (i.e.~an object of $E$), and not some arbitrary stack. Thus for any family of algebras of observables parameterized by a space, there is only a space of vacua above them. Otherwise there would be a category of vacua, with physical processes between them. One would expect this to happen only with ``false'' vacua, and not the real ones.

\subsection*{What Geometric Properties does $\pi$ have?}
This includes the paradigmatic geometric questions one may ask of any map. For $\pi$ such questions encode pressing physical problems. For example:
\begin{itemize}
\item Existence of vacua: for what maps does $\pi$ have the lifting property? That is, for which diagrams below can the dashed arrow be found?
\begin{center}
\begin{tikzpicture}
\matrix (m) [matrix of math nodes, nodes in empty cells, row sep = 1cm, column sep = 1cm, text height = 1.5ex, text depth = .25ex]{
& \mathbf{Vac} \\
X & \ast \mathbf{Alg}^{op} \\
};

\path[->] (m-2-1) edge (m-2-2)
		  (m-1-2) edge node[auto] {$\pi$} (m-2-2)
		  (m-2-1) edge[dashed] (m-1-2);
\end{tikzpicture}
\end{center}

Lifting over the point means existence of a vacuum state. More general lifting properties mean the existence of families of vacua. One is clearly tempted to study the homotopy theoretical properties of $\pi$.
\item Uniqueness of the vacuum: where is $\pi$ an isomorphism, locally on $\ast\mathbf{Alg}^{op}$?
\item Isolated vacua: where is $\pi$ an isomorphism, locally on $\mathbf{Vac}$?
\item Discreteness of vacua: where is $\pi$ a covering projection?
\item Existence for first order deformations: where is $\pi$ a submersion?
\item Uniqueness for first order deformations: where is $\pi$ an immersion?
\item Existence of families of moduli spaces of vacua: where is $\pi$ flat?
\item Locally universal families: where is $\pi$ locally trivial?
\end{itemize}
All of them are extremely useful (and many have been assumed!) in path integral-type arguments. For example we have:
\begin{conjecture}[Witten's Theorem]\label{wittenstheorem}
The $\hbar$-family of vacua of 4d $N=2$ super Yang-Mills theory is trivial, with fiber $\mathbb{C}$.
\end{conjecture}
This result is the starting point for Witten's reformulation of Donaldson theory \cite[Lectures 17-19]{witten}.

At points where $\pi$ lacks most of the good properties listed above, the vacua run amok. To control this chaos one must also investigate the singular behavior of $\pi$.
\begin{problem}
What kinds of singularities does $\pi$ have? Where can they occur?
\end{problem}
Since catastrophes in the sense of Thom can actually happen in physics, one expects the answer to be ``all of them, essentially everywhere''. For example, branch points represent bifurcations of vacua under a variation of parameters, something that can happen even in the classical limit \cite[$\mathsection$1.1]{witten}. This second question is thus equally important -- ``where be dragons?'', so to speak. Must we retreat to the holomorphic heaven of supersymmetry, or is there life in the hills and valleys of broken symmetry?

\subsection{Speculation on the Nature of the Path Integral}
Let $F$ be some observable, and write
\begin{displaymath}
\langle F \rangle = \int F e^{\frac{i}{\hbar}S}\, D\Phi
\end{displaymath}
for its vacuum expectation value in path integral form. Compute formally
\begin{equation}\label{hconnection}
\frac{\partial}{\partial\hbar} \langle F \rangle = -\frac{i}{\hbar^{2}} \langle F \rangle
\end{equation}
Now consider the question:
\begin{quote}
Why does the quantization of classical systems typically depend only on these systems?
\end{quote}
You may think that the implicit claim is outrageous, and patently false, but I don't observe physicists in the wild arguing about quantization ambiguities. In our best theorem on deformation quantization \cite{kontsevich} the result is also essentially unique.

Why is the Poisson structure -- a tangent vector in the space of algebras of observables -- sufficient to determine the observables for $\hbar \approx 10^{{-34}}$? That's a small number, but not infinitesimal. Why are there no ``$\hbar$-phase transitions'' in which the observables radically change their nature?

Naively, one would have to expect a ``quantization vector field'', which controls changes in $\hbar$ at positive values, in addition to the Poisson structure, which controls things at $\hbar = 0$. This is what equation \ref{hconnection} provides.
\begin{conjecture}
The path integral is a connection, in the sense of differential geometry, on a $\hbar$-family of vacua. Its content for states in the vacuum sector is summarized by equation \ref{hconnection}.
\end{conjecture}
What about the observables? We know a priori how to differentiate those, and the results should coincide.
\begin{conjecture}
The $\hbar$-derivative of the operator product expansion defines an associative deformation of the OPE algebra.
\end{conjecture}

\paragraph{}
At face value, this would contradict the common expectation of extended field theory, that the path integral is essentially tied to locality and gluing conditions. If spacetime is emergent in any capacity, then either this is not true, or emergent spacetime exceeds the expressive capacity of quantum theory\footnote{I do not grant claims of emergence unless a significant portion of General Relativity emerges as well, dynamically, with a range of geometries, time included. This is because GR is part of our concept of spacetime.}. Looking at string theory, I find my self leaning toward the latter.

Regardless of that, in the perspective developed here the 1-dimensional gluing law should only be expected in cases with specified time evolution. Pursuing this analogy to higher dimensions leads to considering functors
\begin{displaymath}
M \longrightarrow \mathbf{Vac},
\end{displaymath}
where $M$ is a bordism, playing the role of spacetime, generalizing time evolution
\begin{displaymath}
\mathbb{R} \longrightarrow \mathbf{Vac}.
\end{displaymath}
More broadly we can consider functors into $\mathbf{Phys}$.

To make this analogy precise, and to make contact with the formalism of extended local field theory, we would need to investigate functors
\begin{displaymath}
n\mathbf{Bord}(M) \longrightarrow n\mathbf{Phys},
\end{displaymath}
from bordisms in a spacetime $M$ to some $n$-categorical version of $\mathbf{Phys}$. Our discussion of gauge theories in appendix \ref{appendixa} certainly suggests that $\mathbf{Phys}$ should be a higher category. An $n$-categorical GNS construction would then provide a link from such functors to ordinary extended local field theories, defined as representations of structured bordism categories on ``$n$-Hilbert spaces'', whatever they turn out to be.

The higher category typically expected to take center stage is
\begin{displaymath}
\mathbf{SymMonCat}(n\mathbf{Bord}, n\mathbf{Phys}),
\end{displaymath}
which I would interpret as the category of all ``universal'', spacetime independent extended local theories. This seems to be the only way of explaining the otherwise bizarre constructions of \cite{freedquant}. More confusingly, one may attempt to make sense of theories over the point.

This discussion also resonates with the idea of  ``generalized physical theories'' defined as objects of the slice 2-category $\mathbf{SymMonCat}/\mathbf{Phys}$. There is a much larger and conceptually sensible framework to be discovered here. In particular, the distinction between what we call a state in this paper and the notion of a whole theory is not clear.
\end{document}